\documentclass[onecolumn,nofootinbib, superscriptaddress,10pt,aps,prx]{revtex4-2}
\usepackage[dvips]{graphicx} 
\usepackage{amsfonts}
\usepackage{csquotes} 
\usepackage{amssymb}
\usepackage{amscd}
\usepackage{amsmath}    
\usepackage{amsthm}
\usepackage{bm} 
\usepackage{bbm}
\usepackage{booktabs}
\usepackage{array}
\usepackage{enumerate}
\usepackage{enumitem}
\usepackage{epsfig}
\usepackage{subfigure}
\usepackage{subfloat}
\usepackage{xcolor}	
\usepackage{physics}
\usepackage[most]{tcolorbox}

\usepackage{tikz}
\usepackage{quantikz}
\usepackage{multirow}

\usepackage{tabularx}
\usepackage{float}
\usepackage{appendix}
\usepackage{makecell}
\usepackage{algpseudocode}  
\usepackage{dsfont}
\usepackage[colorlinks, linkcolor=blue, anchorcolor=blue, citecolor=blue]{hyperref}
\usepackage{MnSymbol}
\usepackage[linesnumbered,ruled,vlined]{algorithm2e}
\usepackage{natbib}
\usepackage{times}
\usepackage{tabularray}

\makeatletter
\newcommand{\printappendixtoc}{%
	\section*{Contents}    
	\@starttoc{atoc}       
}

\newcommand{\l@atocsection}{\@dottedtocline{1}{0em}{2.3em}}
\newcommand{\l@atocsubsection}{\@dottedtocline{2}{1.5em}{3em}}
\makeatother

\newtheorem{theorem}{Theorem}
\newtheorem{fact}{Fact}
\newtheorem{lemma}{Lemma}
\newtheorem{corollary}{Corollary}

\newtheorem{definition}{Definition}

\newtheorem*{remark}{Remark}

\newcommand{\bF}{\mathbb{F}}

\newcommand{\bR}{\mathbb{R}}

\newcommand{\bu}{\bm{u}}
\newcommand{\btheta}{\bm{\theta}}

\newcommand{\cD}{\mathcal{D}}
\newcommand{\cE}{\mathcal{E}}
\newcommand{\cH}{\mathcal{H}}

\newcommand{\cO}{\mathcal{O}}
\newcommand{\cT}{\mathcal{T}}
\newcommand{\cU}{\mathcal{U}}
\newcommand{\wt}{\mathrm{wt}}
\newcommand{\supp}{\mathrm{supp}}
\newcommand{\Bin}{\mathrm{Bin}}

\newcommand{\GR}{\mathrm{GR}}

\newcommand{\TVD}[2]{\mathrm{TV}\left(#1, #2\right)}

\definecolor{dred}{rgb}{.8,0.2,.2}
\definecolor{dblue}{rgb}{.2,0.2,.8}

\newtcolorbox[auto counter]{mybox}[2][]{
	enhanced,
	breakable,
	colback=blue!5!white,
	colframe=blue!75!black,
	fonttitle=\bfseries,
	title=Box \thetcbcounter: #2,#1
}

\begin{document}
	\title{Complexity-driven transitions in quantum observation}
	\author{Zhenyu Du}
	\thanks{These authors contributed equally to this work.}
	\affiliation{Center for Quantum Information, Institute for Interdisciplinary Information Sciences, Tsinghua University, Beijing 100084, China}
	
	\author{Siyuan Cheng}
	\thanks{These authors contributed equally to this work.}
	\affiliation{Center for Quantum Information, Institute for Interdisciplinary Information Sciences, Tsinghua University, Beijing 100084, China}
	
	\author{Han Ye}
	\affiliation{Center for Quantum Information, Institute for Interdisciplinary Information Sciences, Tsinghua University, Beijing 100084, China}
	
	\author{Junjie Chen}
	\affiliation{Center for Quantum Information, Institute for Interdisciplinary Information Sciences, Tsinghua University, Beijing 100084, China}

	\author{Xiao Yuan}
	\email{xiaoyuan@pku.edu.cn}
	\affiliation{Center on Frontiers of Computing Studies, Peking University, Beijing 100871, China}
	\affiliation{School of Computer Science, Peking University, Beijing 100871, China}
	
	\author{Xiongfeng Ma}
	\email{xma@tsinghua.edu.cn}
	\affiliation{Center for Quantum Information, Institute for Interdisciplinary Information Sciences, Tsinghua University, Beijing 100084, China}

	\begin{abstract}
		
		Observing the physical world is a foundational pursuit in science. In the quantum realm, however, observation necessitates a fundamental quantum-to-classical conversion: destructive measurements irreversibly project quantum states into classical data, inevitably incurring a loss of information. What physical principles govern this information loss, and how can we construct optimal measurements to maximize the readout? Here, we address these questions by establishing an intrinsic relationship between readout capability—quantified by the ratio of accessible classical Fisher information to the total quantum Fisher information (QFI), and measurement complexity—defined as the quantum circuit depth required prior to projection.
		Remarkably, we uncover a sudden emergence of observability: a sharp hidden-to-visible transition driven entirely by measurement complexity. We rigorously prove that below critical depth thresholds—$\Theta((\log n)^{1/\delta})$ for $\delta$-dimensional architectures and $\Theta(\log\log n)$ for all-to-all connectivity—readout capability decays exponentially with system size $n$, rendering the quantum information fundamentally inaccessible.
		Surprisingly, immediately above this threshold, the system enters a visible regime: we demonstrate that randomized measurements universally recover a constant fraction of the QFI using approximate unitary 3-designs, for which we explicitly develop optimal-depth circuit constructions tailored to finite-dimensional architectures. By unveiling the fundamental scaling laws and transitions that govern quantum observation, our results delineate definitive resource boundaries for quantum learning, state certification, and quantum metrology.
		
	\end{abstract}
	
	\maketitle
	
	\section{Introduction}
	
	Observing the physical world with ultimate precision is a foundational pursuit that drives the frontier of modern science, enabling transformative breakthroughs ranging from gravitational wave detection~\cite{Caves1981QuantumMechanicalNoise, Tse2019QuantumEnhancedLIGO, Acernese2019IncreasingAstrophysical} and high-resolution imaging~\cite{Tsang2016Superresolution, Boto2000BeatDiffractionLimit} to precision metrology~\cite{Giovannetti2004BeatingStandardLimit, Giovannetti2011AdvancesMetrology, PedrozoPenafiel2020AtomicClock}. 
	At the heart of these diverse applications lies a unified objective: extracting the rich physical information embedded within quantum states, where the ultimate observational limit is rigorously governed by the quantum Fisher information (QFI)~\cite{Helstrom1969EstimationTheory, Braunstein1994GeometryStates}. 
	However, because we reside in a macroscopic classical world, extracting this intrinsic information necessitates a fundamental quantum-to-classical conversion. 
	Measurement serves as the interface for this extraction, irreversibly projecting quantum states into classical data. 
	This destructive process inevitably risks a loss of information: the surviving classical signal, quantified by the classical Fisher information (CFI) of the outcome distribution~\cite{Fisher1925TheoreyStatisticalEstimation, Cramer1946MathematicalMethods}, is intrinsically upper-bounded by the QFI. 
	Consequently, suboptimal measurement strategies can leave a vast amount of quantum information fundamentally inaccessible, severely compromising the final observation.
	Identifying measurement schemes that efficiently translate quantum potential into classical reality (QFI into CFI) has thus emerged as a central quest in quantum science~\cite{Braunstein1994GeometryStates, Pezze2008MachZehnderInterferometry, Szczykulska2016Multiparameter, Pezze2017SimulaneousQuantumEstimation, Zhu2018FisherSymmetric, Zhou2026RandomizedMetrology, Lu2026FisherRandomMeasurements}.
	
	However, fulfilling this quest is fundamentally bottlenecked by an inherent tradeoff between readout capability and measurement complexity. 
	While optimal extraction approaching the QFI limit typically demands highly nonlocal operations via deep quantum circuits---rendering them practically intractable~\cite{Szczykulska2016Multiparameter, Zhu2018FisherSymmetric, Zhou2020SaturatingLOCC, Zhou2026RandomizedMetrology, Lu2026FisherRandomMeasurements, Chen2022InformationGeometryHierarchical}---realistic quantum devices are physically constrained by finite coherence times to operate at shallow depths. 
	Yet, such restricted measurements often impose a severe observational bottleneck, leaving the intrinsic quantum information effectively inaccessible~\cite{Bennett1999NonlocalityWithoutEntanglement, Terhal2001HidingBits, Eggeling2002HidingClassicalData, Tsang2011NonlocalityInterferometry, Chen2022Memory, Huang2022QuantumAdvantage, Liu2025QuantumLearningPhotonic}. 
	Driven by this stringent physical reality, extensive efforts have sought to bridge the gap using simplified schemes, yielding notable successes in low-depth classical shadow estimation~\cite{Huang2020Predicting, Schuster2024RandomUnitaries, Hu2025DemonstrationShallowShadow, Hu2023ShallowShadow}, local state certification~\cite{Huang2024Certifying, Gupta2025SingleQubitCertification, Du2025Localizable, Coladangelo2026TwoBases}, and distributed metrology~\cite{Pezze2025DistributedOptimalLocal}. 
	Nevertheless, these advances remain largely task-specific, state-dependent, or reliant on demanding adaptive operations. 
	A universal theoretical framework is thus still missing, leaving a profound question unanswered: what fundamental scaling laws govern the ultimate relationship between readout capability and measurement complexity?
	
	\begin{figure*}[!htbp]
		\centering
		\includegraphics[width=0.95\linewidth]{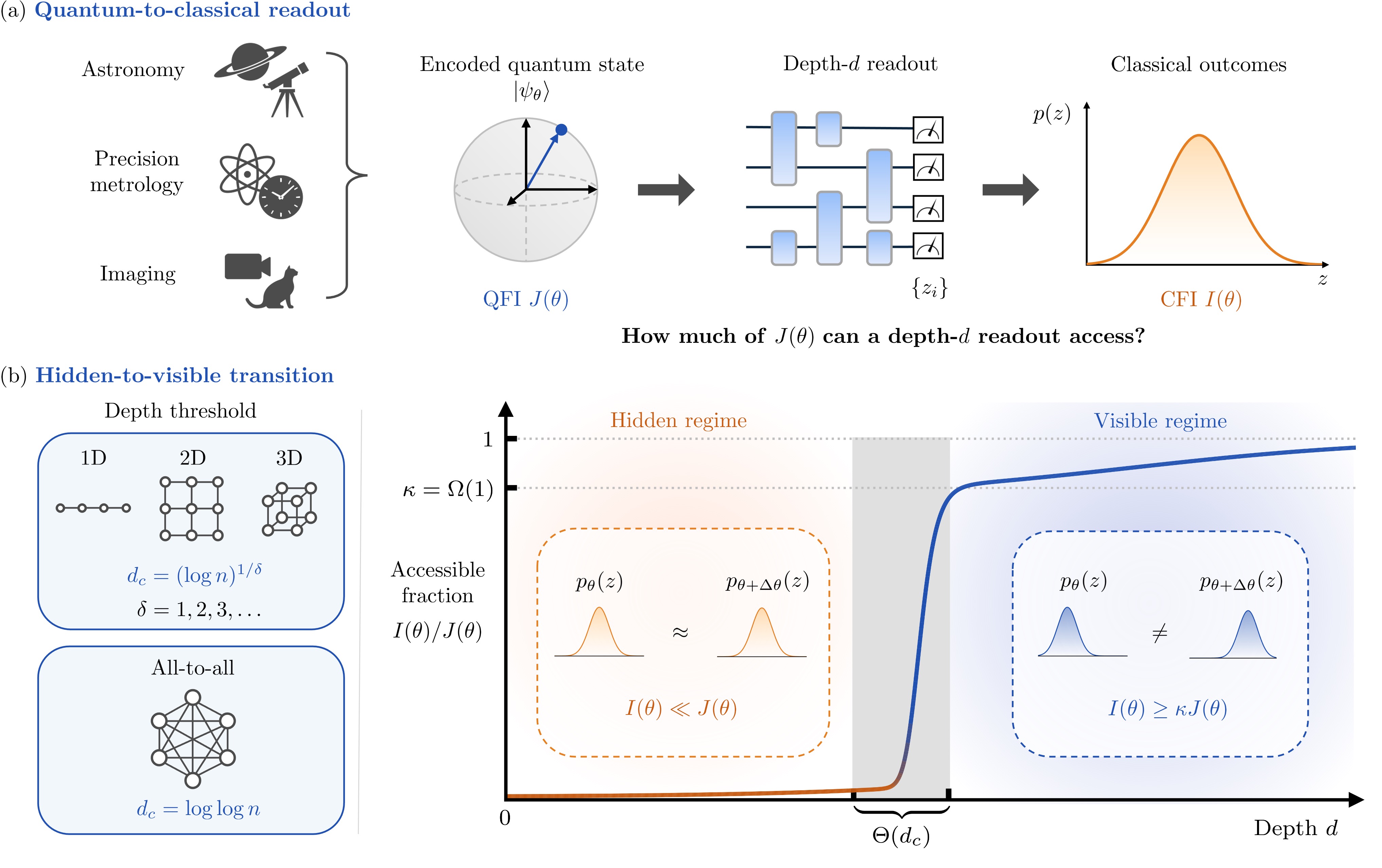}
		
		\caption{The quantum observation task and the complexity-driven transitions.
			(a) The quantum-to-classical conversion.  
			To access the encoded information, the $n$-qubit pure state must be converted into classical outcomes via measurements.
			The quantum state carries quantum Fisher information (QFI) $J(\theta)$ about a parameter $\theta$. 
			After applying a depth-$d$ circuit and measuring in the computational basis, the resulting outcome distribution carries classical Fisher information (CFI) $I(\theta)$, which quantifies the information accessible after measurement.
			(b) The hidden-to-visible transition. Our main result reveals a sharp, complexity-driven transition in readout capability---quantified by the accessible information fraction $I(\theta)/J(\theta)$---dictated by measurement circuit depth.
			Below this critical threshold lies the hidden regime, where there exists a pure-state encoding for which all allowed depth-$d$ readouts produce outcome distributions that are nearly independent of $\theta$, yielding $I(\theta)\ll J(\theta)$. 
			Above the matching threshold lies the visible regime, where we construct depth-optimal measurements that extract a constant fraction $\kappa = \Omega(1)$ of the QFI for every pure-state encoding and every parameter point, ensuring $I(\theta)\ge \kappa J(\theta)$. 
			We establish these circuit depth thresholds at $d_c=(\log n)^{1/\delta}$ for $\delta$D architectures and $d_c=\log\log n$ for all-to-all architectures. 
		}
		\label{fig:summary}
	\end{figure*}
	
	In this work, we address these fundamental questions by uncovering a sharp, hidden-to-visible transition in quantum readout capability, driven entirely by measurement circuit depth (Fig.~\ref{fig:summary}). We establish precise critical thresholds---$\Theta((\log n)^{1/\delta})$ for $\delta$-dimensional architectures and $\Theta(\log \log n)$ for all-to-all connectivity---below which the system resides in a fundamentally unobservable regime. 
	We prove that any measurement below these depth bounds is strictly constrained by ``phase-hiding'' and ``data-hiding'' mechanisms: rather than simply converting the quantum signal, low-depth measurements can irreversibly erase the information, yielding an exponentially vanishing classical signal.
	This rigorous no-go result applies even when the shallow measurement circuit is globally optimized with full knowledge of the encoding, demonstrating that quantum information remains strictly inaccessible without sufficient measurement complexity.
	Consequently, we establish, for the first time, tight circuit-depth lower bounds on essential tasks such as parameter estimation, state certification, and fidelity estimation.
	
	Surprisingly, immediately above this depth scale, the system enters a visible regime. We prove that randomized measurements induced by approximate unitary 3-designs bypass these hiding mechanisms, universally extracting a constant fraction of the full quantum Fisher information matrix (QFIM) for any multiparameter pure-state encoding. To realize this, we develop depth-optimal circuit constructions for these approximate designs across various hardware connectivities, utilizing a novel implementation of exact unitary 2-designs on $\delta$-dimensional architectures as a key ingredient. Crucially, this advance provides the first realization of multiplicative-error approximate unitary designs that scale optimally with system size $n$ in any architecture dimension $\delta \ge 2$, closing a prominent gap in the literature~\cite{Schuster2024RandomUnitaries}.
	
	Ultimately, this hidden-to-visible transition governs the readout of both continuous parameters and discrete quantum information. By unveiling the fundamental depth requirements for quantum-to-classical information readout, our results delineate the definitive resource boundaries for essential tasks across quantum metrology, learning, and certification.

	\section{Main Results}
	We begin by formalizing the physical process of quantum observation as a quantum information readout task. Consider the estimation of a real parameter $\theta$ encoded within a smooth family of $n$-qubit pure states, $ \cE=\{\ket{\psi_\theta}:\theta\in\mathbb{R}\}$.
	The intrinsic information about $\theta$ carried by the quantum state is quantified by the QFI,
	\begin{equation}
		J_{\cE}(\theta) = 4\left(\braket{\dot\psi_\theta}{\dot\psi_\theta} - \left|\braket{\psi_\theta}{\dot\psi_\theta}\right|^2 \right),
		\qquad
		\ket{\dot\psi_\theta}:=\partial_\theta\ket{\psi_\theta}.
	\end{equation}
	To access this information, a measurement $M=\{M_x\}$ converts this quantum state into a classical outcome distribution $p_\theta(x) = \tr\left(M_x\ketbra{\psi_\theta}{\psi_\theta}\right)$.
	The surviving classical signal is quantified by the CFI
	\begin{equation}
		I_{\cE}^{M}(\theta) = \sum_{x:p_\theta(x)>0} \frac{\left(\partial_\theta p_\theta(x)\right)^2}{p_\theta(x)} .
	\end{equation}
	Given multiple independent experimental runs, $J_{\cE}(\theta)$ and $I_{\cE}^{M}(\theta)$ dictate the best estimation precision before and after measurement, respectively, via the quantum and classical Cram\'er--Rao bounds~\cite{Helstrom1969EstimationTheory, Braunstein1994GeometryStates}.
	This framework naturally extends to multiparameter encoding. For a state encoding a parameter vector $\btheta\in\mathbb R^m$, the information is quantified by the $m \times m$ QFIM and CFIM, $J_{\cE}(\btheta)$ and $I_{\cE}^{M}(\btheta)$ (formal definitions are given in Methods).
	
	The CFI is upper-bounded by the QFI ($I_{\cE}^{M}(\theta)\le J_{\cE}(\theta)$, and $I_{\cE}^{M}(\btheta) \preceq J_{\cE}(\btheta)$ for multiparameter encodings), and measurements with limited complexity can fail to approach this theoretical limit, resulting in a severe loss of information.
	To investigate the fundamental relationship between this information loss and measurement complexity, we restrict the allowed measurements to depth-$d$ unitary-basis readouts on a specified architecture $G$. 
	Such a readout first applies a depth-$d$ circuit $U$, constrained by the connectivity of $G$, to the input state $\rho$. 
	This is followed by a computational-basis measurement of all data qubits, yielding the outcome distribution $p_\rho^U(z)=\bra z U\rho U^\dagger\ket z$ for $z\in\{0,1\}^n$.
	We also allow randomized readouts, where the circuit $U$ applied in each experimental run is sampled from an ensemble $\mathcal{U}$.
	While clean ancillas may assist in implementing the data-qubit unitaries when explicitly stated, the final computational-basis measurements are always restricted to the data qubits.
	The precise readout model is defined in Methods.

	The metric of interest is the readout capability of a depth-$d$ circuit family: the guaranteed fraction $\kappa$ of the QFI that an optimal measurement within this class can convert into CFI for arbitrary encoding.
	Concretely, for a given encoding $\cE$ and parameter point $\theta$, we ask whether one can always choose a depth-$d$ readout $M$ such that $I_{\cE}^{M}(\theta) \ge \kappa J_{\cE}(\theta)$. 
	Our main result establishes the critical depth threshold at which this readout capability transitions from being exponentially vanishing to order one across $\delta$D and all-to-all architectures (Fig.~\ref{fig:summary}).

	\begin{theorem}[Complexity-driven transition in quantum information readout]
		\label{thm:depth-optimal-readout}
		For sufficiently large $n$, unitary-basis readouts exhibit the following transition:
		\begin{center}
			\setlength{\tabcolsep}{12pt}
			\begin{tabular}{cccc}
				\toprule
				Architecture 
				& \begin{tabular}[c]{@{}c@{}}Hidden regime \\ $\kappa \le \exp\left[-n^{\Omega(1)}\right]$\end{tabular}
				& \begin{tabular}[c]{@{}c@{}}Visible regime \\ $\kappa = \Omega(1) $\end{tabular}
				& Visible implementation \\ \midrule
				1D 
				& $d \lesssim \log n$ 
				& $d \gtrsim \log n$ 
				& ancilla-free  \\ \midrule
				
				$\delta$D, $\delta\ge 2$ 
				& $d \lesssim (\log n)^{1/\delta}$ 
				& $d \gtrsim (\log n)^{1/\delta}$
				& clean ancillas\\ \midrule
				
				all-to-all 
				& $d \lesssim \log\log n$ 
				& $d \gtrsim \log\log n$ 
				& clean ancillas\\ 
				\bottomrule
			\end{tabular}
		\end{center}
		Here, $\lesssim$ and $\gtrsim$ denote inequalities up to constant factors, which depend solely on the architecture.
		
		\begin{enumerate}
			\item {Hidden regime: there exists an $n$-qubit single-parameter pure-state encoding $\cE$ satisfying $J_{\cE}(\theta)=1$ for every parameter point $\theta$, while any allowed depth-$d$ unitary-basis readout yields $I_{\cE}^{M}(\theta) \le \exp[-n^{\Omega(1)}]$. This bound remains valid when clean ancillas are permitted.}
			\item {Visible regime: there exists a randomized unitary-basis readout implementable in depth $d$ satisfying $I_{\cE}^{M}(\theta) \ge \kappa J_{\cE}(\theta)$ for every smooth $n$-qubit single-parameter pure-state encoding $\cE$ and parameter $\theta$, where $\kappa>0$ is an absolute constant. The same measurement protocol satisfies $I_{\cE}^{M}(\btheta) \succeq \kappa J_{\cE}(\btheta)$ for multiparameter pure-state encodings.}
		\end{enumerate}
	\end{theorem}
	
	We emphasize that the threshold established in Theorem \ref{thm:depth-optimal-readout} governs the extraction of not only continuous parameters, but also discrete quantum information, such as extracting a single classical bit encoded in two perfectly orthogonal pure states (Sec.~\ref{sec:phase-hiding}). 
	Beyond revealing the physical principles governing quantum observation, these hidden-to-visible transitions have profound operational implications, which we detail below.
	
	The hidden regime is governed by a phase-hiding phenomenon (Theorem~\ref{thm:phase_hiding_low_depth}).  
	At its core, we construct orthogonal states whose relative phase encodes a unit QFI, yet every readout below the depth threshold produces statistics that are almost independent of the phase.  
	While these adversarially constructed states may not be efficiently preparable or representative of typical natural dynamics, their existence proves that shallow circuits cannot guarantee universal constant-fraction readout.
	This mechanism naturally extends to a low-depth data-hiding phenomenon (Corollary~\ref{cor:low_depth_data_hiding}), imposing, for the first time, tight depth lower bounds on quantum information processing tasks such as parameter estimation, fidelity estimation, and state certification.
	
	Conversely, the visible regime is achieved by circumventing these hiding mechanisms via randomized measurements. We demonstrate that measurements generated by multiplicative-error approximate unitary $3$-designs (statistically pseudorandom ensembles that reproduce the uniform Haar measure up to the third moments) universally read out a constant fraction of the full QFIM for any multiparameter pure-state encoding (Theorem~\ref{thm:approx-design-readout}). 
	Crucially, this ensures that the achievable parameter estimation precision maintains the same scaling with the number of experimental runs as the ultimate quantum limit.
	Furthermore, all task dependence is entirely deferred to the classical postprocessing of the measurement record, thereby realizing a powerful ``measure first, ask questions later'' protocol~\cite{Elben2023randomized}.

	To physically implement this universal readout, we provide explicit circuit constructions of approximate unitary designs with optimal depth scaling in system size $n$ across finite-dimensional architectures with $\delta \ge 2$. 
	This closes a gap in the literature: while a recent seminal work~\cite{Schuster2024RandomUnitaries} established optimal constructions for 1D and all-to-all connectivities, the optimal scaling for higher-dimensional architectures had previously remained elusive.
	As a key intermediate ingredient, our framework also yields depth-optimal constructions of exact unitary $2$-designs in finite-dimensional architectures~\cite{Cleve2016ExactDesign}.
	Consequently, our method exponentially reduces the measurement circuit depth required for multiparameter estimation: compared to recent notable advancements relying on exact 3-designs~\cite{Zhou2026RandomizedMetrology}, we lower the depth requirement from $\cO(n)$ down to merely $\cO((\log n)^{1/\delta})$ in $\delta$D architectures, and from $\cO(\log n)$ to $\cO(\log \log n)$ in all-to-all architectures. 
	By tightly matching the established lower bounds of the hidden regime with these explicit circuit constructions, we definitively achieve depth-optimal quantum information readout across all considered architectures.
	Beyond this readout advantage, these low-depth constructions for random unitaries find broad utility across the field, given their foundational role in physics and quantum information~\cite{Schuster2024RandomUnitaries}.
	
	\section{Phase hiding below the threshold}
	\label{sec:phase-hiding}
	We now proceed to reveal the physical mechanisms underlying these hidden-to-visible transitions.
	We first establish the hidden regime of Theorem~\ref{thm:depth-optimal-readout} by deriving depth lower bounds for quantum information readout. 
	The key mechanism is a phase-hiding phenomenon, as illustrated in Fig.~\ref{fig:phase_hiding}. 
	Specifically, we construct states $\ket{\psi_\theta} = \frac{1}{\sqrt2} \left(\ket{\eta_0}+e^{i\theta}\ket{\eta_1}\right)$  that differ only by a relative phase between two orthogonal components $\ket{\eta_0}$ and $\ket{\eta_1}$, yet remain indistinguishable to any low-depth circuit. Consequently, while the information is present in the quantum state, it remains inaccessible to the entire family of low-depth measurements. 
	Crucially, this holds even when the readout is tailored with full prior knowledge of the encoding.
	
	\begin{figure*}[!htbp]
		\centering
		\includegraphics[width=0.9\linewidth]{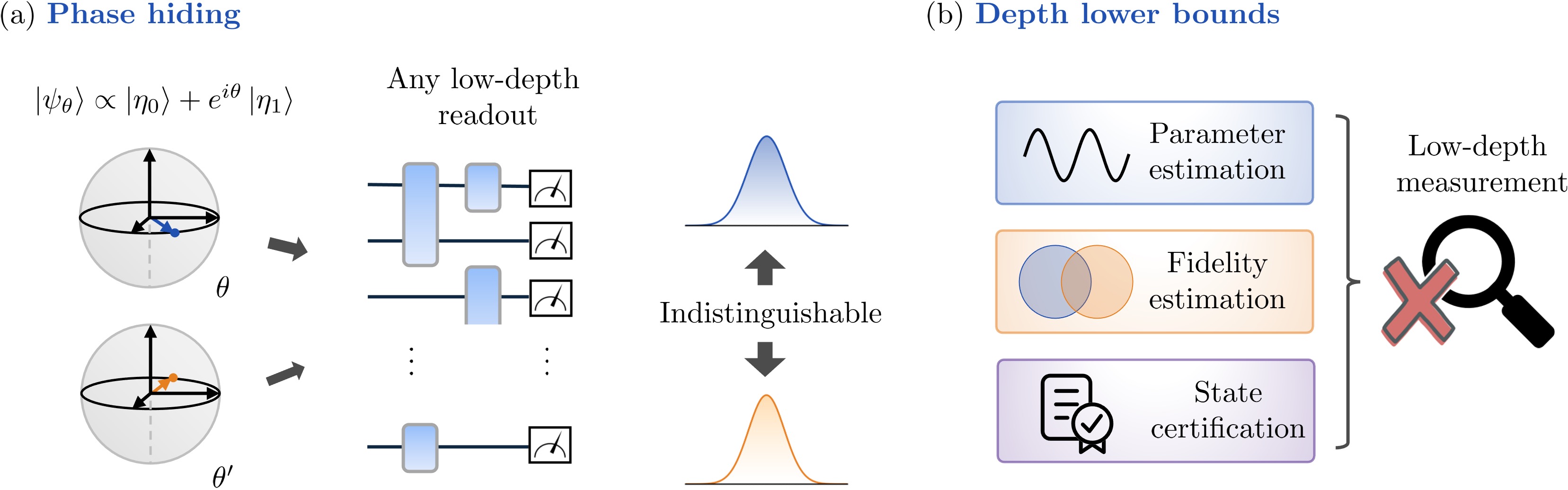}
		\caption{Phase hiding and depth lower bounds for quantum information readout.
			(a) Phase hiding. We prove that there exist orthogonal states $\ket{\eta_0}$ and $\ket{\eta_1}$ such that, for the phase family 
			$\ket{\psi_\theta}=(\ket{\eta_0}+e^{i\theta}\ket{\eta_1})/\sqrt2$, any unitary-basis readout below the depth threshold cannot distinguish states with different phases $\theta$ and $\theta'$. 
			Thus, the phase information is present in the quantum state but hidden from all low-depth readout circuits. 
			(b) Depth lower bounds. 
			The phase indistinguishability implies that, in the worst case, any sample-efficient protocol for parameter estimation, fidelity estimation, or state certification must use readout circuits with depth exceeding the corresponding threshold.}
		\label{fig:phase_hiding}
	\end{figure*}

	This no-go result goes beyond existing depth lower bounds for generating unitary designs~\cite{Schuster2024RandomUnitaries, Cui2025UnitaryDesignsOptimal}.
	Because those prior results apply only to random unitaries, one might naturally expect that task-specific measurement could bypass these constraints to extract information at a much shallower depth.
	Our result refutes this intuition, proving that even fully tailored measurements require the same depth scaling as randomized protocols.
	
	\begin{theorem}[Phase hiding against low-depth readouts, informal]
		\label{thm:phase_hiding_low_depth}
		For sufficiently large $n$ and any circuit depth $d$ within the hidden regime, there exist two orthonormal $n$-qubit states $\ket{\eta_0}$ and $\ket{\eta_1}$ whose relative phase is inaccessible to any allowed readout. 
		
		Specifically, define $\ket{\psi_\theta} = \frac{1}{\sqrt 2} \left( \ket{\eta_0}+e^{i\theta}\ket{\eta_1} \right)$. For every allowed depth-$d$ unitary-basis readout, and for all phases $\theta,\theta'\in\mathbb R$,
		\begin{equation}
			\label{eq:phase_hiding_tv}
			\TVD{p_{\psi_\theta}^U}{p_{\psi_{\theta'}}^U} \le \exp[-n^{\Omega(1)}].
		\end{equation}
		Here, $\TVD{p}{q} \coloneqq \frac{1}{2}\sum_z \abs{p(z)-q(z)}$ is the total-variation distance. 
	\end{theorem}
	
	The proof proceeds in three key steps. First, the computational-basis measurement following $U$ is a sequence of commuting dephasing operations, generated by the Heisenberg-evolved observables $U^\dagger Z_iU$. 
	Second, the bounded circuit depth ensures that many of these local dephasings possess small, mutually disjoint backward light cones. 
	Third, we construct a random product-state code designed so that a constant fraction of these local dephasings contract the off-diagonal coherent phase.
	Because the light cones are disjoint, these individual contractions multiply and exponentially suppress the phase dependence uniformly over all allowed circuits. 
	The full proof is detailed in Appendix~\ref{app:phase-hiding}.
	
	We now connect the phase-hiding phenomenon to a CFI bound.  
	Consider the phase-hiding family $\cE_{\rm hid} =\{\ket{\psi_{\theta}}\}_{\theta\in[0,2\pi)}$. Since the resulting measurement distributions are exponentially close, their sensitivity to the parameter $\theta$ is severely limited.
	Consequently, the CFI extracted by any allowed depth-$d$ readout is exponentially small, yielding $I_{\cE_{\rm hid}}^U(\theta) \le \exp[-n^{\Omega(1)}]$  (see Methods for the detailed derivation). 
	In contrast, the QFI of the same encoding is $J_{\cE_{\rm hid}}(\theta)=1$. 
	This gap proves that the readout capability is exponentially suppressed below the depth threshold, establishing the hidden regime of Theorem~\ref{thm:depth-optimal-readout}.
	
	This phase-hiding mechanism naturally extends to encode and hide a discrete bit.
	For any $\theta \in \bR$, encode a binary variable \(b\in\{0,1\}\) by
	\begin{equation}
		\ket{\phi_b}
		:=
		\ket{\psi_{\theta + b\pi}}
		=
		\frac{1}{\sqrt 2}
		\left(
		\ket{\eta_0}+(-1)^be^{i\theta}\ket{\eta_1}
		\right).
	\end{equation}
	The two codewords $\ket{\phi_0}$ and $\ket{\phi_1}$ are orthogonal, so an optimal measurement can recover the bit perfectly from a single copy. 
	However, Theorem~\ref{thm:phase_hiding_low_depth} implies they remain almost indistinguishable to low-depth measurements. 
	
	\begin{corollary}[Data hiding against low-depth readouts]
		\label{cor:low_depth_data_hiding}
		For sufficiently large $n$ and any circuit depth $d$ within the hidden regime, there exist two orthogonal $n$-qubit pure states $\ket{\phi_0}$ and $\ket{\phi_1}$ that are indistinguishable to any allowed readout. 
		
		Specifically, any protocol using allowed single-copy readouts to distinguish $\ket{\phi_0}^{\otimes T}$ and $\ket{\phi_1}^{\otimes T}$ with an error probability of at most $1/3$  requires $T=\Omega\left(\exp[n^{\Omega(1)}]\right)$ copies.
	\end{corollary}
	
	We note that these hiding phenomena persist even when classical randomization and adaptive circuit choices are allowed across independent copies (see Appendix~\ref{app:randomized_adaptive}). 
	Nonetheless, our model does not allow adaptive operations within a single copy. 
	This distinction is necessary as adaptive local operations and classical communication (LOCC) can distinguish any two orthogonal pure states~\cite{Walgate2000LocalDistinguish} and saturate the QFI of pure-state encodings~\cite{Zhou2020SaturatingLOCC}. 
	However, implementing such protocols requires sequential, on-the-fly basis updates, imposing extensive classical processing overhead and prohibitively long quantum-memory coherence times.
	Furthermore, our hiding mechanisms differ from standard data hiding against LOCC measurements~\cite{Bennett1999NonlocalityWithoutEntanglement, Terhal2001HidingBits,Eggeling2002HidingClassicalData}.
	While those protocols exploit locality constraints to hide information within mixed states, our mechanism hides orthogonal pure states by restricting measurement complexity.
	
	This data hiding result immediately implies circuit depth lower bounds for sample-efficient fidelity estimation and state certification, since any protocol that reliably estimates or certifies fidelity can also distinguish two orthogonal states. 
	Conversely, multiplicative-error approximate $3$-designs enable efficient fidelity estimation via classical shadow tomography~\cite{Huang2020Predicting, Schuster2024RandomUnitaries}.
	Our $\delta$D constructions in Theorem~\ref{thm:finite-dimensional-design}, alongside known 1D and all-to-all constructions~\cite{Schuster2024RandomUnitaries, Cui2025UnitaryDesignsOptimal}, collectively achieve this depth scaling. 
	Consequently, our lower bounds are tight up to constant factors, even when the target state is fully known before the measurement circuit is chosen.
	This demonstrates that the same complexity-driven transition governs the extraction of discrete quantum information: below the critical depth, exponentially many samples are required for state discrimination, whereas above this depth scale, randomized measurements can effectively extract this information with a sample complexity independent of system size.
	
	Furthermore, our findings complement known separations in distributed quantum metrology.  
	There, spatially separated sensors impose locality constraints, resulting in quadratic CFI gaps between local and entangled measurements in mixed-state sensing tasks such as nonlocal optical interferometry~\cite{Gottesman2012LongerBaselineTelescope,Tsang2011NonlocalityInterferometry,Stas2026NonlocalInterferometry}.
	Here, by contrast, we demonstrate that the CFI gap between measurements with and without complexity restrictions can be exponentially large.
	
	\section{Universal readout above the threshold}
	Having revealed the constraints of the hidden regime, we now demonstrate how exceeding the complexity threshold unlocks the encoded quantum information. 
	This shifts the system into the visible regime, where we show that a constant fraction of the QFIM can be universally extracted. 
	To achieve this, our readout protocol (Fig.~\ref{fig:random_unitary}) operates in two steps.
	First, we scramble the state using random unitaries sampled from multiplicative-error approximate 3-designs, proving that subsequent measurements guarantee effective information extraction.
	Second, we present depth-optimal constructions of these designs across various hardware architectures.
	
	\begin{figure*}[!htbp]
		\centering
		\includegraphics[width=0.6\linewidth]{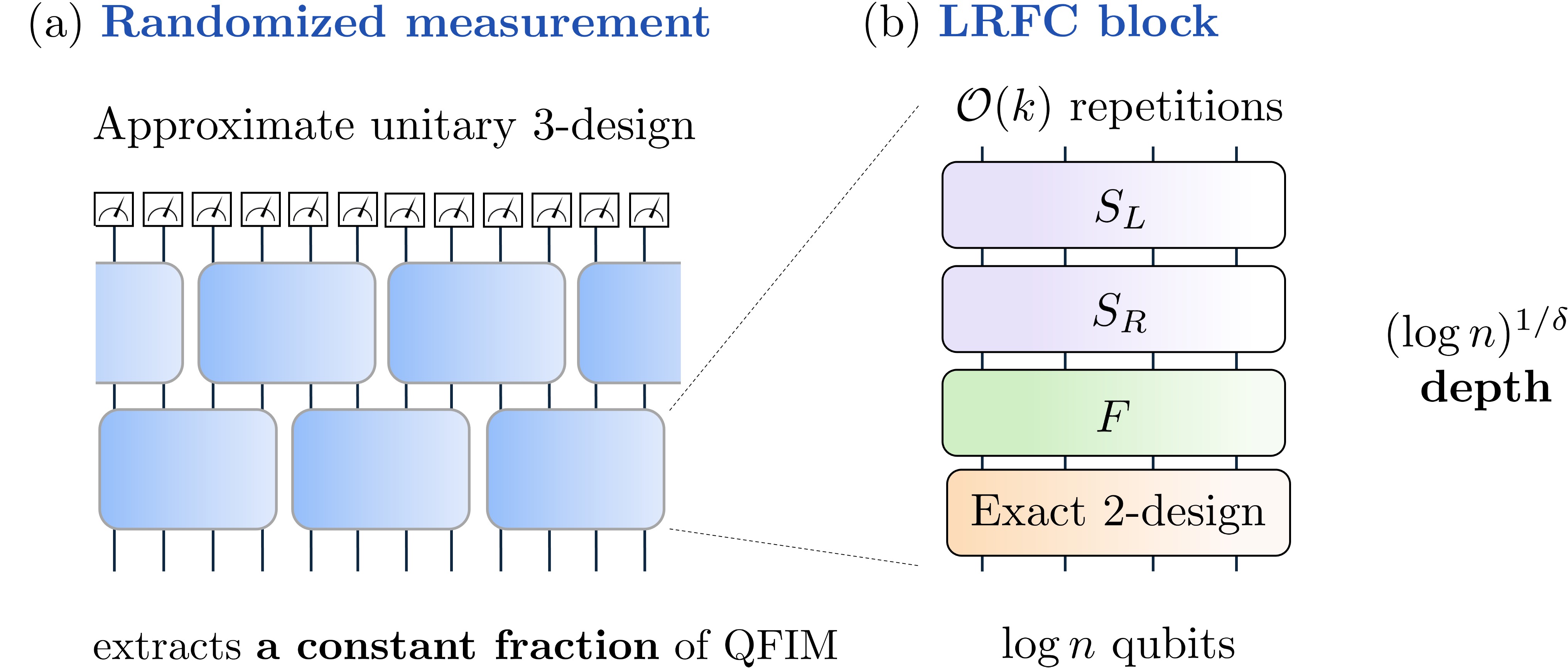}
		\caption{Universal quantum information readout from approximate unitary designs.
			(a) A randomized unitary-basis readout is obtained by sampling a circuit from a multiplicative-error approximate unitary $3$-design, applying it to the data qubits, and measuring in the computational basis. 
			We prove that this measurement universally extracts a constant fraction of the QFIM for any pure-state multiparameter encoding. 
			The global approximate design is realized by a double-layer blocked circuit, where the data qubits are partitioned into $\Theta(\log n)$-qubit patches, and independent local random unitaries are applied to neighboring patch pairs in two shifted layers. 
			(b) Internal structure of the local random unitary. Each local unitary implements $\cO(k)$ repetitions of the Luby-Rackoff-Function-Clifford (LRFC) block, $U_{\rm LRFC}=S_LS_RFC$, where $C$ is an exact unitary 2-design, $F$ is a $2k$-wise independent phase operation, and $S_L$ and $S_R$ are $2k$-wise independent shuffle operations. 
			We present efficient implementations of these primitives on a $\delta$-dimensional architecture, yielding optimal depth scaling $\cO\left((\log n)^{1/\delta}\right)$ with respect to system size $n$ for both approximate unitary $k$-designs and exact unitary $2$-designs.
		}
		\label{fig:random_unitary}
	\end{figure*}
	
	Formally, let $\cU=\{q_a,U_a\}$ be an ensemble of unitaries on a $D$-dimensional Hilbert space.  
	The procedure of sampling $a$, applying $U_a$, and measuring in the computational basis defines the POVM $M_{\cU} = \left\{q_a\,U_a^\dagger\ketbra{y}{y}U_a \right\}_{a,y}$.
	For a smooth multiparameter encoding $\cE=\{\ket{\psi_{\btheta}}\}$ with $\btheta\in\bR^m$, let $J_{\cE}(\btheta)$ denote the $m \times m$ QFIM, and let $I_{\cE}^{M_{\mathcal{U}}}(\btheta)= \mathbb E_{U\sim\cU} I_{\cE}^{U}(\theta)$ denote the corresponding CFIM achieved by this measurement. 
	The following theorem establishes that when $\cU$ forms an approximate unitary $3$-design, this randomized measurement universally extracts a constant fraction of the full QFIM for any multiparameter encoding.
	
	\begin{theorem}[Constant-fraction readout via approximate $3$-design]
		\label{thm:approx-design-readout}
		Let $M_{\cU}$ be the POVM induced by a
		multiplicative-$\epsilon$ approximate unitary $3$-design on a $D$-dimensional Hilbert space, with $0\le \epsilon<1/4$.  
		Then, for every smooth multiparameter pure-state encoding $\cE$ and every parameter point $\btheta$,
		\begin{equation}
			I_{\cE}^{M_{\cU}}(\btheta)
			\succeq
			\kappa_D(\epsilon)
			J_{\cE}(\btheta),
			\qquad
			\kappa_D(\epsilon)
			=
			\frac{(1-4\epsilon)^2}{1+6\epsilon}
			\frac{D+2}{4(D+1)} .
			\label{eq:qfim-design-readout}
		\end{equation}
	\end{theorem}
	
	For any fixed $\epsilon<1/4$, the coefficient $\kappa_D(\epsilon)$ is bounded below by a positive constant.
	A $3$-design suffices because the matrix inequality can be reduced to every one-dimensional tangent direction. 
	Along each such direction, the CFI can be lower-bounded by a Cauchy--Schwarz ratio involving the induced second and third moments of the randomized measurement. 
	The design condition then ensures that these moments approximate the corresponding Haar moments and yield the coefficient $\kappa_D(\epsilon)$ uniformly over all tangent directions. The full proof is in Appendix~\ref{app:metrology}.
	
	Random Clifford unitaries form an exact 3-design~\cite{Zhu2017MultiqubitClifford} and achieve a constant readout capability~\cite{Zhou2026RandomizedMetrology}. 
	However, implementing a generic multiqubit Clifford unitary requires depth $\cO(n)$ on $\delta$D architectures and $\cO(\log n)$ on all-to-all architectures~\cite{Aasonsom2004ImprovedSimulation,Bravyi2021HadamardClifford, Moore2001ParallelComputing}.  
	These requirements far exceed the depth thresholds established in Theorem~\ref{thm:depth-optimal-readout}. 
	Theorem~\ref{thm:approx-design-readout} circumvents this bottleneck by demonstrating that exact unitary designs are unnecessary, and a multiplicative approximation suffices.

	This relaxation is crucial for drastically reducing the required circuit depth. 
	Indeed, existing low-depth constructions already realize the required approximate unitary $3$-designs in depth $\cO(\log n)$ for one-dimensional circuits and $\cO(\log\log n)$ for all-to-all architectures~\cite{Schuster2024RandomUnitaries, Cui2025UnitaryDesignsOptimal}. 
	While these works successfully establish the depth upper bound of Theorem~\ref{thm:depth-optimal-readout} for those specific connectivities, a gap remains regarding the optimal construction for general $\delta$D architectures.
	We close this gap by introducing constructions of approximate unitary designs that scale optimally with system size $n$ in every $\delta$D architecture.
	
	\begin{theorem}[Low-depth approximate designs with $\delta$D implementations, informal]
		\label{thm:finite-dimensional-design}
		Fix an architecture dimension $\delta \ge 1$. An $n$-qubit multiplicative-$\epsilon$ approximate unitary $k$-design can be  implemented on a $\delta$D architecture using $\cO(nk)$ clean ancillas in depth $\cO\left(k \log (k) \left(k\log\frac{n k}{\epsilon}\right)^{1/\delta}\right)$.
	\end{theorem}
	
	The construction, illustrated in Fig.~\ref{fig:random_unitary}, employs a double-layer blocked-circuit gluing strategy. 
	The data qubits are partitioned into patches of size $\xi = \Theta(\log\frac{n k}{\epsilon})$.
	Realizing a local approximate $k$-design with multiplicative error $\cO(\epsilon/n)$ on neighboring patches suffices to form the desired global multiplicative-error design~\cite{Schuster2024RandomUnitaries}. 
	To construct these local designs, we repeat an Luby-Rackoff-Function-Clifford (LRFC) block, $U_{\rm LRFC}=S_LS_RFC$, for $\cO(k)$ times~\cite{Cui2025UnitaryDesignsOptimal}. 
	Within this block, $C$ is an exact unitary $2$-design, $F$ applies a $2k$-wise independent phase, and $S_L,S_R$ perform $2k$-wise independent shuffles.
	We prove that all components within an LRFC block can be implemented in depth $\cO\left(\log(k) (k\xi)^{1/\delta}\right)$. 
	Notably, our construction also gives the first depth-optimal exact unitary 2-design on $\delta$D architectures~\cite{Cleve2016ExactDesign}.
	
	Finally, specializing Theorem~\ref{thm:finite-dimensional-design} to $k=3$ with a constant $\epsilon$ gives multiplicative approximate unitary $3$-designs in depth $\cO((\log n)^{1/\delta})$.
	Combined with Theorem~\ref{thm:approx-design-readout}, this establishes the visible regime of Theorem~\ref{thm:depth-optimal-readout}, thereby definitively achieving depth-optimal quantum information readout across all considered architectures.
	
	\section{Discussion}
	Our results open numerous directions for future work. To move beyond our worst-case depth lower bounds, a crucial first step is to investigate measurement complexity in specific tasks. Practical problems often come with structural promises regarding the encoding, such as limited entanglement or a known dynamical model, which may enable substantially shallower readouts than our established limits~\cite{Pezze2025DistributedOptimalLocal, Huang2024Certifying, Du2025Localizable}. 
	A complementary challenge lies in computational complexity. Even when information is accessible in an information-theoretic sense, recovering it from measurement data can remain computationally intractable. Therefore, developing a unified framework that captures the interplay among measurement, sample, and computational complexity is an important direction for future research.
	
	Beyond complexity considerations, optimizing the readout protocols for practical hardware presents another key challenge. For example, the higher-dimensional and all-to-all implementations currently rely on clean ancillas. Determining whether these ancillas can be removed without compromising the optimal depth scaling remains an important open problem. Furthermore, extending our framework beyond the single-copy, pure-state regime to other scenarios, such as low-rank mixed-state encodings and collective measurements on a small number of copies~\cite{Chen2022InformationGeometryHierarchical}, is a highly valuable yet largely unexplored direction.
	
	Perhaps the most intriguing direction is to develop a systematic resource theory of measurements.
	Historically, quantum information theory has primarily been organized around resources carried by quantum states, such as entanglement~\cite{Horodecki2009Entanglement, Chitambar2019QuantumResourceTheories}, coherence~\cite{Baumgratz2014QuantifyingCoherence}, and magic~\cite{Bravyi2005Universal, Howard2017ResourceTheoryMagic}. 
	Our work identifies measurement complexity as an operational resource required to access information stored in quantum states.
	Other measurement resources, such as limited ancillary qubits,  mid-circuit measurements, and quantum communication, should give rise to their own resource hierarchies. 
	Recent classifications of joint measurements based on finite entanglement have already highlighted the rich structure of measurement~\cite{Pauwels2025ClassificationJointMeasurements}.  We anticipate that a resource theory of measurements, complementary to existing resource theories of states, will open new directions in quantum foundations and in information-processing tasks such as quantum learning, certification, and metrology.
	
	\section*{Acknowledgments}
	Z.D., S.C., H.Y., J.C, and X.M. acknowledge the support from the National Natural Science Foundation of China Grants No.~12174216,  the Innovation Program for Quantum Science and Technology Grant No.~2021ZD0300804,  No.~2021ZD0300702, the CCF-QuantumCtek Superconducting Quantum Computing Special Cooperation Program (Grant No.~CCF-QC2025005), and the Turing AI Institute of Nanjing.
	X.Y. is supported by Beijing Natural Science Foundation Z250004, the National Natural Science Foundation of China NSAF (Grant No.~U2330201) and  Grant (No.~12361161602),  the Quantum Science and Technology-National Science and Technology Major Project (2023ZD0300200),   and Beijing Science and Technology Planning Project (Grant No.~Z25110100810000). 
	
	\bibliography{ref}
	
	\section*{Methods}
	
	\subsection{Readout model and Fisher information}
	We first formalize the bounded-depth readout model. Let $G$ be an interaction graph with vertices corresponding to the data qubits $[n] \coloneqq \{1,2,\cdots,n\}$. We denote by $\cU_{n,d}^{G}$ the set of depth-$d$ circuits on $G$, where any $U \in \cU_{n,d}^{G}$ can be decomposed as
	\begin{equation}
		U=U^{(d)}\cdots U^{(1)},\qquad
		U^{(t)}=\bigotimes_j U_j^{(t)},
		\label{eq:depth-d-circuit}
	\end{equation}
	where the gates $U_j^{(t)}$ within each layer have pairwise disjoint supports, and each is either a one-qubit gate or a two-qubit gate acting on an edge of $G$. 
	
	We also consider a restricted ancilla-assisted model for implementing data-qubit unitaries.
	In this setting, $G$ denotes a fixed architecture with vertex set $Q\sqcup A$, where $Q$ is the set of $n$ data qubits and $A$ consists of ancillary qubits. An $n$-qubit unitary $U$ belongs to $\cU_{n,d}^{G,\mathrm{anc}}$ if there exists an initial ancilla state $\tau_A$ and a depth-$d$ circuit $V$ on $G$, such that 
	\begin{equation}
		\tr_A\left[V(\rho\otimes\tau_A)V^\dagger\right] = U\rho U^\dagger
		\label{eq:ancilla-assisted-implementation}
	\end{equation}
	for every input state $\rho$ on $Q$.
	The final ancilla state need not equal $\tau_A$. 
	When we say that an implementation uses \emph{clean ancillas}, we mean the special case in which the ancillas are initialized in the product state $\ket{0}^{\otimes |A|}$ and are returned to $\ket{0}^{\otimes |A|}$ at the end of the circuit.  
	Thus, clean ancillas are reusable qubits, and they are included as a special case of \(\cU_{n,d}^{G,\mathrm{anc}}\).
	For all-to-all architectures, $G$ is the complete graph on $Q\sqcup A$.
	For $\delta$D architectures, $G$ is a $\delta$D grid, with the data subset $Q\subseteq G$ fixed as part of the architecture.
	
	This class should not be confused with arbitrary depth-$d$ ancilla-assisted measurements. 
	The ancillas are used only to implement an effective unitary channel on the data register. The final readout is still the computational-basis measurement of the data qubits, with outcome distribution $p_\rho^U(z)= \bra{z}U\rho U^\dagger\ket{z}$.
	Thus, the model remains a unitary-basis readout model. More general protocols that directly measure ancillas, use mid-circuit measurements, or apply feed-forward within a single copy define stronger measurement models and are not included in $\cU_{n,d}^{G,\mathrm{anc}}$~\cite{Du2025SpaceTimeMeasurement, Foxman2025RandomUnitariesConstantQuantum}. 
	
	We quantify the information extracted by these readout circuits using Fisher information matrices. 
	Consider a smooth $m$-parameter pure-state encoding 
	$\cE= \{\rho_{\btheta}=\ketbra{\psi_{\btheta}}\}_{\btheta\in\bR^m}$.
	Let $L_i$ be the symmetric logarithmic derivatives (SLDs), defined by $\partial_i \rho_{\btheta} = \frac{1}{2}\left(L_i\rho_{\btheta}+\rho_{\btheta}L_i\right)$,
	and define the QFIM by
	\begin{equation}
		[J_{\cE}(\btheta)]_{ij} := \Re\tr\left(\rho_{\btheta}L_iL_j\right).
	\end{equation}
	For a POVM $M=\{M_x\}_x$, with outcome probabilities
	$p_{\btheta}(x)=\Tr(M_x\rho_{\btheta})$, the corresponding CFIM is
	\begin{equation}
		[I_{\cE}^M(\btheta)]_{ij} := \sum_{x:p_{\btheta}(x)>0} \frac{\partial_i p_{\btheta}(x)\,\partial_j p_{\btheta}(x)}{p_{\btheta}(x)}.
	\end{equation}
	
	\subsection{Proof strategy for the hidden regime}
	In this section, we formally state the phase-hiding theorem, outline the core ideas of its proof, and convert this theorem into a bound on the CFI.
	
	\begin{theorem}[Phase hiding against low-depth readouts, formal version of Theorem~\ref{thm:phase_hiding_low_depth}]
		\label{thm:phase_hiding_low_depth-formal}
		There exist positive constants \(\{\alpha_\delta\}_{\delta\ge1}\) and \(\alpha_{\mathrm{a2a}}>0\) such that the following holds for sufficiently large \(n\).  Let \(G\) be either a $\delta$D architecture or the all-to-all architecture.  Suppose the readout depth satisfies $d\le \alpha_\delta(\log n)^{1/\delta}$ if $G$ is a $\delta$D architecture, or $d\le \alpha_{\mathrm{a2a}}\log\log n$ if $G$ is an all-to-all architecture.
		Then there exist two orthonormal \(n\)-qubit states
		\(\ket{\eta_0}\) and \(\ket{\eta_1}\) such that the phase family $\ket{\psi_\theta} =  \frac{1}{\sqrt 2} \left( \ket{\eta_0}+e^{i\theta}\ket{\eta_1} \right)$ is hidden from every allowed depth-\(d\) unitary-basis readout in the architecture \(G\).  More precisely, for every such unitary $U$, and for all phases \(\theta,\theta'\in\mathbb R\),
		\begin{equation}
			\label{eq:phase_hiding_tv-formal}
			\TVD{p_{\psi_\theta}^U}{p_{\psi_{\theta'}}^U} \le \exp[-n^{\Omega(1)}].
		\end{equation}
		The bound remains valid when the architecture uses clean ancillas.
	\end{theorem}
	
	\subsubsection{Proof ideas}
	The phase information of $\ket{\psi_\theta}$ is contained entirely in the coherence term 
	\begin{equation}
		\ketbra{\psi_\theta} - \bar\rho =\frac12\left(e^{-i\theta}\ket{\eta_0}\bra{\eta_1}+e^{i\theta}\ket{\eta_1}\bra{\eta_0}\right),
	\end{equation}
	where $\bar\rho = \frac12 \left(\ket{\eta_0}\bra{\eta_0} + \ket{\eta_1}\bra{\eta_1}\right)$.
	We show that every low-depth measurement almost completely destroys this coherence.
	
	A computational-basis measurement after $U$ can be written, in the Heisenberg
	picture, as a sequence of commuting dephasings.  Namely, with $A_i(U)=U^\dagger Z_iU$, $\mathcal D_{A_i}(X)=\frac12\left(X+A_iXA_i\right)$, the measurement channel is the product of the dephasings $\mathcal D_{A_i}$.  Therefore, the total variation distance between two output distributions is controlled by the trace norm of the coherence term after these dephasings.
	
	Low depth imposes many small, disjoint light cones.
	For a $\delta$D depth-$d$ circuit, one can choose $m=\Omega(n/L)$ output qubits with pairwise disjoint backward light cones of size $L=\mathcal O(d^\delta)$.  For an all-to-all depth-$d$ circuit, each backward light cone has size $L=2^d$. A greedy packing yields $m=\Omega(n/L^2)$ disjoint light cones, and the number of possible support patterns is at most $R\le (n+1)^n$.
	
	The states $\ket{\eta_0}$ and $\ket{\eta_1}$ are built from a random product-state code.  
	One samples many random product states, splits them into two sets, and forms two nearly orthogonal uniform superpositions over those sets.  
	For any dephasing channel on a light cone of size $L$, the coherence between two independent random product states is contracted with probability at least $\exp[-\mathcal O(L)]$.  
	Since the selected light cones are disjoint, these contractions multiply.
	Therefore, after the selected dephasings, the coherence terms are exponentially suppressed, so all phase states produce almost identical classical distributions. 
	A concentration argument, together with a net over local dephasings and a union bound over possible light-cone support patterns, makes the contraction uniform over all $U\in\cU_{n,d}^G$, proving Theorem~\ref{thm:phase_hiding_low_depth-formal}. 
	The full proof is given in the Appendix~\ref{app:phase-hiding}.

	\subsubsection{Convert phase hiding into a CFI bound}
	For the phase-hiding family $\cE_{\rm hid} =\{\ket{\psi_{\theta}}\}_{\theta\in[0,2\pi)}$.
	For a fixed circuit $U$ below the depth threshold, let $a_z=\bra z U\ket{\eta_0}$, $b_z=\bra z U\ket{\eta_1}$.
	Then
	\begin{equation}
		p_{\psi_{\theta}}^U(z) = \bar p^U(z) + r_z\cos(\theta+\alpha_z),
	\end{equation}
	where $\bar p^U(z) \coloneqq \mathbb E_\varphi p_{\psi_{\varphi}}^U = \frac12\bigl(|a_z|^2+|b_z|^2\bigr)$ and $r_z e^{i\alpha_z} = b_z a_z^*$.
	Positivity of $p_{\psi_{\varphi}}^U(z)$ for all $\varphi$ implies $\bar p^U(z)\ge r_z$.
	Convexity of total variation and Eq.~\eqref{eq:phase_hiding_tv-formal} give $\TVD{p_{\psi_{\theta}}^U}{\bar p^U} \le \exp[-n^{\Omega(1)}]$. Averaging over \(\theta\) yields $\sum_z r_z \le \exp[-n^{\Omega(1)}]$.
	The CFI of the measurement distribution is therefore bounded by
	\begin{equation}
		\begin{split}
			I_{\cE_{\rm hid}}^U(\theta) &= \sum_z \frac{r_z^2\sin^2(\theta+\alpha_z)}{\bar p^U(z)+r_z\cos(\theta+\alpha_z)}  \\
			&\le 2\sum_z r_z = \exp[-n^{\Omega(1)}].
		\end{split}
		\label{eq:cfi_phase_hiding}
	\end{equation}
	This proves that the readout capability is exponentially suppressed below the depth threshold, establishing the hidden-regime statement in Theorem~\ref{thm:depth-optimal-readout}.

	\subsection{Proof strategy for the visible regime}
	In this section, we outline the two key ingredients for establishing the visible regime. 
	We first explain how approximate unitary 3-designs guarantee constant-fraction QFIM extraction, and then describe the explicit circuit construction of these designs in $\delta$D architectures using LRFC blocks.
	
	\subsubsection{Proof ideas of Theorem~\ref{thm:approx-design-readout}}
	
	Fix any $\bu\in\mathbb R^m$ and restrict to the one-parameter encoding
	$t\mapsto \rho_{\btheta+t\bu}$.
	Then
	\begin{equation}
		\bu^{\mathsf T}I_{\cE}^{M_{\cU}}(\btheta)\bu
		=
		I_t,
		\qquad
		\bu^{\mathsf T}J_{\cE}(\btheta)\bu
		=
		J_t,
	\end{equation}
	where $I_t$ and $J_t$ are respectively the CFI and QFI of the one-parameter encoding.
	Let $L_t$ denote the corresponding SLD. 
	For the POVM induced by $\cU=\{q_a,U_a\}$, write its outcomes as
	$x=(a,y)$ and set $\ket{\phi_x} =U_a^{\dagger} \ket{y}$, $w_x:=\frac{q_a}{D}$.
	One finds
	\begin{equation}
		I_t =\frac{D}{4} \sum_{x:p_t(x)>0} w_x \frac{\langle\phi_x|L_t|\phi_x\rangle^2}{\langle\phi_x|\rho_t|\phi_x\rangle}.
	\end{equation}
	
	A Cauchy--Schwarz inequality lower-bounds this quantity by the square of a second-moment term divided by a third-moment term:
	\begin{equation}
		I_t
		\ge
		\frac{D}{4}
		\frac{\left(\sum_x w_x\bra{\phi_x}L_t\ket{\phi_x}^2\right)^2}
		{\sum_x w_x\bra{\phi_x}L_t\ket{\phi_x}^2\,\bra{\phi_x}\rho_t\ket{\phi_x}}.
	\end{equation}
	
	By assumption, the second and third moments of $M_{\cU}$ are close to the corresponding Haar moments, and these Haar moments can be evaluated explicitly for pure-state SLDs.
	This yields $I_t\ge \kappa_D(\epsilon) J_t$ for every $\bu$, and therefore $I_\cE^{M_{\cU}}(\btheta) \succeq \kappa_D(\epsilon) J_{\cE}(\btheta)$. The full proof is given in Appendix~\ref{app:metrology}.
	
	\subsubsection{Depth-optimal implementation of the designs}
	\label{sec:methods-design-implementation}
	
	Having established that approximate 3-designs are sufficient for QFI readout, we now explicitly construct them at the required depth. We first build local approximate designs on logarithmic-size blocks using LRFC blocks, and then glue these local patches into a global multiplicative-error design using a double-layer blocked circuit. 
	The full constructions and proofs are detailed in Appendices~\ref{app:finite-dimensional-arithmetic}--\ref{app:random-unitaries}.
	
	Let $\Lambda$ be an even-size block of $\ell=2h$ qubits, with bipartition $\Lambda=\Lambda_L\sqcup \Lambda_R$ and $|\Lambda_L|=|\Lambda_R|=h$. We identify the computational-basis states on the two halves with $(x_L,x_R)\in \mathbb F_{2^h}\times \mathbb F_{2^h}$, where $\mathbb{F}_{2^h}$ denotes the finite field of order $2^h$. A single LRFC block applies the unitary 
	\begin{equation}
		U_{\rm LRFC}=S_LS_RFC,
		\label{eq:methods-lrfc}
	\end{equation}
	composed of four carefully structured components:
	\begin{enumerate}
		\item Clifford ($C$): $C$ is sampled from an exact unitary 2-design on the full block $\Lambda$. To implement this efficiently without the demanding overhead of the full Clifford group, we construct $C$ via a restricted finite-field Clifford ensemble. Identifying the Hilbert space with $\mathrm{span}\{\ket{x}:x\in K\}$ for $K=\mathbb F_{2^\ell}$, we first define the Weyl displacement operators $D_{(a,b)} = i^{\mathrm{Tr}(ab)}X_aZ_b$, where $X_a\ket{x}=\ket{x+a}$, $Z_b\ket{x}=(-1)^{\mathrm{Tr}(bx)}\ket{x}$ and $\mathrm{Tr}:K\to\mathbb F_2$ is the finite-field trace. We then sample $C$ uniformly from the restricted ensemble
		\begin{equation}
			\mathcal C_{\rm res}(K) = \{\,U_MD_u:\, M\in{\rm SL}(2,K),\,u\in K^2\,\}.
		\end{equation} 
		Here, $U_M$ is a Clifford lift that transforms the Weyl basis via $U_MD_vU_M^\dagger = \pm D_{Mv}$. By restricting the symplectic transformations to ${\rm SL}(2,K)$, the operator $C$ can be compiled into a short circuit consisting only of finite-field displacements, Fourier transforms, scalings, and quadratic shears (see Appendix~\ref{app:exact-2-design}).
		
		\item Diagonal phase ($F$): $F$ applies a phase $F\ket{x}=(-1)^{f(x)}\ket{x}$. The function $f:\mathbb F_2^\ell\to\mathbb F_2$ is sampled from a $2k$-wise independent family by setting $f(x)=\lambda(P(x))$, where $P$ is a random polynomial of degree $<2k$ over $\mathbb F_{2^\ell}$ and $\lambda$ is a fixed nonzero linear functional.
		
		\item Conditional shuffles ($S_R, S_L$): $S_R$ is a conditional shuffle of the right half, $S_R\ket{x_L,x_R} = \ket{x_L,x_R+\sigma_R(x_L)}$, where $\sigma_R:\mathbb F_{2^h}\to\mathbb F_{2^h}$ is sampled from a $2k$-wise independent vector-valued function family. $S_L$ is the analogous shuffle of the left half, $S_L\ket{x_L,x_R} = \ket{x_L+\sigma_L(x_R),x_R}$, driven by an independently sampled $2k$-wise independent function $\sigma_L$.
	\end{enumerate}
	
	\begin{theorem}[Finite-dimensional implementation of LRFC components]
		\label{thm:methods-lrfc-implementation}
		Fix an architecture dimension $\delta\ge1$. On an $\ell$-qubit block, every sampled component $C,F,S_R,S_L$ in Eq.~\eqref{eq:methods-lrfc} can be implemented on a $\delta$D architecture using $\cO(k\ell)$ clean ancillas and depth $\cO\left(\log(k)(k\ell)^{1/\delta}\right)$.
	\end{theorem}
	The specific construction details and depth analysis for these components are provided in the proofs of Theorems~\ref{thm:hash-eval} and \ref{thm:finite-dimensional-exact-2-design} in the Appendix.
	
	It remains to turn these local blocks into a global design. We partition the system into patches $P_1, \ldots, P_m$ (as shown in Fig.~\ref{fig:random_unitary}), setting the individual patch size to $|P_i| = \Theta(\xi)$ with $\xi = \Theta(\log \frac{nk}{\epsilon})$. The global circuit is executed in two interleaved layers. The first layer applies independent local random unitaries to adjacent disjoint pairs, such as $P_1\sqcup P_2$, $P_3\sqcup P_4$, $\cdots$. The second layer then applies independent random unitaries to the shifted pairs, such as $P_2\sqcup P_3$, $P_4\sqcup P_5$, $\cdots$. 
	Consequently, adjacent blocks from successive layers overlap on an entire patch.
	Within each local block of size $\ell = \Theta(\xi)$, repeating the LRFC circuit $\cO(k)$ times forms a local approximate $k$-design~\cite{Cui2025UnitaryDesignsOptimal}. By the gluing theorem from Ref.~\cite{Schuster2024RandomUnitaries}, this rigorously yields a global multiplicative-$\epsilon$ approximate unitary $k$-design.
	
	Because all blocks in a single layer are disjoint and run in parallel, Theorem~\ref{thm:methods-lrfc-implementation}, combined with the $\cO(k)$ repetitions, yields a total depth of 
	\begin{equation}
		\cO\left( k\log(k)(k\xi)^{1/\delta} \right) = \cO\left( k\log(k)[k\log(nk/\epsilon)]^{1/\delta} \right),
	\end{equation}
	using $\cO(nk)$ clean ancillas, proving Theorem~\ref{thm:finite-dimensional-design}.
	For the readout application, we set the design order to $k=3$ and take a constant error $\epsilon<1/4$, which gives the depth $\cO((\log n)^{1/\delta})$. Combined with Theorem~\ref{thm:approx-design-readout}, this proves the visible regime of Theorem~\ref{thm:depth-optimal-readout} for $\delta$D architectures.
	The one-dimensional and all-to-all implementations follow from Refs.~\cite{Schuster2024RandomUnitaries, Cui2025UnitaryDesignsOptimal}.

	\clearpage
	\onecolumngrid
	\appendix
	
	\begin{center}
		{\large \textbf{Supplementary Material}}
	\end{center}
	
	\section*{Contents} 
	
	\makeatletter
	\let\oldaddcontentsline\addcontentsline
	\renewcommand{\addcontentsline}[3]{%
		\def\target{#1}%
		\def\toclabel{toc}%
		\ifx\target\toclabel
		\oldaddcontentsline{atoc}{atoc#2}{#3}%
		\else
		\oldaddcontentsline{#1}{#2}{#3}%
		\fi
	}
	
	\@starttoc{atoc}
	\makeatother
	
	\vspace{1cm}
	
	\renewcommand{\thetheorem}{S\arabic{theorem}}
	\renewcommand{\thefact}{S\arabic{fact}}
	\renewcommand{\thelemma}{S\arabic{lemma}}
	\renewcommand{\thedefinition}{S\arabic{definition}}
	\renewcommand{\theproposition}{S\arabic{proposition}}
	\renewcommand{\thecorollary}{S\arabic{corollary}}
	\renewcommand{\theclaim}{S\arabic{claim}}
	\renewcommand{\thepage}{S\arabic{page}}
	\renewcommand{\thefigure}{S\arabic{figure}}
	
	\renewcommand{\theHtheorem}{S\arabic{theorem}}
	\renewcommand{\theHfact}{S\arabic{fact}}
	\renewcommand{\theHlemma}{S\arabic{lemma}}
	\renewcommand{\theHdefinition}{S\arabic{definition}}
	\renewcommand{\theHproposition}{S\arabic{proposition}}
	\renewcommand{\theHcorollary}{S\arabic{corollary}}
	\renewcommand{\theHclaim}{S\arabic{claim}}
	\renewcommand{\theHfigure}{S\arabic{figure}}
	
	\setcounter{theorem}{0}
	\setcounter{fact}{0}
	\setcounter{lemma}{0}
	\setcounter{equation}{0}
	\setcounter{definition}{0}
	\setcounter{proposition}{0}
	\setcounter{claim}{0}
	\setcounter{corollary}{0}
	\setcounter{figure}{0}
	\setcounter{page}{1}
	\setcounter{section}{0}
	\setcounter{equation}{0}
	
	\section{Phase hiding against low-depth readouts}\label{app:phase-hiding}
	This appendix proves the phase-hiding phenomenon. 
	In Sec.~\ref{subsec:unified_theorem}, we first present a unified phase-hiding theorem, Theorem~\ref{thm:support_family_hardness}, for general circuit families satisfying the support-family light-cone condition introduced in Definition~\ref{def:support_family_structure}. 
	In Sec.~\ref{subsec:unified_proof}, we then provide the proof of Theorem~\ref{thm:support_family_hardness}, based on two main ingredients: the light-cone structure of low-depth circuits and a random product-state code construction. 
	In Sec.~\ref{sec:application_unified_theorem}, we instantiate this unified theorem for $\delta$D architectures and all-to-all circuits, thereby deriving Theorem~\ref{thm:phase_hiding_low_depth-formal}.
	Finally, in Sec.~\ref{app:randomized_adaptive}, we show that the hiding bound applies under classical randomization and adaptive choices across different copies, and in Sec.~\ref{app:clean-ancilla}, we extend the argument to ancilla-assisted measurement circuits.
	
	\subsection{A unified phase-hiding theorem}
	\label{subsec:unified_theorem}
	We first present a unified phase-hiding theorem that captures the common structure underlying the low-depth measurements. 
	The key point is that, for every low-depth circuit, one can identify many qubits whose backward light cones are pairwise disjoint and each supported on a small set. 
	In fixed architectures, these supporting sets can often be chosen deterministically, whereas in more flexible architectures, such as all-to-all circuits, they may depend on the specific circuit. 
	The following definition abstracts this feature by allowing, for each circuit, the relevant light cones to be chosen from a finite family of admissible support patterns.
	
	\begin{definition}[Support-family light-cone structure]
		\label{def:support_family_structure}
		Let $\mathfrak F$ be a finite collection of ordered families
		\begin{equation}
			\mathcal S=(S_1,\dots,S_m)
		\end{equation}
		of pairwise disjoint nonempty subsets of $[n]$ satisfying $|S_j|\le L$ for all $j$. We say that a circuit class $\cU$ admits a support-family light-cone structure with parameters $(L,m,\mathfrak F)$ if, for every $U\in\cU$, there exist a family $(S_1(U),\dots,S_m(U))\in\mathfrak F$ and distinct output sites $i_1(U),\dots,i_m(U)$ such that
		\begin{equation}
			\supp\bigl(U^\dagger Z_{i_j(U)}U\bigr)\subseteq S_j(U),
			\qquad j\in[m].
		\end{equation}
	\end{definition}
	
	The next theorem shows that this structural condition alone already implies a uniform phase-hiding statement for the entire circuit class.
	
	\begin{theorem}[Unified support-family lower bound]
		\label{thm:support_family_hardness}
		There exist universal constants $C_0,C_1,c_1,c_2>0$ such that the following holds. Suppose that a circuit class $\cU$ admits a support-family light-cone structure with parameters $(L,m,\mathfrak F)$, and set $R:=|\mathfrak F|$. Define
		\begin{equation}
			\label{eq:ell-harder}
			\ell_{L,m,R}:=
			\left\lceil
			C_1\frac{\log(R+1)+m4^LL+1}{p_Lm}
			\right\rceil,
			\qquad
			p_L:=\frac{1}{16\cdot3^L}.
		\end{equation}
		Then, for sufficiently large $m$, there exist orthonormal states $\ket{\eta_0},\ket{\eta_1}\in\cH_2^{\otimes n}$ such that, for all $\theta, \theta' \in \bR$,
		\begin{equation}
			\label{eq:support-family-bound}
			\sup_{U\in\cU}
			\TVD{p_{\psi_\theta}^U}{p_{\psi_{\theta'}}^U}
			\le
			\exp\Bigl(C_1\log(\ell_{L,m,R}+1)-c_1m\exp(-C_0L)\Bigr)
			+
			\exp(-c_2m),
		\end{equation}
		where $\ket{\psi_\theta}=2^{-1/2}\left(\ket{\eta_0}+e^{i\theta}\ket{\eta_1}\right)$.
	\end{theorem}

	\subsection{Proof of Theorem~\ref{thm:support_family_hardness}}
	\label{subsec:unified_proof}
	We now prove Theorem~\ref{thm:support_family_hardness}. 
	The proof has three steps:
	\begin{enumerate}
		\item In Lemmas~\ref{lem:decomposition} and \ref{lem:rankone}, we rewrite a low-depth measurement as a sequence of commuting dephasing channels. 
		\item In Lemma~\ref{lem:local_random_product_contraction}, we show that a balanced dephasing acting on a small light-cone typically contracts the coherence between two random product states by a definite amount. 
		\item In Lemmas~\ref{lem:balanced_net} and \ref{lem:robust_product_code}, we construct a random product-state code for which this local contraction occurs simultaneously on many code states, uniformly over the whole circuit class. 
	\end{enumerate}
	
	We begin by showing that a computational-basis measurement after $U$ can be viewed as applying a sequence of commuting dephasing operations with the Heisenberg-evolved Pauli-$Z$ observables.
	\begin{lemma}[Measurement as successive dephasing]
		\label{lem:decomposition}
		Let $U$ be an $n$-qubit unitary. Define the computational-basis dephasing channel
		\begin{equation}
			\Delta(X):=\sum_{z\in\{0,1\}^n}\ketbra z X\ketbra z,
		\end{equation}
		and the corresponding measurement channel $\cD_U(X):=U^\dagger\Delta(UXU^\dagger)U$. Then for any density matrices $\rho$ and $\sigma$,
		\begin{equation}
			\label{eq:total_variational_identity}
			\TVD{p_\rho^U}{p_\sigma^U}
			=\frac12\|\cD_U(\rho-\sigma)\|_1.
		\end{equation}
		Moreover, let $A_i(U):=U^\dagger Z_iU$ and
		\begin{equation}
			\cE_i^U(X):=\frac12\bigl(X+A_i(U)XA_i(U)\bigr).
		\end{equation}
		Then
		\begin{equation}
			\label{eq:dephasing_decomposition}
			\cD_U=\cE_1^U\circ\cdots\circ\cE_n^U,
		\end{equation}
		and each $\cE_i^U$ is trace-norm contractive.
	\end{lemma}
	
	\begin{proof}
		Equation~\eqref{eq:total_variational_identity} is the standard identity between total variation distance and the trace norm after dephasing:
		\begin{equation}
			\TVD{p_\rho^U}{p_\sigma^U}
			=\frac12\sum_z\bigl|\bra zU(\rho-\sigma)U^\dagger\ket z\bigr|
			=\frac12\|\Delta(U(\rho-\sigma)U^\dagger)\|_1.
		\end{equation}
		Conjugating by $U$ gives~\eqref{eq:total_variational_identity}. For a single site $i$, the computational-basis dephasing is $\mathcal F_i(X)=\frac12(X+Z_iXZ_i)$. Since $Z_1,\ldots,Z_n$ commute, $\Delta=\mathcal F_1\circ\cdots\circ\mathcal F_n$. Conjugating by $U$ gives~\eqref{eq:dephasing_decomposition}. Finally, each $\cE_i^U$ is an average of two unitary conjugations and is therefore trace-norm contractive.
	\end{proof}
	
	The next lemma gives an exact expression for the contraction of a rank-one coherence term under a dephasing channel.
	\begin{lemma}[Rank-one dephasing]
		\label{lem:rankone}
		Let $A=P_+-P_-$ be a Hermitian unitary, where $P_\pm$ are the projectors onto its $\pm1$ eigenspaces. Let
		\begin{equation}
			\cE_A(X):=\frac12(X+AXA)=P_+XP_+ + P_-XP_-.
		\end{equation}
		Then, for any states $\ket u,\ket v$,
		\begin{equation}
			\|\cE_A(\ket u\bra v)\|_1
			=
			\|P_+\ket u\|\,\|P_+\ket v\|
			+
			\|P_-\ket u\|\,\|P_-\ket v\|.
		\end{equation}
	\end{lemma}
	
	\begin{proof}
		We have
		\begin{equation}
			\cE_A(\ket u\bra v)=P_+\ket u\bra vP_+ + P_-\ket u\bra vP_-.
		\end{equation}
		The two summands have orthogonal supports, so the trace norm is additive. Each summand is rank one, giving the stated formula.
	\end{proof}
	
	We now show that, for any balanced local dephasing, the coherence between a fixed product state and a random product state is contracted with non-negligible probability.
	\begin{lemma}[Random product states contract every balanced dephasing]
		\label{lem:local_random_product_contraction}
		Let $A=P_+-P_-$ be a Hermitian unitary on $t$ qubits with
		\begin{equation}
			\operatorname{rank}(P_+)=\operatorname{rank}(P_-)=2^{t-1}.
		\end{equation}
		Let $\ket u$ be any fixed product state on these $t$ qubits, and let $\ket v=\bigotimes_{a=1}^t\ket{v_a}$ be a Haar-random product state, with the single-qubit states $\ket{v_a}$ independent. Then
		\begin{equation}
			\label{eq:local-random-contraction}
			\Pr_v\left[
			\left\|\cE_A(\ket u\bra v)\right\|_1
			\le
			1-\frac{1}{32\cdot 3^t}
			\right]
			\ge
			\frac{1}{16\cdot 3^t}.
		\end{equation}
	\end{lemma}
	
	\begin{proof}
		Let $a_\pm:=\|P_\pm\ket u\|^2$, and choose $s\in\{+,-\}$ such that $a_s\le1/2$. Set $P:=P_s$ and $a:=a_s$. For a product state $\ket v$, define $\theta(v):=\bra vP\ket v$. By Lemma~\ref{lem:rankone},
		\begin{equation}
			\left\|\cE_A(\ket u\bra v)\right\|_1
			=\sqrt{a\theta(v)}+\sqrt{(1-a)(1-\theta(v))}.
		\end{equation}
		The right-hand side is the Bhattacharyya coefficient between the Bernoulli distributions $(a,1-a)$ and $(\theta(v),1-\theta(v))$. Since the squared Hellinger distance dominates half the squared total variation distance,
		\begin{equation}
			\label{eq:hellinger-lower}
			1-\left\|\cE_A(\ket u\bra v)\right\|_1
			\ge \frac12(\theta(v)-a)^2.
		\end{equation}
		Let $B:=P-aI$. Then $\theta(v)-a=\bra vB\ket v$. Put $D:=2^t$. Since $P$ has rank $D/2$,
		\begin{equation}
			\label{eq:B-hs-lower}
			\|B\|_2^2
			=\Tr\bigl((P-aI)^2\bigr)
			=D\left(a^2-a+\frac12\right)
			\ge \frac D4.
		\end{equation}
		Expand $B$ in the Pauli basis as
		\begin{equation}
			B=2^{-t}\sum_{Q\in\{I,X,Y,Z\}^{\otimes t}}\widehat B_Q Q,
			\qquad
			\widehat B_Q:=\Tr(BQ).
		\end{equation}
		For a Haar-random single-qubit state, non-identity Pauli expectations have mean zero and second moment $1/3$, and different Pauli directions are uncorrelated. Therefore
		\begin{equation}
			\begin{split}
				\mathbb E_v\bigl(\bra vB\ket v\bigr)^2
				&=4^{-t}\sum_Q \widehat B_Q^2\,3^{-\wt(Q)}  \\
				&\ge 4^{-t}3^{-t}\sum_Q\widehat B_Q^2
				=2^{-t}3^{-t}\|B\|_2^2
				\ge \frac{1}{4\cdot 3^t},
			\end{split}
		\end{equation}
		where the last inequality uses~\eqref{eq:B-hs-lower}. Set $\sigma_t^2:=1/(4\cdot3^t)$ and $\eta_t:=\sigma_t/2=1/(4\cdot3^{t/2})$. Since $|\bra vB\ket v|\le1$,
		\begin{equation}
			\Pr\bigl[|\bra vB\ket v|\ge\eta_t\bigr]
			\ge
			\frac{\sigma_t^2-\eta_t^2}{1-\eta_t^2}
			\ge
			\frac{1}{16\cdot3^t}.
		\end{equation}
		On this event,~\eqref{eq:hellinger-lower} gives
		\begin{equation}
			\left\|\cE_A(\ket u\bra v)\right\|_1
			\le 1-\frac12\eta_t^2
			=1-\frac{1}{32\cdot3^t}.
		\end{equation}
		This proves the lemma.
	\end{proof}
	
	To make our result uniform over all possible Heisenberg-evolved observables, we discretize this continuous set using an $\varepsilon$ net, a finite subset such that any observable is within an operator-norm distance $\varepsilon$ from at least one element in the net.
	
	\begin{lemma}[A net for balanced Hermitian unitaries]
		\label{lem:balanced_net}
		There is a universal constant $C_{\mathrm{net}}>0$ such that, for every $0<\varepsilon<1$, the set of balanced Hermitian unitaries on $L$ qubits admits an $\varepsilon$-net in operator norm of size at most
		\begin{equation}
			\label{eq:balanced-net-size}
			K_L(\varepsilon):=\exp\left(C_{\mathrm{net}}4^L\log\frac{C_{\mathrm{net}}}{\varepsilon}\right).
		\end{equation}
		
	\end{lemma}
	\begin{proof}
		Let $D:=2^L$. 
		Every balanced Hermitian unitary on $L$ qubits can be written as $A=VZ_1V^{\dagger}$ for some $V\in U(D)$, where $Z_1$ is the Pauli $Z$ operator on the first qubit. By the standard volumetric covering bound for the unitary group,
		$U(D)$ admits an $\eta$-net in operator norm of cardinality at most $\left(C/\eta\right)^{C D^2}$ for a universal constant $C>0$ (see, e.g., \cite[Theorem~7]{Szarek1998MetricEntropy}). If $\|V-W\|_\infty\le\eta$, then $\|VZ_1V^\dagger-WZ_1W^\dagger\|_\infty \le 2\eta$.
		Thus, taking $\eta=\varepsilon/2$, we obtain an $\varepsilon$-net for balanced Hermitian
		unitaries on $L$ qubits of size at most
		\begin{equation}
			\left(\frac{2C}{\varepsilon}\right)^{C4^L}
			\le
			\exp\left(C_{\mathrm{net}}4^L\log\frac{C_{\mathrm{net}}}{\varepsilon}\right),
		\end{equation}
		after enlarging the universal constant $C_{\mathrm{net}}$.
	\end{proof}

	We now combine the local contraction statement with the net argument and a probabilistic code construction. 
	The goal is to build many product states such that, for every admissible support family and every local dephasing on those supports, almost all pairs are contracted on a positive fraction of the blocks.
	
	\begin{lemma}[Robust random product code for a finite support family]
		\label{lem:robust_product_code}
		There exist universal constants $c_\star,C_\star,c_{\mathrm{ov}}>0$ such that the following holds. Let $L,m,\mathfrak F$ be as in Definition~\ref{def:support_family_structure}, and let $R:=|\mathfrak F|$. Define
		\begin{equation}
			p_L:=\frac{1}{16\cdot3^L},
			\qquad
			\gamma_L:=\frac{1}{32\cdot3^L},
			\qquad
			\delta_L:=\frac{p_L}{4},
			\qquad
			\lambda_L:=1-\frac{\gamma_L}{2},
		\end{equation}
		\begin{equation}
			\varepsilon_L:=\frac{\gamma_L}{4},
			\qquad
			K_L:=K_L(\varepsilon_L),
		\end{equation}
		where $K_L(\varepsilon_L)$ is the net size from Lemma~\ref{lem:balanced_net}. Let
		\begin{equation}
			\label{eq:ell-general}
			\ell:=\left\lceil
			C_\star\frac{\log(R+1)+m\log K_L+1}{p_Lm}
			\right\rceil,
			\qquad
			r:=2\ell^2.
		\end{equation}
		Then, for sufficiently large $m$, there exist product states
		\begin{equation}
			\ket{\Phi_s}=\bigotimes_{u=1}^n\ket{\phi_{s,u}},
			\qquad s=1,\dots,N,
		\end{equation}
		with
		\begin{equation}
			\label{eq:N-general}
			N:=2\left\lfloor\frac12\exp(c_\star p_L\gamma_Lm)\right\rfloor
		\end{equation}
		such that the following properties hold.
		\begin{enumerate}[label=\textup{(\alph*)}]
			\item For all distinct $s,t\in[N]$,
			\begin{equation}
				\label{eq:robust-code-overlap}
				|\langle\Phi_s|\Phi_t\rangle|\le\exp(-c_{\mathrm{ov}}n).
			\end{equation}
			\item For every $(S_1,\dots,S_m)\in\mathfrak F$ and every choice of balanced Hermitian unitaries $A_j$ acting on $\cH_2^{\otimes S_j}$, all but at most $r-1$ unordered pairs $\{s,t\}$ satisfy
			\begin{equation}
				\label{eq:robust-code-good}
				\left|
				\left\{
				j\in[m]:
				\left\|
				\cE_{A_j}\left(\ket{\Phi_s^{S_j}}\bra{\Phi_t^{S_j}}\right)
				\right\|_1
				\le \lambda_L
				\right\}
				\right|
				\ge \delta_Lm,
			\end{equation}
			where $\ket{\Phi_s^{S_j}}:=\bigotimes_{u\in S_j}\ket{\phi_{s,u}}$ is the restriction of $\ket{\Phi_s}$ to $S_j$.
		\end{enumerate}
	\end{lemma}
	
	\begin{proof}
		Sample the states $\ket{\Phi_1},\dots,\ket{\Phi_N}$ independently, each as a Haar-random product state over all $n$ qubits.
		
		For fixed $s\ne t$,
		\begin{equation}
			\mathbb E|\langle\Phi_s|\Phi_t\rangle|^2=2^{-n}.
		\end{equation}
		By Markov's inequality, for a sufficiently small universal $c_{\mathrm{ov}}>0$,
		\begin{equation}
			\Pr\left[|\langle\Phi_s|\Phi_t\rangle|>\exp(-c_{\mathrm{ov}}n)\right]
			\le \exp(-c n)
		\end{equation}
		for a universal $c>0$. Since $m\le n$ and $\log N\le c_\star m$, choosing $c_\star$ sufficiently small allows a union bound over all pairs, proving~\eqref{eq:robust-code-overlap} with probability at least $1-\exp(-cn)$ after adjusting the constant $c$.
		
		We next prove property~\textup{(b)}. Fix a support family $(S_1,\dots,S_m)\in\mathfrak F$, and first assume that each local Hermitian unitary $\widetilde A_j$ belongs to the $\varepsilon_L$-net from Lemma~\ref{lem:balanced_net} for the corresponding support size $|S_j|$. For an unordered pair $\{s,t\}$, define
		\begin{equation}
			H(s,t):=
			\left|
			\left\{
			j\in[m]:
			\left\|
			\cE_{\widetilde A_j}\left(\ket{\Phi_s^{S_j}}\bra{\Phi_t^{S_j}}\right)
			\right\|_1
			\le 1-\gamma_L
			\right\}
			\right|.
		\end{equation}
		Condition on $\Phi_s$. For each $j$, Lemma~\ref{lem:local_random_product_contraction}, applied with $t=|S_j|\le L$, $u=\Phi_s^{S_j}$, $v=\Phi_t^{S_j}$, and $A=\widetilde A_j$, implies
		\begin{equation}
			\Pr\left[
			\left\|
			\cE_{\widetilde A_j}\left(\ket{\Phi_s^{S_j}}\bra{\Phi_t^{S_j}}\right)
			\right\|_1
			\le 1-\gamma_L
			\;\middle|\;\Phi_s
			\right]
			\ge p_L.
		\end{equation}
		Because $S_1,\dots,S_m$ are pairwise disjoint, these events are conditionally independent over $j$. Thus $H(s,t)$ stochastically dominates $\Bin(m,p_L)$, and Chernoff's inequality gives
		\begin{equation}
			\label{eq:fixed-pair-bad-general}
			\Pr[H(s,t)<\delta_Lm]\le\exp(-c p_Lm)
		\end{equation}
		for a universal constant $c>0$.
		
		For the fixed support family and fixed net, define the bad-pair graph on the vertex set $[N]$ by connecting $\{s,t\}$ if $H(s,t)<\delta_Lm$. We use the elementary fact that every graph with at least $2\ell^2$ edges contains either a matching of size $\ell$ or a star with $\ell$ leaves. Indeed, if a graph has no such matching and no such star, then a maximal matching has fewer than $\ell$ edges, and every vertex has degree at most $\ell-1$; every edge meets the maximal matching, so the graph has fewer than $2\ell^2$ edges.
		
		For a fixed $\ell$-matching in the bad-pair graph, the corresponding bad events involve disjoint pairs of codewords and are independent. Hence, by~\eqref{eq:fixed-pair-bad-general}, the probability that all its edges are bad is at most $\exp(-c\ell p_Lm)$. For a fixed $\ell$-star, condition on the center codeword; the leaf codewords are then independent, and the same bound gives probability at most $\exp(-c\ell p_Lm)$.
		
		There are at most $R$ choices of the support family, at most $K_L^m$ choices of the nets for each support family, at most $N^{2\ell}$ possible $\ell$-matchings, and at most $N^{\ell+1}$ possible $\ell$-stars. Thus, the total failure probability for all support families and all nets is at most
		\begin{equation}
			\label{eq:union-robust-code}
			R K_L^m N^{2\ell}e^{-c\ell p_Lm}
			+
			R K_L^m N^{\ell+1}e^{-c\ell p_Lm}.
		\end{equation}
		Since $\log N\le c_\star p_L\gamma_Lm\le c_\star p_Lm$, choose $c_\star>0$ small enough so that the powers of $N$ are absorbed into the negative exponential. Then choose $C_\star$ large enough in~\eqref{eq:ell-general} so that~\eqref{eq:union-robust-code} is strictly smaller than $1$ for all sufficiently large $m$.
		
		It remains to pass from fixed net to arbitrary balanced Hermitian unitaries. Let $A_j$ be arbitrary and choose a net point $\widetilde A_j$ with $\|A_j-\widetilde A_j\|_\infty\le\varepsilon_L$. For any rank-one operator $X=\ket x\bra y$,
		\begin{equation}
			\begin{split}
				\|\cE_{A_j}(X)-\cE_{\widetilde A_j}(X)\|_1
				&=\frac12\|A_jXA_j-\widetilde A_jX\widetilde A_j\|_1 \\
				&\le \frac12\|(A_j-\widetilde A_j)XA_j\|_1
				+\frac12\|\widetilde A_jX(A_j-\widetilde A_j)\|_1 \\
				&\le \|A_j-\widetilde A_j\|_\infty
				\le \varepsilon_L.
			\end{split}
		\end{equation}
		Hence a block with $\|\cE_{\tilde{A}_j}\left(\ket{\Phi_s^{S_j}}\bra{\Phi_t^{S_j}}\right)\|_1 \le 1-\gamma_L$ satisfies 
		\begin{equation}
			\|\cE_{A_j}\left(\ket{\Phi_s^{S_j}}\bra{\Phi_t^{S_j}}\right)\|_1 \le 1-\gamma_L+\varepsilon_L
			=1-\frac{3\gamma_L}{4}
			\le 1-\frac{\gamma_L}{2}
			=\lambda_L.
		\end{equation}
		Therefore, the robust property for all nets implies~\eqref{eq:robust-code-good} for all balanced Hermitian unitaries. Together with the overlap property, this completes the proof.
	\end{proof}
	
	Having constructed a robust random product code whose pairwise coherence is contracted by low-depth measurements, we now use it to build phase-hiding states. The key idea is to form coherent superpositions of these random code states and show that the relative phases remain hidden from low-depth measurements, precisely because the corresponding coherence terms are strongly contracted under the measurement.
	
	\begin{proof}[Proof of Theorem~\ref{thm:support_family_hardness}]
		Apply Lemma~\ref{lem:robust_product_code} to $\mathfrak F$. Let $N$ and $\ket{\Phi_1},\dots,\ket{\Phi_N}$ be the resulting product states. Their Gram matrix $G_{st}:=\langle\Phi_s|\Phi_t\rangle$ satisfies
		\begin{equation}\label{eq:Gram_bound}
			\|G-I\|_\infty
			\le (N-1)e^{-c_{\mathrm{ov}}n}
			\le e^{-\beta m}
		\end{equation}
		for a universal $\beta>0$, after choosing $c_\star$ small enough, since $m\le n$.
		
		Split $[N]$ into two halves $A:=\{1,\dots,N/2\}$ and $B:=\{N/2+1,\dots,N\}$. Define
		\begin{equation}
			\ket{w_+}:=\sqrt{\frac2N}\sum_{s\in A}\ket{\Phi_s},
			\qquad
			\ket{w_-}:=\sqrt{\frac2N}\sum_{s\in B}\ket{\Phi_s}.
		\end{equation}
		Let $W:\cH_2\to\cH_2^{\otimes n}$ be given by $W\ket{e_0}=\ket{w_+}$ and $W\ket{e_1}=\ket{w_-}$, and let $M:=W^\dagger W$. The bound \eqref{eq:Gram_bound} implies $\|M-I_2\|_\infty\le e^{-\beta m}$ after decreasing $\beta$ if necessary. For sufficiently large $m$, $M^{-1/2}$ exists and $\|M^{-1/2}-I_2\|_\infty\le C e^{-\beta m}$.
		
		Set
		\begin{equation}
			\ket{\eta_0}:=WM^{-1/2}\ket{e_0},
			\qquad
			\ket{\eta_1}:=WM^{-1/2}\ket{e_1}.
		\end{equation}
		These states are orthonormal. For $\theta \in \bR$, let
		\begin{equation}
			\ket{v_\theta}:=\frac1{\sqrt2}(\ket{e_0}+e^{i\theta}\ket{e_1}),
			\qquad
			\ket{\psi_\theta}:=WM^{-1/2}\ket{v_\theta},
		\end{equation}
		and define the unorthogonalized phase state
		\begin{equation}
			\ket{w_\theta^\circ}:=W\ket{v_\theta}
			=\frac1{\sqrt N}\left(\sum_{s\in A}\ket{\Phi_s}+e^{i\theta}\sum_{s\in B}\ket{\Phi_s}\right).
		\end{equation}
		Writing $\rho_\theta:=\ket{\psi_\theta}\bra{\psi_\theta}$ and $\rho_\theta^\circ:=\ket{w_\theta^\circ}\bra{w_\theta^\circ}$, we have, for sufficiently large $m$,
		\begin{equation}
			\|W\|_\infty^2=\|M\|_\infty\le\frac32,
			\qquad
			\|M^{-1/2}\|_\infty\le\sqrt2.
		\end{equation}
		Therefore
		\begin{equation}
			\begin{split}
				\|\rho_\theta-\rho_\theta^\circ\|_1
				&\le \|W\|_\infty^2\,
				\left\|
				M^{-1/2}\ket{v_\theta}\bra{v_\theta}M^{-1/2}
				-\ket{v_\theta}\bra{v_\theta}
				\right\|_1 \\
				&\le C e^{-\beta m},
			\end{split}
			\label{eq:orthogonalization-error-unified}
		\end{equation}
		after increasing $C$ if necessary. Also set
		\begin{equation}
			\rho_M:=\frac1N\sum_{s=1}^N\ket{\Phi_s}\bra{\Phi_s}.
		\end{equation}
		
		Fix $U\in\cU$. By Definition~\ref{def:support_family_structure}, choose $(S_1(U),\dots,S_m(U))\in\mathfrak F$ and output sites $i_1(U),\dots,i_m(U)$ such that $\supp(U^\dagger Z_{i_j(U)}U)\subseteq S_j(U)$. Let $A_j(U)$ be the restriction of $U^\dagger Z_{i_j(U)}U$ to $S_j(U)$. This restriction is a balanced Hermitian unitary.
		Define the selected dephasing channel
		\begin{equation}
			P_U:=\cE_{A_1(U)}\circ\cdots\circ\cE_{A_m(U)},
		\end{equation}
		where each local dephasing is understood as acting on the corresponding subset $S_j(U)$. Since the supports $S_j(U)$ are pairwise disjoint, $P_U$ factorizes over these supports. Also, the observables $U^\dagger Z_iU$ commute with one another, so the dephasing maps in Lemma~\ref{lem:decomposition} commute. Hence, after moving the selected dephasings together and omitting the remaining trace-norm contractive dephasings,
		\begin{equation}
			\label{eq:selected-dephasing-contracts}
			\|\cD_U(X)\|_1\le \|P_U(X)\|_1
		\end{equation}
		for every operator $X$.
		
		For $s\ne t$, set $X_{st}:=\ket{\Phi_s}\bra{\Phi_t}$. By Lemma~\ref{lem:robust_product_code}(b), applied to the support family associated with $U$ and to the local phasing $A_j(U)$, all but at most $r-1$ unordered pairs $\{s,t\}$ satisfy
		\begin{equation}
			\left|
			\left\{
			j\in[m]:
			\left\|\cE_{A_j(U)}\left(\ket{\Phi_s^{S_j(U)}}\bra{\Phi_t^{S_j(U)}}\right)\right\|_1
			\le \lambda_L
			\right\}
			\right|
			\ge \delta_Lm.
		\end{equation}
		For every such pair, tensor factorization gives
		\begin{equation}
			\label{eq:offdiag-unified}
			\|\cD_U(X_{st})\|_1
			\le \|P_U(X_{st})\|_1
			\le \lambda_L^{\delta_Lm}.
		\end{equation}
		
		Define coefficients
		\begin{equation}
			\label{eq:alpha_coef}
			\alpha_s^{(\theta)}:=
			\begin{cases}
				1, & s\in A,\\
				e^{i\theta}, & s\in B.
			\end{cases}
		\end{equation}
		Then
		\begin{equation}
			\rho_\theta^\circ-\rho_M
			=\frac1N\sum_{s\ne t}\alpha_s^{(\theta)}\overline{\alpha_t^{(\theta)}}X_{st}.
		\end{equation}
		Using~\eqref{eq:offdiag-unified}, and noting that at most $2(r-1)$ ordered pairs correspond to exceptional unordered pairs, we obtain
		\begin{equation}
			\label{eq:mixed-comparison-unified}
			\|\cD_U(\rho_\theta^\circ-\rho_M)\|_1
			\le
			\frac{2(r-1)}{N}+N\lambda_L^{\delta_Lm}.
		\end{equation}
		
		We now estimate the right-hand side. Since $p_L\gamma_L=(512\cdot9^L)^{-1}\ge \exp(-C_0L)$ for a universal $C_0>0$ and $N$ is given by~\eqref{eq:N-general}, while $-\log\lambda_L\ge\gamma_L/2$ and $\delta_L=p_L/4$, choosing $c_\star$ small enough yields
		\begin{equation}
			N\lambda_L^{\delta_Lm}\le\exp\bigl(-c m\exp(-C_0L)\bigr).
		\end{equation}
		Moreover $r=2\ell^2$, with $\ell$ bounded by $\ell_{L,m,R}$ after increasing constants in~\eqref{eq:ell-harder}. Hence
		\begin{equation}
			\frac{2(r-1)}{N}
			\le
			\exp\Bigl(C\log(\ell_{L,m,R}+1)-c m\exp(-C_0L)\Bigr).
		\end{equation}
		Combining this estimate with~\eqref{eq:orthogonalization-error-unified},~\eqref{eq:mixed-comparison-unified}, and trace-norm contractivity gives
		\begin{equation}
			\|\cD_U(\rho_\theta-\rho_M)\|_1
			\le
			\exp\Bigl(C\log(\ell_{L,m,R}+1)-c m\exp(-C_0L)\Bigr)
			+
			\exp(-c'm).
		\end{equation}
		Finally, Lemma~\ref{lem:decomposition} gives
		\begin{equation}
			\TVD{p_{\psi_\theta}^U}{p_{\rho_M}^U}
			=\frac12\|\cD_U(\rho_\theta-\rho_M)\|_1.
		\end{equation}
		Applying the triangle inequality to $\theta$ and $\theta'$ proves~\eqref{eq:support-family-bound}, uniformly in $U$.
	\end{proof}
	
	\subsection{Proof of Theorem~\ref{thm:phase_hiding_low_depth-formal}}
	\label{sec:application_unified_theorem}
	We now specialize Theorem~\ref{thm:support_family_hardness} to the two circuit architectures appearing in Theorem~\ref{thm:phase_hiding_low_depth-formal}. 
	We first consider $\delta$D circuits, where the relevant light-cone blocks can be chosen deterministically from the underlying geometry.
	
	\begin{lemma}[Packing disjoint light cones in $\delta$D grids]
		\label{lem:Ddim_packing}
		Fix $\delta\ge 1$. There exist constants $a_\delta,b_\delta>0$, depending only on $\delta$, such
		that the following holds. Let $G$ be a $\delta$D grid.
		Then there exist pairwise disjoint subsets $B_1,\dots,B_m\subseteq[n]$ and
		sites $i_1,\dots,i_m\in[n]$, all independent of $U$, satisfying
		\begin{equation}
			|B_j|\le a_\delta(d+1)^\delta,
			\qquad
			m\ge b_\delta\frac{n}{(d+1)^\delta},
		\end{equation}
		and $\supp\bigl(U^\dagger Z_{i_j}U\bigr)\subseteq B_j$ for every $U\in\cU_{n,d}^G$ and every $j\in[m]$.
	\end{lemma}
	\begin{proof}
		Choose a maximal set of grid sites $i_1,\dots,i_m$ whose pairwise graph
		distances are strictly larger than $2d$. Let $B_j$ be the graph ball of radius
		$d$ centered at $i_j$. Then the sets $B_1,\dots,B_m$ are pairwise disjoint.
		
		Since each circuit layer can enlarge support by at most one graph step, we have $  \supp\bigl(U^\dagger Z_{i_j}U\bigr)\subseteq B_j$ for every $U\in\cU_{n,d}^G$.
		Moreover, a radius-$d$ ball in a $\delta$D grid has
		size at most $a_\delta(d+1)^\delta$ for some constant $a_\delta>0$ depending only on $\delta$, so $|B_j|\le a_\delta(d+1)^\delta$.
		
		By maximality, the radius-$2d$ balls centered at $i_1,\dots,i_m$ cover all $n$
		sites. Since each such ball has size at most $a_\delta'(2d+1)^\delta$ for some constant
		$a_\delta'>0$ depending only on $\delta$, we obtain $n\le m\,a_\delta'(2d+1)^\delta$.
		After adjusting constants, this gives $m\ge b_\delta\frac{n}{(d+1)^\delta}$.
	\end{proof}
	
	Applying Theorem~\ref{thm:support_family_hardness} to these fixed blocks gives the phase hiding results for $\delta$D architectures.
	
	\begin{corollary}[Phase hiding in $\delta$D architectures]
		\label{cor:main_hardness}
		Fix $\delta\ge 1$. There exist constants $C_{0,\delta}$,$C_{1,\delta}$,$c_{1,\delta}$,$c_{2,\delta}>0$,
		depending only on $\delta$, such that the following holds. Let $G$ be a
		$\delta$D grid. Then, for sufficiently large
		$n/(d+1)^\delta$, there exist orthonormal states
		$\ket{\eta_0},\ket{\eta_1}\in\cH_2^{\otimes n}$ such that, defining
		\begin{equation}
			\ket{\psi_\theta}
			:=
			\frac{1}{\sqrt2}\bigl(\ket{\eta_0}+e^{i\theta}\ket{\eta_1}\bigr),
		\end{equation}
		then, for all $\theta,\theta'\in\bR$,
		\begin{equation}
			\label{eq:Ddim-main-bound}
			\sup_{U\in\cU_{n,d}^G}
			\TVD{p_{\psi_\theta}^U}{p_{\psi_{\theta'}}^U}
			\le
			\exp\Bigl(
			C_{1,\delta}(d+1)^\delta
			-
			c_{1,\delta}\frac{n}{(d+1)^\delta}
			\exp\bigl(-C_{0,\delta}(d+1)^\delta\bigr)
			\Bigr)
			+
			\exp\Bigl(
			-c_{2,\delta}\frac{n}{(d+1)^\delta}
			\Bigr).
		\end{equation}
	\end{corollary}
	
	\begin{proof}
		Let
		\begin{equation}
			L_\delta:=\left\lceil a_\delta(d+1)^\delta\right\rceil,
			\qquad
			m_\delta:=\left\lfloor b_\delta\frac{n}{(d+1)^\delta}\right\rfloor.
		\end{equation}
		By Lemma~\ref{lem:Ddim_packing}, there exist pairwise disjoint blocks
		$B_1,\dots,B_m$ and sites $i_1,\dots,i_m$, independent of $U$, with
		$m\ge m_\delta$, $|B_j|\le L_\delta$, and $\supp\bigl(U^\dagger Z_{i_j}U\bigr)\subseteq B_j$ for every $U\in\cU_{n,d}^G$. 
		This shows that $\cU_{n,d}^G$ admits a support-family light-cone structure with parameters
		$(L_\delta,m_\delta,\mathfrak F)$, where
		\begin{equation}
			\mathfrak F=\{(B_1,\dots,B_{m_\delta})\}.
		\end{equation}
		Hence $R=1$, and
		\begin{equation}
			\ell_{L_\delta,m_\delta,1}
			\le
			\left\lceil C\frac{4^{L_\delta}L_\delta+1}{p_{L_\delta}}\right\rceil
			\le
			\exp(CL_\delta).
		\end{equation}
		Substituting this into Theorem~\ref{thm:support_family_hardness} and absorbing
		all constants depending only on $\delta$ proves~\eqref{eq:Ddim-main-bound}.
	\end{proof}
	
	We next apply the unified theorem to the all-to-all architecture. The key ingredient is the following greedy packing lemma.
	
	\begin{lemma}[Packing disjoint all-to-all light cones]
		\label{lem:a2a_packing}
		Let $U\in\cU_{n,d}^{\mathrm{a2a}}$, set $L:=2^d$, and define
		\begin{equation}
			S_i(U):=\supp\bigl(U^\dagger Z_iU\bigr),
			\qquad i\in[n].
		\end{equation}
		Then the following hold:
		\begin{enumerate}[label=\textup{(\alph*)}]
			\item $|S_i(U)|\le L$ for every $i\in[n]$;
			\item for every site $u\in[n]$, the number of indices $i\in[n]$ such that
			$u\in S_i(U)$ is at most $L$;
			\item consequently, there exist distinct indices $i_1,\dots,i_m$ with
			$m\ge\lfloor n/L^2\rfloor$ such that the supports
			$S_{i_1}(U),\dots,S_{i_m}(U)$ are pairwise disjoint.
		\end{enumerate}
	\end{lemma}
	
	\begin{proof}
		The first statement follows because the backward light cone of a single output qubit can at most double in size at each layer. The second statement is the corresponding forward light-cone bound: one input qubit can influence at most $2^d=L$ output qubits. 
		
		For the third statement, greedily select one remaining support $S_i(U)$ and delete all supports intersecting it. Each selected support has size at most $L$, and each input site belongs to at most $L$ supports, so each selection deletes at most $L^2$ supports. Starting from $n$ supports, this gives at least $\lfloor n/L^2\rfloor$ pairwise disjoint supports.
	\end{proof}
	
	\begin{corollary}[Phase hiding for all-to-all circuits]
		\label{thm:a2a_general}
		There exist universal constants $C_0,C_1,c_1,c_2>0$ such that the following
		holds. Let $L:=2^d$. Then, for sufficiently large $n/L^2$, there exist
		orthonormal states $\ket{\eta_0},\ket{\eta_1}\in\cH_2^{\otimes n}$
		such that, defining
		\begin{equation}
			\ket{\psi_\theta}
			:=
			\frac{1}{\sqrt2}\bigl(\ket{\eta_0}+e^{i\theta}\ket{\eta_1}\bigr),
		\end{equation}
		then, for all $\theta,\theta'\in\bR$,
		\begin{equation}
			\label{eq:a2a-main-bound}
			\sup_{U\in\cU_{n,d}^{\mathrm{a2a}}}
			\TVD{p_{\psi_\theta}^U}{p_{\psi_{\theta'}}^U}
			\le
			\exp\Bigl(
			C_1L+C_1\log\log(n+1)
			-c_1\frac{n}{L^2}\exp(-C_0L)
			\Bigr)
			+
			\exp\Bigl(-c_2\frac{n}{L^2}\Bigr).
		\end{equation}
	\end{corollary}

	\begin{proof}
		Let $m_0:=\lfloor n/L^2\rfloor$. Let $\mathfrak F_{\mathrm{a2a}}$ be the collection of all ordered families $(S_1,\dots,S_{m_0})$ of pairwise disjoint nonempty subsets of $[n]$ with $|S_j|\le L$. Its cardinality satisfies
		\begin{equation}
			R:=|\mathfrak F_{\mathrm{a2a}}|
			\le (m_0+1)^n
			\le (n+1)^n,
		\end{equation}
		because each site is either unused or assigned to one of the $m_0$ ordered subsets.
		
		By Lemma~\ref{lem:a2a_packing}, every $U\in\cU_{n,d}^{\mathrm{a2a}}$ has at least $m_0$ pairwise disjoint supports $\supp(U^\dagger Z_iU)$, each of size at most $L$. Keeping any ordered subfamily of exactly $m_0$ such supports shows that $\cU_{n,d}^{\mathrm{a2a}}$ admits a support-family light-cone structure with parameters $(L,m_0,\mathfrak F_{\mathrm{a2a}})$.
		
		It remains to estimate $\ell_{L,m_0,R}$. Since $m_0\ge c n/L^2$ whenever $n/L^2$ is sufficiently large,
		\begin{equation}
			\frac{\log(R+1)}{m_0}
			\le C L^2\log(n+1).
		\end{equation}
		Using $p_L^{-1}=16\cdot3^L$, we obtain
		\begin{equation}
			\ell_{L,m_0,R}
			\le
			\exp(CL)\bigl(L^2\log(n+1)+4^LL+1\bigr),
		\end{equation}
		and therefore, for sufficiently large $n$,
		\begin{equation}
			\log(\ell_{L,m_0,R}+1)
			\le C L+C\log\log(n+1).
		\end{equation}
		Substituting this estimate into Theorem~\ref{thm:support_family_hardness} gives~\eqref{eq:a2a-main-bound}, after adjusting universal constants.
	\end{proof}
	
	These two corollaries give the ancilla-free part of
	Theorem~\ref{thm:phase_hiding_low_depth-formal}. 
	\begin{proof}[Proof of Theorem~\ref{thm:phase_hiding_low_depth-formal}]
		We prove the statement for ancilla-free measurements. The ancilla-assisted extension is handled in Sec.~\ref{app:clean-ancilla}.
		
		If $G$ is a $\delta$D grid, Corollary~\ref{cor:main_hardness} yields
		Eq.~\eqref{eq:Ddim-main-bound}. Choose $\alpha_\delta>0$ sufficiently small so that,
		for all sufficiently large $n$, every
		$d\le \alpha_\delta(\log n)^{1/\delta}$ satisfies $C_{0,\delta}(d+1)^\delta \le \frac14 \log n$. Then
		\begin{equation}
			\frac{n}{(d+1)^\delta}\exp\bigl(-C_{0,\delta}(d+1)^\delta\bigr)
			\ge
			c\frac{n^{3/4}}{\log n},
		\end{equation}
		while $(d+1)^\delta=\cO(\log n)$. Hence, the right-hand side of
		Eq.~\eqref{eq:Ddim-main-bound} is at most $\exp[-n^{\Omega(1)}]$.
		
		If $G$ is all-to-all, Corollary~\ref{thm:a2a_general} yields
		Eq.~\eqref{eq:a2a-main-bound}. Choose $\alpha_{\mathrm{a2a}}>0$
		sufficiently small so that, for all sufficiently large $n$, every
		$d\le \alpha_{\mathrm{a2a}}\log\log n$ satisfies $L=2^d\le (\log n)^{1/2}$.
		Then, for sufficiently large $n$, $e^{-C_0L}\ge n^{-1/4}$, $\frac{n}{L^2}\ge c\frac{n}{\log n}$, and therefore
		\begin{equation}
			\frac{n}{L^2}e^{-C_0L}
			\ge
			c\frac{n^{3/4}}{\log n}.
		\end{equation}
		Since also $L+\log\log(n+1)=o\!\left(\frac{n^{3/4}}{\log n}\right)$, the right-hand side of Eq.~\eqref{eq:a2a-main-bound} is at most
		$\exp[-n^{\Omega(1)}]$.
	\end{proof}
	
	\subsection{Extension to randomized measurements and adaptivity across copies}
	\label{app:randomized_adaptive}
	
	We now show that our results apply to low-depth measurements when classical randomization of the circuit and adaptive choices across independent copies are allowed.
	
	First, consider a randomized single-copy measurement.  If a classical seed $R$, independent of the input
	state, selects a depth-$d$ circuit $U_R$ according to a distribution $\mu$, then the joint outcome distribution is $P_\theta^{R,Z}(r,z)=\mu(r)\,p_{\psi_\theta}^{U_r}(z)$.
	Therefore, for any $\theta, \theta' \in \bR$,
	\begin{equation}
		\label{eq:randomized_measurement_tv}
		\begin{split}
			\TVD{P_\theta^{R,Z}}{P_{\theta'}^{R,Z}} &=
			\mathbb E_R\,
			\TVD{p_{\psi_\theta}^{U_R}}{p_{\psi_{\theta'}}^{U_R}}      \\
			&\le
			\sup_{U\in\cU_{n,d}^G}
			\TVD{p_{\psi_\theta}^{U}}{p_{\psi_{\theta'}}^{U}}
			\le
			\varepsilon_{n,d}^G .
		\end{split}
	\end{equation}
	Thus, classical randomization does not increase the distinguishing power of
	low-depth measurements.
	
	A similar argument extends to adaptive protocols across multiple copies.  Suppose the
	protocol receives \(N\) copies and, in round \(k\), chooses a depth-\(d\) circuit as an
	arbitrary function of the previous transcript \(Y_{<k}\).  The choice may also be randomized,
	with the random seed included in the transcript.  For each fixed history \(y_{<k}\), the
	conditional one-copy experiment is still a randomized depth-\(d\) measurement, so
	Eq.~\eqref{eq:randomized_measurement_tv} gives
	\begin{equation}
		\TVD{P_\theta^{Y_k|Y_{<k}=y_{<k}}}{P_{\theta'}^{Y_k|Y_{<k}=y_{<k}}} \le  \varepsilon_{n,d}^G.
	\end{equation}
	By the standard chain-rule bound for total variation distance,
	\begin{equation}
		\TVD{P_\theta^{Y_{1:N}}}{P_{\theta'}^{Y_{1:N}}} \le \sum_{k=1}^N \sup_{y_{<k}}
		\TVD{P_\theta^{Y_k|Y_{<k}=y_{<k}}}{P_{\theta'}^{Y_k|Y_{<k}=y_{<k}}}
		\le N\varepsilon_{n,d}^G .
		\label{eq:adaptive_copy_bound}
	\end{equation}
	Therefore, following a derivation analogous to that of Eq.~\eqref{eq:cfi_phase_hiding}, even allowing classical randomization and full adaptivity across rounds, the total CFI obtainable from $N$ copies is at most $\cO(N\varepsilon_{n,d}^G)$.  Similarly, distinguishing two orthogonal phase states with constant error probability requires $N=\Omega\left((\varepsilon_{n,d}^G)^{-1}\right)$ copies for any single-copy schemes.
	
	\subsection{Extension to ancilla-assisted measurements}
	\label{app:clean-ancilla}
	
	We now explain that our proof also extends to low-depth measurements whose
	effective unitary on data qubits is implemented using ancillas, without requiring the
	ancillas to be returned to their initial state. Let $Q=[n]$ be the data
	register and let $A$ be an arbitrary ancilla register initialized in a
	fixed state $\tau_A$. The final ancilla state is not constrained.
	
	We say that a physical unitary $V$ on $Q\sqcup A$ is an ancilla-assisted implementation of an $n$-qubit unitary $U$ on the data
	register if, for every data input state $\rho$,
	\begin{equation}
		\tr_A\!\left[
		V(\rho\otimes \tau_A)V^\dagger
		\right]
		=
		U\rho U^\dagger .
		\label{eq:dirty_ancilla_equivalence}
	\end{equation}
	For any data input state $\rho$, applying $V$ to \(\rho\otimes \tau_A\), discarding the final ancillas, and measuring the data qubits in the computational basis gives the same output distribution as applying $U$ to $\rho$ and measuring the data qubits:
	\begin{align}
		p_{\rho\otimes \tau_A}^V(z)
		&:=
		\tr\!\left[
		\bigl(\ket{z}\bra{z}_Q\otimes I_A\bigr)
		V(\rho\otimes \tau_A)V^\dagger
		\right]  \notag\\
		&=
		\bra{z}U\rho U^\dagger\ket{z}
		=
		p_U^\rho(z).
		\label{eq:dirty_ancilla_same_distribution}
	\end{align}
	Thus, it remains only to check that the effective Heisenberg observables on the data register have the same light-cone support properties as in the no-ancilla case.
	
	\begin{lemma}[Effective light cone for ancilla-assisted implementations]
		\label{lem:effective-clean-observable}
		Let \(V\) be an ancilla-assisted implementation of $U$. For every
		data output qubit $i\in Q$, define $ O_i\coloneqq V^\dagger(Z_i\otimes I_A)V$
		as an operator on $Q\sqcup A$, and define $P_i\coloneqq U^\dagger Z_iU$ as an operator on $Q$. Then $\supp_Q(P_i)  \subseteq  \operatorname{supp}(O_i)\cap Q $.
	\end{lemma}
	
	\begin{proof}
		For any data input state \(\rho\), Eq.~\eqref{eq:dirty_ancilla_equivalence}
		implies
		\begin{equation}
			\begin{split}
				\tr_Q(\rho P_i) &= \tr_Q(Z_iU\rho U^\dagger)  \notag\\
				&=  \tr_{QA}\!\left[
				(Z_i\otimes I_A)V(\rho\otimes \tau_A)V^\dagger
				\right]  \notag\\
				&=  \tr_{QA}\!\left[
				(\rho\otimes \tau_A)O_i
				\right].
			\end{split}
		\end{equation}
		Equivalently, 
		\begin{equation}\label{eq:equivalence_Pi_Oi}
			P_i =  \tr_A\left[(I_Q\otimes \tau_A)O_i \right].
		\end{equation}
		We can write 
		\begin{equation}
			O_i = O_{S_iT_i} \otimes I_{Q\setminus S_i} \otimes I_{A\setminus T_i},
		\end{equation}
		where $S_i:=\supp(O_i) \cap Q$, $T_i:=\supp(O_i) \cap A$.
		Substituting Eq.~\eqref{eq:equivalence_Pi_Oi},  
		\begin{equation}
			P_i = \tr_A\left[(O_{S_iT_i} \otimes I_{A\setminus T_i}) (I_{S_i} \otimes \tau_A)\right] \otimes I_{Q\setminus S_i}.
		\end{equation}
		Thus, $\supp_Q(P_i) \subseteq S_i$.
	\end{proof}
	
	We now discuss the consequences for the all-to-all architecture. Let
	$\cU^{\mathrm{a2a,anc}}_{n,d}$ denote the class of induced data
	unitaries \(U\) that admit a depth-$d$ ancilla-assisted implementation
	$V$, where $V$ is an all-to-all circuit on $Q\sqcup A$ whose layers
	consists of one- and two-qubit gates with disjoint supports. Set $L:=2^d$.
	For every data output qubit $i\in Q$, the physical backward light cone of $O_i=V^\dagger(Z_i\otimes I_A)V$ 
	has size at most $L$. By Lemma~\ref{lem:effective-clean-observable}, $\operatorname{supp}_Q(P_i) \subseteq \supp(O_i) \cap Q$ for $P = U^{\dagger}Z_i U$, and hence
	\begin{equation}
		\left|\operatorname{supp}_Q(P_i)\right|
		\le L.
	\end{equation}
	The corresponding forward light-cone bound also remains true: for every
	data input qubit $u\in Q$, the number of data output qubits $i\in Q$
	such that
	\begin{equation}
		u\in \supp_Q(P_i)
	\end{equation}
	is at most \(L\), because such an \(i\) must lie in the physical forward
	light cone of \(u\) under \(V\). Therefore, the greedy packing argument in
	Lemma~\ref{lem:a2a_packing} applies verbatim and produces at least $\left\lfloor \frac{n}{L^2}\right\rfloor$ pairwise disjoint data supports $\operatorname{supp}_Q(U^\dagger Z_iU)$.
	Thus, the family $\cU^{\mathrm{a2a,anc}}_{n,d}$ admits the same
	support-family light-cone structure as in the no-ancilla all-to-all case.
	Consequently, Theorem~\ref{thm:phase_hiding_low_depth-formal} holds with the
	supremum over $U\in\cU_{n,d}^{\mathrm{a2a}}$ extended to $\cU^{\mathrm{a2a,anc}}_{n,d}$.
	
	The same argument also applies to fixed light-cone architectures. 
	As the data qubits $Q$ are embedded in a $\delta$D grid $G$, the light-cone $P_i$ for a data qubit $i\in Q$ is confined to a radius-$d$ ball centered at site $i$. This ball contains at most $a_{\delta}(d+1)^{\delta}$ data qubits, and each radius-$2d$ ball has a size of at most $a'_{\delta}(d+1)^{\delta}$. Consequently, the class $\cU^{G,\mathrm{anc}}_{n,d}$ satisfies the same light-cone condition as in Lemma~\ref{lem:Ddim_packing}, and Theorem~\ref{thm:phase_hiding_low_depth-formal} applies without modification.
	
	\section{Preliminaries on random unitaries}
	This appendix introduces the basics of random unitary integrals. 
	Let $\cU=\{q_a,U_a\}_{a\in\mathsf A}$ be an ensemble of unitaries on a $D$-dimensional Hilbert space.  Its $k$-th unitary moment channel is
	\begin{equation}
		\cT_{\cU}^{(k)}(X)
		:=
		\sum_{a\in\mathsf A}q_a\,U_a^{\otimes k}X(U_a^{\dagger})^{\otimes k} .
	\end{equation}
	The Haar moment channel is denoted by $\cT_{\mathrm{Haar}}^{(k)}$.
	
	\begin{definition}[Multiplicative-error unitary design]
		\label{def:app_relative_design}
		The ensemble $\cU$ is a multiplicative-$\epsilon$ approximate unitary $k$-design if
		\begin{equation}
			(1-\epsilon)\cT_{\mathrm{Haar}}^{(k)}
			\preceq
			\cT_{\cU}^{(k)}
			\preceq
			(1+\epsilon)\cT_{\mathrm{Haar}}^{(k)},
			\label{eq:app_relative_design_definition}
		\end{equation}
		where $\preceq$ denotes the completely-positive order.  Equivalently, both differences in Eq.~\eqref{eq:app_relative_design_definition} are completely positive maps.
	\end{definition}
	
	Sampling $a$ with probability $q_a$, applying $U_a$, and measuring in the computational basis gives a POVM
	\begin{equation}
		M_{a,y}=q_a\,U_a^\dagger\ketbra y U_a,
		\qquad a\in\mathsf A,\quad y\in[D].
	\end{equation}
	We combine $(a,y)$ into a single outcome label $x$ and write
	\begin{equation}
		M_x
		=
		D w_x\ketbra{\phi_x},
		\qquad
		\ket{\phi_x}:=U_a^\dagger\ket y,
		\qquad
		w_x:=\frac{q_a}{D}.
		\label{eq:induced_measurement}
	\end{equation}
	Then $\sum_x w_x\ketbra{\phi_x}=I/D$.  The higher-order moment is
	\begin{equation}
		\Phi_M^{(j)}
		:=
		\sum_x w_x\ketbra{\phi_x}^{\otimes j}
		=
		\frac1D\sum_a q_a\sum_{y=1}^D
		\left(U_a^\dagger\ketbra y U_a\right)^{\otimes j}.
		\label{eq:induced_projective_moment}
	\end{equation}
	The Haar moment is
	\begin{equation}
		\Phi_{\mathrm{Haar}}^{(j)}
		:=
		\int \ketbra{\phi}^{\otimes j}\,d\mu_{\mathrm{Haar}}(\phi)
		=
		\frac{\Pi_{\mathrm{sym}}^{(j)}}{\binom{D+j-1}{j}},
		\label{eq:haar_projective_moment}
	\end{equation}
	where $\Pi_{\mathrm{sym}}^{(j)}=\frac1{j!}\sum_{\pi\in S_j}P_\pi$ is the projector onto the fully symmetric subspace of $\cH_D^{\otimes j}$.
	
	\begin{lemma}
		\label{lem:relative_design_induces_projective_design}
		If $\cU$ is a multiplicative-$\epsilon$ approximate unitary $k$-design, then for every $1\le j\le k$,
		\begin{equation}
			(1-\epsilon)\Phi_{\mathrm{Haar}}^{(j)}
			\preceq
			\Phi_M^{(j)}
			\preceq
			(1+\epsilon)\Phi_{\mathrm{Haar}}^{(j)}.
			\label{eq:induced_multiplicative_design}
		\end{equation}
		Consequently, for $\Delta^{(j)}:=\Phi_M^{(j)}-\Phi_{\mathrm{Haar}}^{(j)}$,
		\begin{equation}
			-\epsilon\Phi_{\mathrm{Haar}}^{(j)}
			\preceq
			\Delta^{(j)}
			\preceq
			\epsilon\Phi_{\mathrm{Haar}}^{(j)}.
			\label{eq:induced_multiplicative_design_deriviation}
		\end{equation}
	\end{lemma}
	
	\begin{proof}
		The inverse ensemble $\cU^{-1}=\{q_a,U_a^\dagger\}$ is also a multiplicative-$\epsilon$ approximate unitary $k$-design, because the Haar measure is invariant under inversion.  For $j\le k$, apply the completely-positive order in Eq.~\eqref{eq:app_relative_design_definition} to the $j$-th moment of $\cU^{-1}$ and to the positive operators $\ketbra y^{\otimes j}$.  Averaging the resulting inequalities over $y\in[D]$ with weight $1/D$ gives Eq.~\eqref{eq:induced_multiplicative_design}.  For the Haar ensemble, the vectors $U^\dagger\ket y$ are Haar-random pure states, so the averaged Haar expression is exactly Eq.~\eqref{eq:haar_projective_moment}.  Subtracting $\Phi_{\mathrm{Haar}}^{(j)}$ gives Eq.~\eqref{eq:induced_multiplicative_design_deriviation}.
	\end{proof}

	\section{Constant-fraction readout via approximate unitary designs}\label{app:metrology}
	
	This appendix details the proof of Theorem~\ref{thm:approx-design-readout} with three main steps.
	We first reduce the multiparameter family to a one-parameter family along an arbitrary direction $\bu$. 
	We then bound the resulting CFI in terms of the second and third moments of the measurement ensemble. 
	Finally, we compare these moments with the corresponding Haar moments and use the approximate $3$-design property to establish the results.

	\begin{proof}[Proof of Theorem~\ref{thm:approx-design-readout}]
		Fix a smooth pure-state encoding
		\begin{equation}
			\cE=\{\rho_{\btheta}=\ketbra{\psi_{\btheta}}\}_{\btheta\in\Theta},
		\end{equation}
		a point $\btheta\in\Theta$, and a direction $\bu\in\mathbb R^m$.
		Consider the restricted one-parameter family
		\begin{equation}
			\rho_t \coloneqq \rho_{\btheta+t\bu}
			= \ketbra{\psi_t},
			\qquad
			\ket{\psi_t}\coloneqq \ket{\psi_{\btheta+t\bu}},
		\end{equation}
		and evaluate all derivatives at $t=0$.
		
		Let $I_t$ and $J_t$ denote the CFI and QFI of the restricted family.
		By the chain rule for Fisher information,
		\begin{equation}
			I_t=\bu^{\mathsf T} I_{\cE}^M(\btheta)\bu,
			\qquad
			J_t=\bu^{\mathsf T} J_{\cE}(\btheta)\bu.
		\end{equation}
		Therefore, it suffices to prove that 
		\begin{equation}\label{eq:app_multi_one_param_goal}
			I_t\ge \kappa_D(\epsilon) J_t
		\end{equation}
		for every $\bu\in\mathbb R^m$.
		If $J_t=0$, then Eq.~\eqref{eq:app_multi_one_param_goal} is trivial, so we may assume $J_t>0$.
		
		We first describe the symmetric logarithmic derivative (SLD) of the restricted model.
		Let $\ket{\dot\psi_t}\coloneqq \partial_t\ket{\psi_t}$. 
		Since $\braket{\psi_t}{\psi_t}=1$, $\Re\braket{\psi_t}{\dot\psi_t}=0$.
		Write
		\begin{equation}
			\ket{\dot\psi_t}
			=
			\alpha_t\ket{\psi_t}
			+
			\beta_t\ket{\psi_t^\perp},
			\qquad
			\braket{\psi_t}{\psi_t^\perp}=0,
			\qquad
			\beta_t\ge 0,
		\end{equation}
		for some unit vector $\ket{\psi_t^\perp}$.
		Then
		\begin{equation}
			J_t
			=
			4\Bigl(\braket{\dot\psi_t}{\dot\psi_t}
			-
			|\braket{\psi_t}{\dot\psi_t}|^2\Bigr)
			=
			4\beta_t^2.
			\label{eq:app_multi_qfi}
		\end{equation}
		The corresponding SLD is
		\begin{equation}
			L_t
			\coloneqq
			2\left(
			\ketbra{\dot\psi_t}{\psi_t}
			+
			\ketbra{\psi_t}{\dot\psi_t}
			\right)
			=
			2\beta_t
			\left(
			\ketbra{\psi_t^\perp}{\psi_t}
			+
			\ketbra{\psi_t}{\psi_t^\perp}
			\right),
			\label{eq:app_multi_sld}
		\end{equation}
		where the $\alpha_t$ term cancels because $\Re(\alpha_t)=0$.
		One readily checks that
		\begin{equation}
			\partial_t\rho_t=\frac12(L_t\rho_t+\rho_tL_t),
		\end{equation}
		so $L_t$ is indeed the SLD of the one-parameter family.
		We will use the identities
		\begin{equation}
			\tr(L_t)=0,
			\qquad
			\tr(L_t^2)=2J_t,
			\qquad
			\tr(\rho_tL_t^2)= \tr(L_t\rho_tL_t)=J_t.
			\label{eq:app_multi_sld_identities}
		\end{equation}
		
		The outcome probabilities for the induced POVM $M$ defined in Eq.~\eqref{eq:induced_measurement} are given by
		\begin{equation}
			p_t(x)=\Tr(M_x\rho_t)=D\,w_x\bra{\phi_x}\rho_t\ket{\phi_x}.
		\end{equation}
		Using the SLD equation~\eqref{eq:app_multi_sld},
		\begin{equation}
			\partial_t p_t(x)
			=
			D\,w_x\bra{\phi_x}\partial_t\rho_t\ket{\phi_x}
			=
			\frac{D\,w_x}{2}\bra{\phi_x}L_t\ket{\phi_x}.
			\label{eq:app_multi_pt_derivative}
		\end{equation}
		Hence
		\begin{equation}
			I_t
			=
			\sum_{x:p_t(x)>0}
			\frac{(\partial_t p_t(x))^2}{p_t(x)}
			=
			\frac{D}{4}
			\sum_{x:p_t(x)>0}
			w_x
			\frac{\bra{\phi_x}L_t\ket{\phi_x}^2}
			{\bra{\phi_x}\rho_t\ket{\phi_x}}.
			\label{eq:app_multi_cfi_expansion}
		\end{equation}
		Since $\bra{\phi_x}\rho_t\ket{\phi_x}=0$ implies $\bra{\phi_x}L_t\ket{\phi_x}=0$, a Cauchy-Schwarz inequality yields
		\begin{equation}
			\left(
			\sum_{x:p_t(x)>0} w_x
			\frac{\bra{\phi_x}L_t\ket{\phi_x}^2}
			{\bra{\phi_x}\rho_t\ket{\phi_x}}
			\right)
			\left(
			\sum_x w_x
			\bra{\phi_x}L_t\ket{\phi_x}^2\,
			\bra{\phi_x}\rho_t\ket{\phi_x}
			\right)
			\ge
			\left(
			\sum_x w_x\bra{\phi_x}L_t\ket{\phi_x}^2
			\right)^2.
		\end{equation}
		Therefore,
		\begin{equation}
			I_t
			\ge
			\frac{D}{4}
			\frac{\left(\sum_x w_x\bra{\phi_x}L_t\ket{\phi_x}^2\right)^2}
			{\sum_x w_x\bra{\phi_x}L_t\ket{\phi_x}^2\,\bra{\phi_x}\rho_t\ket{\phi_x}}.
			\label{eq:cfi_cs_bound}
		\end{equation}
		
		To estimate the numerator and denominator, we use the moment operators defined
		in Eq.~\eqref{eq:induced_projective_moment}. They give
		\begin{align}
			\sum_x w_x\bra{\phi_x}L_t\ket{\phi_x}^2
			&=
			\tr\!\left(\Phi_M^{(2)} L_t^{\otimes 2}\right),\\
			\sum_x w_x\bra{\phi_x}L_t\ket{\phi_x}^2\,\bra{\phi_x}\rho_t\ket{\phi_x}
			&=
			\tr\!\left(\Phi_M^{(3)}(L_t^{\otimes 2}\otimes \rho_t)\right).
		\end{align}
		
		For the Haar moments, we use
		\begin{equation}
			\Phi_{\mathrm{Haar}}^{(2)}=\frac{I+F}{D(D+1)},
			\qquad
			\Phi_{\mathrm{Haar}}^{(3)}=\frac{1}{D(D+1)(D+2)}\sum_{\pi\in S_3} P_\pi,
		\end{equation}
		where $F$ is the swap operator on two copies, and $P_\pi$ is the permutation operator corresponding to $\pi\in S_3$.
		Using Eq.~\eqref{eq:app_multi_sld_identities}, we obtain
		\begin{equation}
			\tr\!\left(\Phi_{\mathrm{Haar}}^{(2)}L_t^{\otimes 2}\right)
			=
			\frac{\tr(L_t^2)+\tr(L_t)^2}{D(D+1)}
			=
			\frac{2J_t}{D(D+1)},
			\label{eq:haar_second}
		\end{equation}
		and
		\begin{equation}
			\tr\!\left(\Phi_{\mathrm{Haar}}^{(3)}(L_t^{\otimes 2}\otimes \rho_t)\right)
			=
			\frac{\tr(L_t^2)\tr(\rho_t)+\tr(\rho_t L_t^2)+\tr(L_t\rho_t L_t)}{D(D+1)(D+2)}
			=
			\frac{4J_t}{D(D+1)(D+2)}.
			\label{eq:haar_third}
		\end{equation}

		Let $\Delta^{(k)} \coloneqq \Phi_M^{(k)}-\Phi_{\mathrm{Haar}}^{(k)}$. 
		Since $M$ is induced by a multiplicative-$\epsilon$-approximate unitary $3$-design, we can apply Lemma~\ref{lem:relative_design_induces_projective_design} to obtain
		\begin{equation}
			\|\Delta^{(2)}\|_\infty
			\le
			\epsilon \|\Phi_{\mathrm{Haar}}^{(2)}\|_\infty
			=
			\frac{2\epsilon}{D(D+1)},
			\qquad
			\|\Delta^{(3)}\|_\infty
			\le
			\epsilon \|\Phi_{\mathrm{Haar}}^{(3)}\|_\infty
			=
			\frac{6\epsilon}{D(D+1)(D+2)}.
			\label{eq:delta_norm_bound}
		\end{equation}
		Also, from Eq.~\eqref{eq:app_multi_sld},
		\begin{equation}
			\|L_t\|_1 = 2\sqrt{J_t},
			\qquad
			\|L_t^{\otimes 2}\|_1 = 4J_t,
			\qquad
			\|L_t^{\otimes 2}\otimes \rho_t\|_1 = 4J_t.
		\end{equation}
		Therefore,
		\begin{align}
			\tr\!\left(\Phi_M^{(2)}L_t^{\otimes 2}\right)
			&=
			\tr\!\left(\Phi_{\mathrm{Haar}}^{(2)}L_t^{\otimes 2}\right)
			+
			\tr\!\left(\Delta^{(2)}L_t^{\otimes 2}\right) \notag\\
			&\ge
			\frac{2J_t}{D(D+1)}
			-
			\|\Delta^{(2)}\|_\infty\,\|L_t^{\otimes 2}\|_1 \notag\\
			&\ge
			\frac{2(1-4\epsilon)J_t}{D(D+1)},
			\label{eq:numerator_bound}
		\end{align}
		and
		\begin{align}
			\tr\!\left(\Phi_M^{(3)}(L_t^{\otimes 2}\otimes \rho_t)\right)
			&=
			\tr\!\left(\Phi_{\mathrm{Haar}}^{(3)}(L_t^{\otimes 2}\otimes \rho_t)\right)
			+
			\tr\!\left(\Delta^{(3)}(L_t^{\otimes 2}\otimes \rho_t)\right) \notag\\
			&\le
			\frac{4J_t}{D(D+1)(D+2)}
			+
			\|\Delta^{(3)}\|_\infty\,\|L_t^{\otimes 2}\otimes \rho_t\|_1 \notag\\
			&\le
			\frac{4(1+6\epsilon)J_t}{D(D+1)(D+2)}.
			\label{eq:denominator_bound}
		\end{align}
		
		Substituting Eqs.~\eqref{eq:numerator_bound} and \eqref{eq:denominator_bound} into Eq.~\eqref{eq:cfi_cs_bound}, we obtain
		\begin{align}
			I_t
			&\ge
			\frac{D}{4}
			\frac{\left(\dfrac{2(1-4\epsilon)J_t}{D(D+1)}\right)^2}
			{\dfrac{4(1+6\epsilon)J_t}{D(D+1)(D+2)}} \notag\\
			&=
			\frac{(1-4\epsilon)^2}{1+6\epsilon}\frac{D+2}{4(D+1)}\,J_t \notag\\
			&=
			\kappa_D(\epsilon)\,J_t.
		\end{align}
		Since this holds for every $\bu\in\mathbb R^m$, we conclude that
		\begin{equation}
			I_{\cE}^M(\btheta)\succeq \kappa_D(\epsilon)\,J_{\cE}(\btheta).
		\end{equation}
	\end{proof}
	
	\section{Routing and arithmetic primitives with $\delta$D implementations}
	\label{app:finite-dimensional-arithmetic}
	This section presents reversible primitives with $\delta$D implementations, which are used later for the LRFC construction of approximate unitary design.  The common target is diameter-time implementation: a computation on $s$ qubits, or on $s$ constant-size registers, should use $\cO_\delta(s)$ space and depth $\cO_\delta(s^{1/\delta})$ on a $\delta$D grid.  
	Here, the notation $\cO_\delta$ means that hidden constant factors in the asymptotic scaling depend on $\delta$.
	All circuits below are clean, meaning that every ancilla register is returned to $0$.
	
	The organization is as follows.  
	Sec.~\ref{subsec:finite-dimensional-routing} presents the routing primitive that enables the realization of permutations on a $\delta$D grid.  
	Sec.~\ref{subsec:coefficientwise-linear-maps} gives the implementations of the coefficientwise maps and linear summations.  
	Sec.~\ref{subsec:polynomial-multiplication} and~\ref{subsec:finite-field-multiplication} then give polynomial and finite-field multiplication in diameter time $\cO_\delta(s^{1/\delta})$.
	
	\subsection{Routing}
	\label{subsec:finite-dimensional-routing}
	
	We first present a basic routing primitive for moving registers in a $\delta$D grid.  
	This primitive will be used to implement register permutations, in which we route the relevant registers to neighboring locations, apply the desired local gates, and then route them back. 
	
	\begin{definition}[$\delta$D grid]
		Let $P_r$ be the path graph on $[r]=\{1,\ldots,r\}$.  The Cartesian product $G\square H$ has vertex set $V(G)\times V(H)$, with $(g,h)\sim(g',h')$ if either $g=g'$ and $h\sim h'$, or $h=h'$ and $g\sim g'$.  The $\delta$D grid of side length $r$ is
		\begin{equation}
			G_{\delta,r}:=P_r^{\square \delta}.
		\end{equation}
		Equivalently, $G_{\delta,r}$ has vertex set $[r]^\delta$, with nearest neighbors differing by $1$ in exactly one coordinate.
	\end{definition}
	
	\begin{definition}[Routing number]
		Let $G=(V,E)$ be a graph.  A matching $M\subseteq E$ specifies a SWAP layer on $G$ which applies SWAP gates simultaneously on all edges in $M$.
		
		For a permutation $\pi:V\to V$, a sequence of SWAP layers implements $\pi$ if the quantum register initially placed at $v$ is moved to $\pi(v)$ for every $v\in V$.  The routing number $\operatorname{rt}(G)$ is the maximum, over all permutations $\pi$, of the minimum SWAP depth needed to implement $\pi$.
	\end{definition}
	
	\begin{lemma}[Permutation routing on $\delta$D grids]
		\label{lem:grid_permutation_routing_app}
		For every $\delta\ge1$,
		\begin{equation}
			\operatorname{rt}(G_{\delta,r})=\cO_\delta(r).
		\end{equation}
		Equivalently, any permutation of quantum registers on a $\delta$D grid of side length $r$ can be implemented by a nearest-neighbor SWAP circuit of depth $\cO_\delta(r)$.
	\end{lemma}
	
	\begin{proof}
		Ref.~\cite{Alon1993RoutingPermutation} proves the Cartesian-product bound
		\begin{equation}
			\operatorname{rt}(G\square H)
			\le
			\min\{2\operatorname{rt}(G)+\operatorname{rt}(H),\,
			2\operatorname{rt}(H)+\operatorname{rt}(G)\}.
			\label{eq:routing_cartesian_product}
		\end{equation}
		The path graph satisfies $\operatorname{rt}(P_r)=\cO(r)$ by odd-even transposition routing.  Iterating Eq.~\eqref{eq:routing_cartesian_product} over $G_{\delta,r}=P_r^{\square\delta}$ gives $\operatorname{rt}(G_{\delta,r})=\cO_\delta(r)$.
	\end{proof}
	
	A direct consequence is that any layer of all-to-all two-qubit gates can be realized in a $\delta$D grid at diameter cost. 
	
	\begin{lemma}[Compiling one all-to-all layer into a $\delta$D grid]
		\label{lem:mesh_compilation_app}
		Fix $\delta\ge1$.  Let $C$ be a $q$-qubit all-to-all circuit of depth $T$, where every layer consists of pairwise-disjoint two-qubit gates.  Then $C$ can be implemented on a $\delta$D grid using $\cO(q)$ qubits and depth
		\begin{equation}
			\cO_\delta\left(Tq^{1/\delta}\right).
		\end{equation}
	\end{lemma}
	
	\begin{proof}
		Place the $q$ logical qubits inside a $\delta$D grid of side length $r=\Theta(q^{1/\delta})$, adding $\cO(q)$ idle dummy qubits if necessary.
		
		It suffices to simulate one all-to-all layer.  Let the layer contain pairwise-disjoint two-qubit gates $G_1,\ldots,G_s$, with $s\le q/2$.  Assign the two input qubits of $G_j$ to neighboring sites, and assign all remaining qubits arbitrarily to the remaining sites.  This assignment is a permutation of the grid sites.
		
		By Lemma~\ref{lem:grid_permutation_routing_app}, this permutation is implemented by $\delta$D nearest-neighbor SWAP gates in depth $\cO_\delta(r)=\cO_\delta(q^{1/\delta})$.  
		The gates $G_1,\ldots,G_s$ can then be applied simultaneously in one nearest-neighbor two-qubit gate layer. 
		Finally, route the inverse permutation to return all qubits to their original locations.  Thus one all-to-all layer costs $\cO_\delta(q^{1/\delta})$ nearest-neighbor depth, and $T$ layers cost $\cO_\delta(Tq^{1/\delta})$.
	\end{proof}
	
	\subsection{Coefficientwise linear maps}
	\label{subsec:coefficientwise-linear-maps}
	We next present two elementary primitives for manipulating arrays of constant-size registers.  
	The first primitive implements constant-width local updates at every coefficient position.  
	
	\begin{lemma}[Constant-width coefficientwise linear maps]
		\label{lem:coeff-linear-map}
		Let $S$ be a fixed finite ring.  Fix constants $a,b$.  Suppose that for each coefficient position $0\le t<N$ we have input registers $x_{1,t},\ldots,x_{a,t}\in S$ and output registers $y_{1,t},\ldots,y_{b,t}\in S$.  For fixed constants $\lambda_{\ell j}\in S$, the reversible update
		\begin{equation}
			y_{\ell,t}
			\longleftarrow
			y_{\ell,t}+\sum_{j=1}^{a}\lambda_{\ell j}x_{j,t},
			\qquad
			1\le \ell\le b,
			\quad
			0\le t<N,
			\label{eq:coeffwise-linear-update}
		\end{equation}
		with all $x_{j,t}$ unchanged, can be implemented in $\delta$D space $\cO_\delta(N)$ and depth $\cO_\delta(N^{1/\delta})$.  The same bound holds if each output position is shifted by one of a constant number of fixed offsets, for example $y_{\ell,t+d_\ell}$ instead of $y_{\ell,t}$.
	\end{lemma}
	
	\begin{proof}
		Each element of the fixed ring $S$ is stored in $\cO_S(1)$ bits.  For a fixed $t$, Eq.~\eqref{eq:coeffwise-linear-update} is a constant-size reversible gate, namely a product of updates $(x,y)\mapsto (x,y+\lambda x)$ with hardwired $\lambda\in S$.  The operations for different $t$ are disjoint.  Put the $\cO(N)$ coefficient registers in a grid of space $\cO_\delta(N)$.  By Lemma~\ref{lem:mesh_compilation_app}, one layer of these disjoint constant-size operations costs $\cO_\delta(N^{1/\delta})$ depth.  Since $ a$ and $ b$ are constants, only a constant number of such layers are needed.
	\end{proof}
	
	The second primitive aggregates contributions from all coefficient positions into a single accumulator. 
	
	\begin{lemma}[Fixed ring-valued summations with $\delta$D implementations]
		\label{lem:fixed-ring-reduction}
		Let $S$ and $R$ be fixed finite rings.  For each $0\le t<N$, let $\mu_t:S\to R$ be a fixed map.  Given input registers $x_0,\ldots,x_{N-1}\in S$ and one accumulator register $y\in R$, the reversible update
		\begin{equation}
			y
			\longleftarrow
			y+\sum_{t=0}^{N-1}\mu_t(x_t)
			\label{eq:fixed-ring-reduction-update}
		\end{equation}
		with all input registers unchanged can be implemented cleanly in $\delta$D space $\cO_\delta(N)$ and depth $\cO_\delta(N^{1/\delta})$.
	\end{lemma}
	
	\begin{proof}
		Place the $N$ input registers in a $\delta$D grid of side length $L=\Theta_\delta(N^{1/\delta})$, padding by idle sites if necessary. At the site of $x_t$, allocate a clean scratch register $s_t\in R$ and compute
		\begin{equation}
			(x_t,s_t)
			\longmapsto
			(x_t,s_t+\mu_t(x_t)).
		\end{equation}
		This is a constant-size gate because $S$ and $R$ are fixed finite rings and $\mu_t$ is hardwired.  The input register $x_t$ is unchanged, and initially $s_t=0$, so after this step $s_t=\mu_t(x_t)$.
		
		It remains to add the scratch registers to one root scratch register.  We do this by dimension-by-dimension nearest-neighbor sweeps.  Along each line parallel to the first coordinate, apply nearest-neighbor additions from the last site toward the first, so the first site of the line contains the sum along that line.  All such lines are processed in parallel, costing $\cO(L)$ depth.  Repeat the same procedure along the second coordinate on the hyperplane with the first coordinate fixed at the first site, and continue through all $\delta$ coordinates.  After the sweeps, the root scratch register contains $\sum_t\mu_t(x_t)$.
		
		Add the root scratch register to $y$, then reverse the sweeps and the initial computation.  This restores every scratch register to $0$ and leaves exactly the update in Eq.~\eqref{eq:fixed-ring-reduction-update}.  Since addition in the fixed ring $R$ is a constant-size reversible gate, the total depth is $\cO_\delta(L)=\cO_\delta(N^{1/\delta})$, and the space is $\cO_\delta(N)$.
	\end{proof}

	\subsection{Polynomial multiplication}
	\label{subsec:polynomial-multiplication}
	
	We now implement the key multiplication primitive used later.  
	The polynomial-multiplication primitive will be used over two types of coefficient rings.  
	Finite fields are needed for arithmetic over $\bF_{2^s}$, in particular for polynomial evaluation in the $t$-wise independent function family and for finite-field multiplication.  
	Galois rings, including $\mathbb Z_4$ and constant-degree extensions of $\mathbb Z_4$, are needed for the mod-$4$ arithmetic appearing in the quadratic phase gates of the finite-field Clifford implementation.  
	We therefore prove the multiplication primitive for fixed finite fields and fixed Galois rings together.
	
	The algorithm is a standard Toom--Cook recursion.  
	One point needing care over rings is interpolation, where unit-separated evaluation points ensure that the Vandermonde matrix is invertible.
	This motivates the following definition.
	
	\begin{definition}[Unit-separated evaluation points]
		\label{def:unit-separated}
		Let $S$ be a finite commutative ring and let $r\ge1$.  A list $\alpha_1,\ldots,\alpha_m\in S$, with $m:=2r-1$, is called \emph{unit-separated} if $\alpha_i-\alpha_j\in S^\times$ for all $i\ne j$.
		
		For such a list, the Vandermonde matrix
		\begin{equation}
			V=(V_{i,j})_{1\le i\le m,\,0\le j\le m-1},
			\qquad
			V_{i,j}:=\alpha_i^j,
		\end{equation}
		is invertible over $S$, because $\det V=\prod_{1\le i<i'\le m}(\alpha_{i'}-\alpha_i)$ is a product of units.
	\end{definition}
	
	With unit-separated evaluation points, Toom--Cook multiplication has the same diameter-time implementation over these rings as over fields.
	
	\begin{lemma}[Polynomial multiplication with fixed unit-separated evaluation points]
		\label{lem:unit-separated-poly-mult}
		Fix $\delta\ge1$.  Let $S$ be a fixed finite field or a fixed Galois ring $\GR(4,a)$.  Suppose that $S$ contains a unit-separated list $\alpha_1,\ldots,\alpha_{2r-1}$ for a sufficiently large constant $r=r(\delta)$.  Then, for every $s$, the reversible map
		\begin{equation}
			\ket P\ket Q\ket C\ket0
			\longmapsto
			\ket P\ket Q\ket{C+PQ}\ket0,
			\label{eq:poly-mult-map}
		\end{equation}
		where $P,Q\in S[z]_{<s}$ and $C\in S[z]_{<2s}$, can be implemented in $\delta$D space $\cO_\delta(s)$ and depth $\cO_\delta(s^{1/\delta})$, with clean ancillas.
	\end{lemma}
	
	\begin{proof}
		We prove the claim by induction on $s$.  Set $m:=2r-1$ and $b:=\lceil s/r\rceil$.  Pad $P,Q$ with zero coefficients to length $rb$, and write
		\begin{equation}
			\begin{split}
				P(z)&=\sum_{u=0}^{r-1}P_u(z)z^{ub},\qquad P_u\in S[z]_{<b},\\
				Q(z)&=\sum_{u=0}^{r-1}Q_u(z)z^{ub},\qquad Q_u\in S[z]_{<b}.
			\end{split}
			\label{eq:tc-blocks-simple}
		\end{equation}
		Introduce a block variable $X$ and define
		\begin{equation}
			\mathcal P(X;z):=\sum_{u=0}^{r-1}P_u(z)X^u,
			\qquad
			\mathcal Q(X;z):=\sum_{u=0}^{r-1}Q_u(z)X^u.
		\end{equation}
		Their block product is
		\begin{equation}
			T(X;z):=\mathcal P(X;z)\mathcal Q(X;z)
			=\sum_{j=0}^{m-1}C_j(z)X^j,
			\qquad
			C_j\in S[z]_{<2b}.
			\label{eq:tc-block-product-simple}
		\end{equation}
		The ordinary product is obtained by substituting $X=z^b$:
		\begin{equation}
			P(z)Q(z)=\sum_{j=0}^{m-1}C_j(z)z^{jb}.
			\label{eq:tc-substitute-simple}
		\end{equation}
		
		For each evaluation point $\alpha_i$, define
		\begin{equation}
			P_i(z):=\mathcal P(\alpha_i;z),
			\qquad
			Q_i(z):=\mathcal Q(\alpha_i;z),
			\qquad
			T_i(z):=P_i(z)Q_i(z).
			\label{eq:tc-evaluations-simple}
		\end{equation}
		The evaluation maps $P_u,Q_u\mapsto P_i,Q_i$ are constant-width coefficientwise $S$-linear maps, so Lemma~\ref{lem:coeff-linear-map} implements them in depth $\cO_\delta(s^{1/\delta})$.
		
		The values $T_i$ determine the block coefficients $C_j$.  Substituting $X=\alpha_i$ in Eq.~\eqref{eq:tc-block-product-simple} gives
		\begin{equation}
			T_i(z)=\sum_{j=0}^{m-1}\alpha_i^j C_j(z),
			\qquad
			i=1,\ldots,m.
			\label{eq:tc-linear-system-simple}
		\end{equation}
		By Definition~\ref{def:unit-separated}, the corresponding Vandermonde matrix is invertible over $S$.  Write $V^{-1}=(\beta_{j,i})_{0\le j\le m-1,\,1\le i\le m}$.  Then
		\begin{equation}
			C_j(z)=\sum_{i=1}^{m}\beta_{j,i}T_i(z).
			\label{eq:tc-interpolation-simple}
		\end{equation}
		The constants $\beta_{j,i}$ are hardwired, so interpolation is again a constant-width coefficientwise $S$-linear map.
		
		We implement multiplication reversibly in batches to keep the workspace linear.  Let $h:=\lfloor r/4\rfloor$ and $B_0:=\lceil m/h\rceil$, and partition $\{1,\ldots,m\}$ into batches $J_1,\ldots,J_{B_0}$, each of size at most $h$.  For one batch $J$, perform the compute-add-uncompute sequence
		\begin{equation}
			\begin{split}
				&\text{compute }P_i,Q_i\text{ for }i\in J,\\
				&\text{recursively compute }T_i=P_iQ_i\in S[z]_{<2b}\text{ for }i\in J,\\
				&\widetilde C
				\longleftarrow
				\widetilde C+
				\sum_{j=0}^{m-1}z^{jb}\sum_{i\in J}\beta_{j,i}T_i(z),\\
				&\text{run the recursive products backward to erase all }T_i,\\
				&\text{uncompute }P_i,Q_i\text{ for }i\in J.
			\end{split}
			\label{eq:tc-batch-add-simple}
		\end{equation}
		Here $\widetilde C$ is the target register padded to length $<2rb$.  After all batches, the target has been updated by
		\begin{equation}
			\sum_{j=0}^{m-1}z^{jb}\sum_{i=1}^{m}\beta_{j,i}T_i(z)
			=\sum_{j=0}^{m-1}z^{jb}C_j(z)
			=P(z)Q(z),
		\end{equation}
		and all scratch registers have been returned to $0$.
		
		Let $V_s$ and $T_s$ be the required space and depth.  At the parent level, the inputs, padded target, evaluations for one batch, and interpolation buffers use $\cO(s)$ ring elements.  A batch contains at most $h$ recursive products of length $b=\lceil s/r\rceil$.  Thus
		\begin{equation}
			V_s\le c_0s+hV_{\lceil s/r\rceil}.
			\label{eq:tc-space-simple}
		\end{equation}
		Since $h/r\le1/4$, the induction $V_s\le Cs$ closes for $C$ large enough.
		
		For depth, evaluation, interpolation, and their inverses cost $\cO_\delta(s^{1/\delta})$ per batch by Lemma~\ref{lem:coeff-linear-map}.  The $h$ recursive products in one batch run in parallel and are used once forward and once backward.  Therefore
		\begin{equation}
			T_s\le 2B_0T_{\lceil s/r\rceil}+c_1s^{1/\delta}.
			\label{eq:tc-depth-simple}
		\end{equation}
		Choose the constant $r=r(\delta)$ large enough that $2B_0r^{-1/\delta}<1/2$.  Then the induction $T_s\le C's^{1/\delta}$ follows from Eq.~\eqref{eq:tc-depth-simple}.  This proves the stated space and depth bounds.
	\end{proof}
	
	The preceding lemma assumes that there are enough unit-separated evaluation points.  
	For any fixed finite field or fixed Galois ring, this is obtained by passing to a constant-degree extension and then copying the result back to the original coefficient ring.
	
	\begin{corollary}[Polynomial multiplication over any fixed finite field or fixed Galois ring]
		\label{cor:base-poly-mult}
		Let $A$ be any fixed finite field or fixed Galois ring $\GR(4,a)$.  Then the clean reversible polynomial multiplication map
		\begin{equation}
			\ket P\ket Q\ket C\ket0
			\longmapsto
			\ket P\ket Q\ket{C+PQ}\ket0,
			\qquad
			P,Q\in A[z]_{<s},\quad C\in A[z]_{<2s},
			\label{eq:base-poly-mult-map}
		\end{equation}
		can be implemented in $\delta$D space $\cO_\delta(s)$ and depth $\cO_\delta(s^{1/\delta})$.
	\end{corollary}
	
	\begin{proof}
		Let $r=r(\delta)$ be the constant used in Lemma~\ref{lem:unit-separated-poly-mult}, and set $M:=2r-1$.  Lemma~\ref{lem:unit-separated-poly-mult} applies to coefficient rings containing $M$ unit-separated evaluation points.  If $A$ itself is too small, we perform the Toom--Cook computation in a fixed constant-degree extension $S$ and then copy the answer back to the original $A$-target.
		
		If $A$ is a finite field, choose a finite field extension $S/A$ with $|S|\ge M$.  Since $M$ depends only on $\delta$, the extension degree is $\cO_\delta(1)$.  Any $M$ distinct elements of $S$ are unit-separated.
		
		If $A=\GR(4,a)$, choose $f=\cO_\delta(1)$ such that $2^{af}\ge M$, and set $S:=\GR(4,af)$.  An element of $S$ is a unit whenever its reduction modulo $2$ is nonzero.  Choose $M$ elements $\alpha_1,\ldots,\alpha_M\in S$ whose reductions modulo $2$ are distinct in $\bF_{2^{af}}$.  Then $\alpha_i-\alpha_j\in S^\times$ for $i\ne j$.
		
		In both cases, an element of $S$ is represented by a constant number of elements of $A$.  Embed $P,Q\in A[z]_{<s}$ coefficientwise into $S[z]$ and apply Lemma~\ref{lem:unit-separated-poly-mult} over $S$ to compute a clean scratch product $PQ\in S[z]$.  Because the inputs came from $A$, the product has only its $A$-coordinate nonzero.  Add that coordinate to the true $A$-target $C$, then run the $S$-valued multiplication backward.  The extension degree is constant, so the space and depth remain $\cO_\delta(s)$ and $\cO_\delta(s^{1/\delta})$.
	\end{proof}
	
	\subsection{Finite-field multiplication}
	\label{subsec:finite-field-multiplication}
	Finite-field multiplication is ordinary polynomial multiplication followed by reduction modulo the fixed irreducible polynomial.  The following reciprocal identity allows us to perform that reduction with only a constant number of polynomial multiplications.
	
	Let
	\begin{equation}
		K=\bF_{2^s}\cong \bF_2[z]/(p(z)),
		\qquad
		p(z)=z^s+p_{<s}(z),
	\end{equation}
	where $p$ is a fixed irreducible monic polynomial of degree $s$.  An element of $K$ is represented by a binary polynomial of degree $<s$.
	
	For a polynomial $F(z)$ of degree $<L$, define its length-$L$ reversal by
	\begin{equation}
		\operatorname{rev}_L(F)(z):=z^{L-1}F(z^{-1}).
	\end{equation}
	Thus $\operatorname{rev}_L$ reverses the list of $L$ coefficients, padding with zeros if needed.  Also write $\operatorname{low}_s(F)$ for the coefficients of degrees $0,\ldots,s-1$.
	
	\begin{lemma}[Reciprocal division modulo a fixed polynomial]
		\label{lem:reciprocal-reduction}
		Let $P(z)\in\bF_2[z]$ have degree $<2s$.  Let $Q(z)=\lfloor P(z)/p(z)\rfloor$ be the quotient in ordinary polynomial division, so $P=Qp+R$ and $\deg R<s$.  Define
		\begin{equation}
			p^*(z):=\operatorname{rev}_{s+1}(p)(z)=z^s p(z^{-1}).
		\end{equation}
		Since $p$ is monic, $p^*(0)=1$, so $p^*$ has an inverse modulo $z^s$.  Let
		\begin{equation}
			u(z)\equiv (p^*(z))^{-1}\pmod{z^s},
			\qquad
			\deg u<s.
		\end{equation}
		Then
		\begin{equation}
			Q^*(z):=\operatorname{rev}_s(Q)(z)
			=\operatorname{low}_s\left(\operatorname{rev}_{2s}(P)(z)u(z)\right).
			\label{eq:reciprocal-formula}
		\end{equation}
		Moreover,
		\begin{equation}
			R=\operatorname{low}_s\left(P+Qp_{<s}\right).
			\label{eq:remainder-low-coeffs}
		\end{equation}
	\end{lemma}
	
	\begin{proof}
		Write $P=Qp+R$ with $\deg Q<s$ and $\deg R<s$.  Reverse this identity using length $2s$:
		\begin{equation}
			\operatorname{rev}_{2s}(P)
			=\operatorname{rev}_s(Q)\operatorname{rev}_{s+1}(p)+z^s\operatorname{rev}_s(R).
		\end{equation}
		Reducing modulo $z^s$ gives $\operatorname{rev}_{2s}(P)\equiv Q^*p^*\pmod{z^s}$.  Multiplying by $u=(p^*)^{-1}\bmod z^s$ proves Eq.~\eqref{eq:reciprocal-formula}.  Finally, $R=P+Qp=P+Q(z^s+p_{<s})$.  Taking the low $s$ coefficients removes $Qz^s$ and yields Eq.~\eqref{eq:remainder-low-coeffs}.
	\end{proof}
	
	Combining reciprocal reduction with polynomial multiplication gives clean multiplication in $K=\bF_{2^s}$.
	\begin{corollary}[Clean finite-field multiplication in diameter time]
		\label{cor:field-mult}
		The reversible finite-field multiplication map
		\begin{equation}
			\ket a\ket b\ket c\ket0
			\longmapsto
			\ket a\ket b\ket{c+ab}\ket0
			\label{eq:field-mult-map}
		\end{equation}
		over $\bF_{2^s}$ has a clean $\delta$D implementation using space $\cO_\delta(s)$ and depth $\cO_\delta(s^{1/\delta})$.  The same bound holds for multiplication by a fixed field element.
	\end{corollary}
	
	\begin{proof}
		Represent $a,b$ by binary polynomials $A,B\in\bF_2[z]_{<s}$.  First compute the ordinary product $P=AB$, of degree $<2s$, into scratch using Corollary~\ref{cor:base-poly-mult} over $\bF_2$.  Then form $H=\operatorname{low}_s(\operatorname{rev}_{2s}(P))$ by wire reversal and coefficient selection.  Compute $Q^*=\operatorname{low}_s(Hu)$, where $u=(p^*)^{-1}\bmod z^s$ is fixed, and reverse $Q^*$ to obtain $Q$.  Next compute $T=Qp_{<s}$, where $p_{<s}$ is fixed.  By Lemma~\ref{lem:reciprocal-reduction}, adding $\operatorname{low}_s(P+T)$ to the target adds $AB\bmod p$.  Finally, reverse the computations of $T,Q,Q^*,H,P$.
		
		There are only a constant number of polynomial multiplications of length $\cO(s)$, plus reversals, shifts, and XORs.  By Corollary~\ref{cor:base-poly-mult}, the space is $\cO_\delta(s)$ and the depth is $\cO_\delta(s^{1/\delta})$.  If one input is a fixed field element, the same circuit is used with that input hardwired.
	\end{proof}
	
	\section{$t$-wise independent functions with $\delta$D implementations}
	\label{sec:t-wise-independent-functions}
	
	The LRFC construction uses both binary phase functions and vector-valued shuffle functions (see \eqref{eq:t-wise-phase}, \eqref{eq:t-wise-shuffle}).  Both are obtained from the same polynomial family over a finite field.
	
	\begin{lemma}[Finite-field polynomial $t$-wise independence]
		\label{lem:twise}
		Let $K=\bF_{2^s}$ and $t\ge1$.  Choose coefficients $a_0,\ldots,a_{t-1}\in K$ independently and uniformly, and define
		\begin{equation}
			P_a(x):=\sum_{j=0}^{t-1}a_jx^j.
		\end{equation}
		Then $x\mapsto P_a(x)$ is a $t$-wise independent function $K\to K$.  If $\ell:K\to\bF_2$ is a nonzero linear functional, then $x\mapsto \ell(P_a(x))$ is a $t$-wise independent binary function.
	\end{lemma}
	
	\begin{proof}
		For distinct $x_1,\ldots,x_r\in K$ with $r\le t$, the evaluation map from coefficients to values has an $r\times t$ Vandermonde matrix.  Its first $r$ columns have determinant $\prod_{i<j}(x_j-x_i)\ne0$, so the map has rank $r$.  Hence, each $r$-tuple of values has exactly $|K|^{t-r}$ preimages, proving independent uniformity.  Applying a nonzero linear functional to independent uniform field elements gives independent uniform bits.
	\end{proof}
	
	We next give a clean evaluation circuit.  The input $x$, the output $y$, and all workspace registers below are field registers, each containing $s$ qubits.
	
	\begin{theorem}[Diameter-time evaluation of polynomial functions]
		\label{thm:hash-eval}
		Let $t=2k$ and $K=\bF_{2^s}$.  For a sampled polynomial $P_a(x)=\sum_{j=0}^{t-1}a_jx^j$ over $K$, the clean reversible circuit
		\begin{equation}
			\ket x\ket y\ket0
			\longmapsto
			\ket x\ket{y+P_a(x)}\ket0
			\label{eq:hash-eval-map}
		\end{equation}
		can be implemented in $\delta$D space $\cO_\delta(ks)$ and depth $\cO_\delta\left(\log(k)(ks)^{1/\delta}\right)$.  The binary phase oracle $\ket x\mapsto (-1)^{\ell(P_a(x))}\ket x$ for any nonzero linear functional $\ell:K\to\bF_2$ has the same asymptotic resources.
	\end{theorem}
	
	\begin{proof}
		The coefficients $a_j\in K$ are classical constants fixed by the random seed.  Set $T:=t-1$.  If $T$ is not a power of two, pad the construction with dummy leaves equal to the field identity $1$.  This changes the number of leaves by at most a factor of two, so we assume below that $T$ is a power of two.
		
		First, using bitwise CNOTs, copy the computational-basis value $x$ into $T$ clean field registers:
		\begin{equation}
			\ket x\ket0^{\otimes T}
			\longmapsto
			\ket x\ket{x}^{\otimes T}.
		\end{equation}
		Call the leaves $G_1,\ldots,G_T$.  Next compute the interval products $G_I:=\prod_{q\in I}G_q$ for all dyadic intervals $I\subseteq\{1,\ldots,T\}$ using a balanced binary tree.  For an internal interval $I=I_L\sqcup I_R$, allocate a clean register $M_I$ and compute
		\begin{equation}
			M_I\longleftarrow M_I+G_{I_L}G_{I_R}.
		\end{equation}
		All products at the same tree level act on disjoint registers.
		
		A downsweep computes prefix products.  For each dyadic interval $I=[u,v]$, define $R_I:=\prod_{q<u}G_q$.  At the root, this value is $1$.  If $I=I_L\sqcup I_R$ with $I_L=[u,w]$ and $I_R=[w+1,v]$, then
		\begin{equation}
			R_{I_L}=R_I,
			\qquad
			R_{I_R}=R_I G_{I_L}.
		\end{equation}
		At a leaf $j$, this gives $R_{\{j\}}=x^{j-1}$.  Allocate a clean register $X_j$ and compute $X_j\leftarrow X_j+R_{\{j\}}G_j$, obtaining $X_j=x^j$ for $1\le j\le T$.  We set $X_0:=1$ as a known constant.
		
		For each $j=0,\ldots,t-1$, allocate a clean term register $Y_j$ and compute
		\begin{equation}
			Y_j\longleftarrow Y_j+a_jX_j,
		\end{equation}
		where $a_j$ is hardwired and $X_0=1$.  The term $Y_0=a_0$ is therefore just the preparation of a known classical field element.  Then use a balanced XOR tree over $K$ to compute
		\begin{equation}
			S=\sum_{j=0}^{t-1}Y_j=P_a(x)
		\end{equation}
		into a clean field register $S$, add $S$ to the target $y$, and reverse the computation.
		
		The construction uses $\cO(t)$ field registers, hence $\cO(ts)=\cO(ks)$ qubits.  There are $\cO(\log t)$ layers of disjoint register operations.  Each layer consists of field copies, field additions, or field multiplications on disjoint $s$-qubit registers.  In a grid of total space $\Theta(ts)$, Lemma~\ref{lem:mesh_compilation_app} and Corollary~\ref{cor:field-mult} implement one such layer in depth $\cO_\delta((ts)^{1/\delta})$.  Therefore the total depth is $\cO_\delta(\log(t)(ts)^{1/\delta})=\cO_\delta(\log(k)(ks)^{1/\delta})$.
		
		For the phase oracle, compute $P_a(x)$ into a clean field register $S$.  In the polynomial basis, the fixed nonzero linear functional $\ell:K\to\bF_2$ has the form $\ell(S)=\sum_{i=0}^{s-1}\lambda_iS_i$ for hardwired bits $\lambda_i\in\bF_2$.  By Lemma~\ref{lem:fixed-ring-reduction}, the update $b\leftarrow b+\ell(S)$ into one clean bit uses space $\cO_\delta(s)$ and depth $\cO_\delta(s^{1/\delta})$, which are dominated by the resources above.  Apply a $Z$ phase to $b$, then reverse the computation.  Thus, the phase oracle has the same asymptotic resources.
	\end{proof}
	
	\section{Exact unitary $2$-designs with $\delta$D implementations}
	\label{app:exact-2-design}
	We next construct the exact local unitary $2$-design, which serves as the first component of each LRFC block.  
	The proof is organized in three steps.  
	We first introduce trace-dual coordinates, because they identify the finite-field Fourier transform with $H^{\otimes s}$ followed by a coordinate conversion.  
	We then prove that the restricted finite-field Clifford ensemble is an exact unitary $2$-design.  
	Finally, we show that every sampled Clifford can be implemented in $\delta$D diameter time.
	
	Throughout this section, let
	\begin{equation}
		K:=\bF_{2^s}\cong \bF_2[\omega]/(p(\omega)),
	\end{equation}
	where $p(z)\in\bF_2[z]$ is a fixed irreducible polynomial of degree $s$.  Every field element is represented in the polynomial basis $x=x_0+x_1\omega+\cdots+x_{s-1}\omega^{s-1}$, where $x_j\in\bF_2$.  Thus $\cH_K:=\mathrm{span}\{\ket x:x\in K\}\cong (\mathbb C^2)^{\otimes s}$.
	
	\subsection{Trace-dual coordinates}
	\label{subsec:trace-dual-coordinates}
	
	\begin{definition}[Finite-field trace]
		\label{def:finite-field-trace}
		The trace from $K=\bF_{2^s}$ to $\bF_2$ is the map $\Tr=\Tr_{K/\bF_2}:K\to\bF_2$ defined by
		\begin{equation}
			\Tr(x):=x+x^2+x^{2^2}+\cdots+x^{2^{s-1}}.
			\label{eq:field-trace-def}
		\end{equation}
	\end{definition}
	
	The trace indeed takes values in $\bF_2$: if $T:=x+x^2+\cdots+x^{2^{s-1}}$, then $T^2=T$, since $x^{2^s}=x$ for every $x\in\bF_{2^s}$.  Hence $T\in\bF_2$.  The trace is $\bF_2$-linear and induces the nondegenerate pairing
	\begin{equation}
		K\times K\to\bF_2,
		\qquad
		(x,y)\longmapsto \Tr(xy).
		\label{eq:trace-pairing}
	\end{equation}
	Let $e_i:=\omega^i$ for $0\le i<s$.  The trace-dual basis $e_0^\vee,\ldots,e_{s-1}^\vee$ is defined by $\Tr(e_i e_j^\vee)=\delta_{ij}$.  Hence, if $x=\sum_{i=0}^{s-1}x_i e_i$ and $y=\sum_{j=0}^{s-1}\eta_j e_j^\vee$, then
	\begin{equation}
		\Tr(xy)=\sum_{i=0}^{s-1}x_i\eta_i.
		\label{eq:trace-dual-dot-product}
	\end{equation}
	
	The following lemma implements the conversion between trace-dual and polynomial coordinates, which is needed for the Fourier transform introduced later.
	
	\begin{lemma}[Trace-dual and polynomial coordinate conversion]
		\label{lem:trace-dual-conversion}
		The linear map
		\begin{equation}\label{eq:trace-dual-transform}
			L_{\vee\to\mathrm{poly}}:
			(\eta_0,\ldots,\eta_{s-1})
			\longmapsto
			\sum_{j=0}^{s-1}\eta_j e_j^\vee
		\end{equation}
		and its inverse can be implemented in-place in $\delta$D space $\cO_\delta(s)$ and depth $\cO_\delta(s^{1/\delta})$.
	\end{lemma}
	
	\begin{proof}
		Write $p(z)=p_0+p_1z+\cdots+p_{s-1}z^{s-1}+z^s$, with $p_s:=1$.  The trace-dual basis of the power basis is given by the standard formula~\cite[Theorem 5.1.12]{Mullen2013HandbookFiniteFields}: 
		\begin{equation}
			e_i^\vee
			=
			\lambda\sum_{j=i+1}^{s}p_j\omega^{j-i-1},
			\qquad
			\lambda:=(p'(\omega))^{-1}\in K.
			\label{eq:power-basis-dual-formula}
		\end{equation}
		Thus, for $y=\sum_{i=0}^{s-1}\eta_i e_i^\vee$,
		\begin{equation}
			y
			=
			\lambda
			\sum_{i=0}^{s-1}\eta_i
			\sum_{j=i+1}^{s}p_j\omega^{j-i-1}.
			\label{eq:dual-conversion-expanded}
		\end{equation}
		Define $E^{\mathrm{rev}}(z):=\sum_{i=0}^{s-1}\eta_i z^{s-1-i}$.  The coefficient of $z^{s+t}$ in $E^{\mathrm{rev}}(z)p(z)$ is
		\begin{equation}
			\sum_{i=0}^{s-1-t}\eta_i p_{i+t+1},
			\qquad
			0\le t<s,
			\label{eq:convolution_coefficients}
		\end{equation}
		which is exactly the coefficient of $z^t$ in the inner sum of Eq.~\eqref{eq:dual-conversion-expanded}.  Therefore, $L_{\vee\to\mathrm{poly}}$ is implemented by reversing the input list, multiplying by the fixed polynomial $p(z)$, extracting coefficients of degrees $s,\ldots,2s-1$, and multiplying the resulting field element by the fixed element $\lambda$.
		
		The multiplication by a fixed $p(z)$ is an ordinary polynomial multiplication of length $\cO(s)$, and the multiplication by $\lambda$ is fixed finite-field multiplication.  Corollaries~\ref{cor:base-poly-mult} and~\ref{cor:field-mult} give space $\cO_\delta(s)$ and depth $\cO_\delta(s^{1/\delta})$.
		
		For the inverse map, if $y=\sum_{j=0}^{s-1}\eta_j e_j^\vee$, then $\eta_i=\Tr(e_i y)=\Tr(\omega^i y)$.  Writing $y=\sum_{t=0}^{s-1}y_t\omega^t$, we get
		\begin{equation}
			\eta_i
			=
			\sum_{t=0}^{s-1}y_t\Tr(\omega^{i+t}),
			\qquad
			0\le i<s.
			\label{eq:inverse-trace-dual}
		\end{equation}
		Let $\tau_r:=\Tr(\omega^r)\in\bF_2$ for $0\le r\le 2s-2$, define $\tau(z):=\sum_{r=0}^{2s-2}\tau_rz^r$, and define $Y^{\mathrm{rev}}(z):=\sum_{t=0}^{s-1}y_tz^{s-1-t}$.  The coefficient of $z^{s-1+i}$ in $Y^{\mathrm{rev}}(z)\tau(z)$ is $\eta_i$.  Hence, the inverse conversion is also a fixed polynomial multiplication of length $\cO(s)$, followed by coefficient extraction.
		
		Finally, an out-of-place implementation of an invertible linear map and its inverse gives an in-place implementation by
		\begin{equation}
			\ket x\ket0
			\longmapsto
			\ket x\ket{Lx}
			\longmapsto
			\ket0\ket{Lx}
			\longmapsto
			\ket{Lx}\ket0,
			\label{eq:inplace-linear-map}
		\end{equation}
		where the second arrow adds $L^{-1}(Lx)=x$ into the first register.  This completes the proof.
	\end{proof}
	
	\subsection{The restricted finite-field Clifford ensemble}
	\label{subsec:restricted-clifford}
	Let $q:=|K|=2^s$.  For $a,b\in K$, define
	\begin{equation}
		X_a\ket x:=\ket{x+a},
		\qquad
		Z_b\ket x:=(-1)^{\Tr(bx)}\ket x.
	\end{equation}
	These satisfy $Z_bX_a=(-1)^{\Tr(ab)}X_aZ_b$.
	
	For $v=(a,b)\in K^2$, define the Hermitian Weyl operator
	\begin{equation}
		D_v:=i^{\Tr(ab)}X_aZ_b.
		\label{eq:self-adjoint-weyl}
	\end{equation}
	Then $D_v^\dagger=D_v$ and $D_v^2=I$.  The operators $\{D_v:v\in K^2\}$ form an orthogonal basis of $\operatorname{End}(\cH_K)$:
	\begin{equation}
		\tr(D_vD_w)=q\delta_{v,w}.
		\label{eq:weyl-orthogonality}
	\end{equation}
	Here $D_0=I$, and all $D_v$ with $v\ne0$ are traceless.
	
	Define the symplectic form
	\begin{equation}
		[u,v]:=\Tr(ad+bc),
		\qquad
		u=(a,b),\quad v=(c,d).
	\end{equation}
	Conjugation by $D_u$ acts diagonally on the Weyl basis:
	\begin{equation}
		D_uD_vD_u=(-1)^{[u,v]}D_v.
		\label{eq:weyl-character-action}
	\end{equation}
	
	Let $\operatorname{SL}(2,K)$ act on $K^2$ by left multiplication on column vectors.  Since the characteristic is two,
	\begin{equation}
		\operatorname{SL}(2,K)
		=
		\left\{
		\begin{pmatrix}
			a&b\\
			c&d
		\end{pmatrix}
		:
		ad+bc=1
		\right\}.
	\end{equation}
	For each $M\in\operatorname{SL}(2,K)$, fix one unitary lift $U_M$ satisfying
	\begin{equation}
		U_MD_vU_M^\dagger=\sigma(M,v)D_{Mv},
		\qquad
		\sigma(M,v)\in\{\pm1\}.
		\label{eq:symplectic-lift-action}
	\end{equation}
	The required lifts are implemented explicitly in Lemma~\ref{lem:finite-dimensional-clifford-generators} below.

	We now define the restricted finite-field Clifford ensemble and prove its second-moment property.  
	\begin{definition}[Restricted finite-field Clifford ensemble]
		\label{def:restricted-finite-field-clifford}
		The restricted finite-field Clifford ensemble is
		\begin{equation}
			\mathcal C_{\mathrm{res}}(K)
			:=
			\{\,U_MD_u:M\in\operatorname{SL}(2,K),\ u\in K^2\,\},
		\end{equation}
		sampled by choosing $M$ uniformly from $\operatorname{SL}(2,K)$ and $u$ uniformly from $K^2$.
	\end{definition}
	
	\begin{lemma}[Restricted finite-field Clifford ensemble is an exact unitary $2$-design]
		\label{lem:restricted-clifford-direct-2-design}
		The ensemble $\mathcal C_{\mathrm{res}}(K)$ is an exact unitary $2$-design on $\cH_K$.
	\end{lemma}
	
	\begin{proof}
		Let
		\begin{equation}
			\cT_{\mathrm{res}}^{(2)}(Y)
			:=
			\mathbb E_{M,u}
			(U_MD_u)^{\otimes2}
			Y
			(D_uU_M^\dagger)^{\otimes2}
		\end{equation}
		be the second-moment twirling channel.  Let $F$ denote the swap operator on $\cH_K\otimes\cH_K$.  For the Hermitian orthogonal basis $\{D_v:v\in K^2\}$,
		\begin{equation}
			F=\frac1q\sum_{v\in K^2}D_v\otimes D_v.
			\label{eq:swap-weyl-expansion}
		\end{equation}
		Every $Y\in\operatorname{End}(\cH_K^{\otimes2})$ has a unique expansion
		\begin{equation}
			Y
			=
			\frac1{q^2}
			\sum_{v,w\in K^2}
			y_{v,w}D_v\otimes D_w,
			\qquad
			y_{v,w}:=\tr\left[(D_v\otimes D_w)Y\right].
			\label{eq:Y-weyl-expansion}
		\end{equation}
		
		First average over the displacement $u$.  By Eq.~\eqref{eq:weyl-character-action},
		\begin{equation}
			D_u^{\otimes2}(D_v\otimes D_w)D_u^{\otimes2}
			=
			(-1)^{[u,v]+[u,w]}D_v\otimes D_w.
		\end{equation}
		The average over $u\in K^2$ is zero unless $v=w$.  Therefore
		\begin{equation}
			\mathcal P^{(2)}(Y)
			:=
			\mathbb E_uD_u^{\otimes2}YD_u^{\otimes2}
			=
			\frac1{q^2}\sum_{v\in K^2}y_{v,v}D_v\otimes D_v.
			\label{eq:pauli-second-twirl}
		\end{equation}
		
		The group $\operatorname{SL}(2,K)$ is transitive on $K^2\setminus\{0\}$.  Thus, for every nonzero $v$,
		\begin{equation}
			\mathbb E_{M\in\operatorname{SL}(2,K)}D_{Mv}\otimes D_{Mv}
			=
			\frac1{q^2-1}\sum_{w\in K^2\setminus\{0\}}D_w\otimes D_w.
			\label{eq:symplectic-average-nonzero}
		\end{equation}
		Using Eq.~\eqref{eq:symplectic-lift-action}, the signs cancel in the tensor square, so combining Eqs.~\eqref{eq:pauli-second-twirl} and~\eqref{eq:symplectic-average-nonzero} gives
		\begin{equation}
			\cT_{\mathrm{res}}^{(2)}(Y)
			=
			\frac{y_{0,0}}{q^2}I\otimes I
			+
			\frac{\sum_{v\ne0}y_{v,v}}{q^2(q^2-1)}
			\sum_{w\ne0}D_w\otimes D_w.
			\label{eq:restricted-twirl-diagonal}
		\end{equation}
		
		Now $y_{0,0}=\tr(Y)$, Eq.~\eqref{eq:swap-weyl-expansion} implies $\tr(FY)=q^{-1}\sum_{v\in K^2}y_{v,v}$, and $\sum_{w\ne0}D_w\otimes D_w=qF-I\otimes I$.  Substituting these identities into Eq.~\eqref{eq:restricted-twirl-diagonal} yields
		\begin{equation}
			\cT_{\mathrm{res}}^{(2)}(Y)
			=
			\frac{q\tr(Y)-\tr(FY)}{q(q^2-1)}I\otimes I
			+
			\frac{q\tr(FY)-\tr(Y)}{q(q^2-1)}F.
			\label{eq:restricted-twirl-haar-form}
		\end{equation}
		This is exactly the second-order Haar twirl.  Therefore $\mathcal C_{\mathrm{res}}(K)$ is an exact unitary $2$-design.
	\end{proof}
	
	\subsection{The restricted Clifford generators with $\delta$D implementations}
	\label{subsec:finite-dimensional-clifford-generators}
	It remains to realize the unitary lifts appearing in Eq.~\eqref{eq:symplectic-lift-action}. 
	We use the standard generators of $\operatorname{SL}(2,K)$: Fourier transform, scaling, and quadratic shear, together with Weyl displacements.
	
	\begin{lemma}[Clifford generators with $\delta$D implementations]
		\label{lem:finite-dimensional-clifford-generators}
		For every $\delta\ge1$, each of the following gates has a clean $\delta$D implementation in space $\cO_\delta(s)$ and depth $\cO_\delta(s^{1/\delta})$: displacements $X_uZ_v$, the additive Fourier transform $F_K$, scalings $M_\lambda$ for $\lambda\in K^\times$, and quadratic shears $P_\gamma$ for $\gamma\in K$.  Their actions on Weyl labels realize the generators
		\begin{equation}
			w=
			\begin{pmatrix}
				0&1\\
				1&0
			\end{pmatrix},
			\qquad
			D(\lambda)=
			\begin{pmatrix}
				\lambda&0\\
				0&\lambda^{-1}
			\end{pmatrix},
			\qquad
			L(\gamma)=
			\begin{pmatrix}
				1&0\\
				\gamma&1
			\end{pmatrix}.
			\label{eq:clifford-generators-matrices}
		\end{equation}
	\end{lemma}
	
	\begin{proof}
		For displacements, $X_u$ is bitwise XOR by a classical constant in the polynomial basis, and $Z_v$ applies the phase $(-1)^{\Tr(vx)}$.  The map $x\mapsto\Tr(vx)$ is a fixed $\bF_2$-linear functional of the $s$ polynomial-basis bits of $x$.  By Lemma~\ref{lem:fixed-ring-reduction}, it can be accumulated into one clean ancilla bit in space $\cO_\delta(s)$ and depth $\cO_\delta(s^{1/\delta})$, phased, and uncomputed.
		
		The additive Fourier transform is
		\begin{equation}
			F_K\ket x
			:=
			2^{-s/2}\sum_{y\in K}(-1)^{\Tr(xy)}\ket y.
			\label{eq:finite-field-fourier-corrected}
		\end{equation}
		This can be written as
		\begin{equation}
			F_K=L_{\vee\to\mathrm{poly}}\circ H^{\otimes s},
			\label{eq:fourier-implementation-corrected}
		\end{equation}
		where $L_{\vee\to\mathrm{poly}}$ is defined in Eq.~\eqref{eq:trace-dual-transform}. Lemma~\ref{lem:trace-dual-conversion} gives the claimed implementation.  Directly from the definition,
		\begin{equation}
			F_KX_aF_K^\dagger=Z_a,
			\qquad
			F_KZ_bF_K^\dagger=X_b,
			\label{eq:fourier-pauli-action-corrected}
		\end{equation}
		so $F_K$ realizes $w$.
		
		For $\lambda\in K^\times$, define $M_\lambda\ket x:=\ket{\lambda x}$.  A clean in-place implementation is obtained from fixed multiplication by $\lambda$ and $\lambda^{-1}$, using Corollary~\ref{cor:field-mult}.  Its Pauli action is
		\begin{equation}
			M_\lambda X_a M_\lambda^\dagger=X_{\lambda a},
			\qquad
			M_\lambda Z_b M_\lambda^\dagger=Z_{\lambda^{-1}b},
			\label{eq:scaling-pauli-action-corrected}
		\end{equation}
		so it realizes $D(\lambda)$.
		
		It remains to construct the shear.  For $x=\sum_{j=0}^{s-1}x_j\omega^j$, lift each bit $x_j\in\{0,1\}$ to $\widetilde x_j\in\mathbb Z_4$ and define $\widetilde X_x(z):=\sum_{j=0}^{s-1}\widetilde x_jz^j\in\mathbb Z_4[z]$.  For $0\le r\le2s-2$, set $\tau_r:=\Tr(\gamma\omega^r)\in\bF_2\subset\mathbb Z_4$, and define
		\begin{equation}
			Q_\gamma(x)
			:=
			\sum_{r=0}^{2s-2}\tau_r[\widetilde X_x(z)^2]_r
			\pmod 4.
			\label{eq:quadratic-refinement-corrected}
		\end{equation}
		Set $P_\gamma\ket x:=i^{Q_\gamma(x)}\ket x$.
		
		For bits lifted to $\mathbb Z_4$, $\widetilde{x_j+y_j}=\widetilde x_j+\widetilde y_j + 2\widetilde x_j\widetilde y_j\pmod4$.  Therefore
		\begin{equation}
			\widetilde X_{x+y}(z)^2
			=
			\widetilde X_x(z)^2+
			\widetilde X_y(z)^2+
			2\widetilde X_x(z)\widetilde X_y(z)
			\pmod4.
		\end{equation}
		Substituting into Eq.~\eqref{eq:quadratic-refinement-corrected} gives
		\begin{equation}
			\begin{split}
				Q_\gamma(x+y)-Q_\gamma(x)-Q_\gamma(y)
				&=
				2\sum_{r=0}^{2s-2}\Tr(\gamma\omega^r)
				[\widetilde X_x(z)\widetilde X_y(z)]_r
				\pmod4\\
				&=
				2\Tr(\gamma xy)
				\pmod4.
			\end{split}
			\label{eq:quadratic-polar-corrected}
		\end{equation}
		It follows that
		\begin{equation}
			P_\gamma X_aP_\gamma^\dagger
			=
			i^{Q_\gamma(a)}X_aZ_{\gamma a},
			\qquad
			P_\gamma Z_bP_\gamma^\dagger=Z_b.
			\label{eq:shear-pauli-action-corrected}
		\end{equation}
		Thus $P_\gamma$ realizes $L(\gamma)$ on Weyl labels, up to an irrelevant phase.
		
		The implementation of $P_\gamma$ is clean and diameter-time.  Compute $\widetilde X_x(z)^2$ over $\mathbb Z_4[z]$ using Corollary~\ref{cor:base-poly-mult} with $A=\mathbb Z_4$.  Then compute the fixed linear functional in Eq.~\eqref{eq:quadratic-refinement-corrected} into a two-bit $\mathbb Z_4$ accumulator.  This is an instance of Lemma~\ref{lem:fixed-ring-reduction} with $N=2s-1$ constant-size $\mathbb Z_4$ registers, so it uses space $\cO_\delta(s)$ and depth $\cO_\delta(s^{1/\delta})$.  Apply the phase $i^{Q_\gamma(x)}$ and reverse the computation.  Hence, $P_\gamma$ has the claimed resources.
	\end{proof}

	\begin{lemma}[Constant-word implementation of symplectic lifts]
		\label{lem:constant-word-symplectic-lifts}
		Every $M\in\operatorname{SL}(2,K)$ can be implemented, up to a global phase and the signs in Eq.~\eqref{eq:symplectic-lift-action}, by a constant-length product of the generators in Lemma~\ref{lem:finite-dimensional-clifford-generators}.
	\end{lemma}
	
	\begin{proof}
		Let
		\begin{equation}
			M=
			\begin{pmatrix}
				a&b\\
				c&d
			\end{pmatrix}
			\in\operatorname{SL}(2,K),
			\qquad
			ad+bc=1.
		\end{equation}
		If $b\ne0$, then
		\begin{equation}
			M=L(d/b)D(b)wL(a/b).
			\label{eq:bruhat-b-nonzero-corrected}
		\end{equation}
		Indeed,
		\begin{equation}
			L(d/b)D(b)wL(a/b)
			=
			\begin{pmatrix}
				a&b\\
				(ad+1)/b&d
			\end{pmatrix},
		\end{equation}
		and $(ad+1)/b=c$ because $ad+bc=1$ and the characteristic is two.  If $b=0$, then $ad=1$ and
		\begin{equation}
			M=L(c/a)D(a).
			\label{eq:bruhat-b-zero-corrected}
		\end{equation}
		Thus, every symplectic part has constant word length.
	\end{proof}
	
	The previous lemma gives the symplectic lift as a constant product of the generators.  
	Combining this with the displacement gives the required exact design.
	
	\begin{theorem}[Exact unitary $2$-design with $\delta$D implementations]
		\label{thm:finite-dimensional-exact-2-design}
		For every $s$, there exists an exact unitary $2$-design on $s$ qubits whose unitaries can be implemented in a $\delta$D grid using $\cO_\delta(s)$ space, $\cO_\delta(s^{1/\delta})$ depth, and clean ancillas.
	\end{theorem}
	
	\begin{proof}
		Use the ensemble $\mathcal C_{\mathrm{res}}(K)$ from Definition~\ref{def:restricted-finite-field-clifford}.  Lemma~\ref{lem:restricted-clifford-direct-2-design} proves that it is an exact unitary $2$-design.  A unitary in $\mathcal C_{\mathrm{res}}(K)$ is a displacement followed by a symplectic lift.  By Lemmas~\ref{lem:finite-dimensional-clifford-generators} and~\ref{lem:constant-word-symplectic-lifts}, this is implemented by a constant number of clean diameter-time generator circuits, so the total space is $\cO_\delta(s)$ and the total depth is $\cO_\delta(s^{1/\delta})$.
	\end{proof}
	
	\section{Random unitaries with $\delta$D implementations}
	\label{app:random-unitaries}
	
	We now assemble the primitives with $\delta$D implementations into low-depth random unitary designs.  The construction has two steps.  First, we use the exact local $2$-design and the $2k$-wise independent functions above to implement each LRFC block directly with $\delta$D implementations.  Second, we glue the local blocks into a global approximate unitary design by a double-layer blocked circuit.  
	
	\subsection{Gluing random unitaries in double-layer blocked circuits}
	\label{subsec:double-layer-blocked-gluing}
	
	\begin{definition}[Double-layer blocked circuit]
		\label{def:double_layer_blocked_circuit}
		Let $P_1,\ldots,P_m$ be disjoint patches whose union is the full set of $n$ data qubits, with $m\ge2$.  Define two nearest-neighbor matchings
		\begin{equation}
			\mathcal M_{\mathrm{odd}}:=\{(a,a+1):a\text{ odd},\ 1\le a<m\},
			\qquad
			\mathcal M_{\mathrm{even}}:=\{(a,a+1):a\text{ even},\ 1\le a<m\}.
		\end{equation}
		For an edge $e=(a,a+1)$, let $Q_e:=P_a\sqcup P_{a+1}$.
		
		A double-layer blocked circuit is a unitary of the form
		\begin{equation}
			U_{\mathrm{dbl}}=U_{\mathrm{even}}U_{\mathrm{odd}},
			\qquad
			U_{\mathrm{odd}}:=\bigotimes_{e\in\mathcal M_{\mathrm{odd}}}U_e,
			\qquad
			U_{\mathrm{even}}:=\bigotimes_{e\in\mathcal M_{\mathrm{even}}}U_e,
			\label{eq:double_layer_blocked_circuit}
		\end{equation}
		where each $U_e$ is supported on $Q_e$.  Given local ensembles $\{\mathcal L_e\}_{e\in\mathcal M_{\mathrm{odd}}\cup\mathcal M_{\mathrm{even}}}$, the associated double-layer blocked ensemble $\cU_{\mathrm{dbl}}$ is obtained by sampling $U_e\sim\mathcal L_e$ independently for all edges $e$ and applying Eq.~\eqref{eq:double_layer_blocked_circuit}.
	\end{definition}
	
	Neighboring local blocks in opposite layers overlap on one patch: $Q_{(a,a+1)}\cap Q_{(a+1,a+2)}=P_{a+1}$.  Thus, if every patch has size at least $\xi$, then every such overlap contains at least $\xi$ qubits.
	
	\begin{fact}[Gluing small unitary designs with double-layer blocked circuits, {\cite[Theorem~6]{Schuster2024RandomUnitaries}}]
		\label{fact:relative_gluing}
		There is a universal constant $C_{\mathrm g}>0$ with the following property.  Consider the double-layer blocked ensemble from Definition~\ref{def:double_layer_blocked_circuit}.  Suppose that $|P_a|\ge\xi$ for every $a\in\{2,\ldots,m-1\}$, and suppose that, for every $e\in\mathcal M_{\mathrm{odd}}\cup\mathcal M_{\mathrm{even}}$, the local ensemble $\mathcal L_e$ is a multiplicative-$\eta$ approximate unitary $k$-design on $Q_e$.  If
		\begin{equation}
			C_{\mathrm g}m\left(\eta+k^2 2^{-\xi}\right)\le\epsilon,
			\label{eq:gluing_condition_app}
		\end{equation}
		then $\cU_{\mathrm{dbl}}$ is a multiplicative-$\epsilon$ approximate unitary $k$-design on $n$ qubits.
	\end{fact}
	
	\begin{remark}
		Fact~\ref{fact:relative_gluing} is a specialized form of
		Ref.~\cite[Theorem~6]{Schuster2024RandomUnitaries}.  In that theorem, one
		considers a general two-layer circuit of overlapping local multiplicative-error
		designs and defines an overlap graph whose vertices are the local unitaries.
		An edge is drawn whenever a first-layer block and a second-layer block overlap
		on at least \(\xi\) qubits.  In the double-layer blocked circuit above, this
		overlap graph is a path, because
		\(Q_{(a,a+1)}\cap Q_{(a+1,a+2)}=P_{a+1}\).
		
		The theorem in Ref.~\cite{Schuster2024RandomUnitaries} is stated with the sufficient choices $\eta\le\epsilon/n$ and $\xi\gtrsim\log(nk^2/\epsilon)$.  The more flexible criterion Eq.~\eqref{eq:gluing_condition_app} is the same proof with the local errors and gluing losses kept explicit: there are $\cO(m)$ local blocks, and each two-block gluing step contributes $\cO(k^2 2^{-\xi})$.
	\end{remark}
	
	\subsection{LRFC designs with $\delta$D implementations}
	\label{subsec:local-finite-dimensional-lrfc-blocks}
	
	We next build the local ensembles used in the double-layer circuit.  Let $\Lambda$ be an even-size block of $\ell=2h$ qubits, with a bipartition $\Lambda=\Lambda_L\sqcup \Lambda_R$ and $|\Lambda_L|=|\Lambda_R|=h$.  A single LRFC circuit on $\Lambda$ has the form
	\begin{equation}
		U_{\mathrm{LRFC}}=S_LS_RF C.
	\end{equation}
	Here $C$ is sampled from an exact unitary $2$-design on $\Lambda$, $F$ is a diagonal random phase unitary, and $S_L,S_R$ are conditional shuffle gates:
	\begin{equation}\label{eq:t-wise-phase}
		F\ket x=(-1)^{f(x)}\ket x,
	\end{equation}
	\begin{equation}\label{eq:t-wise-shuffle}
		S_L\ket{x_L,x_R}
		=
		\ket{x_L+\sigma_L(x_R),x_R},
		\qquad
		S_R\ket{x_L,x_R}
		=
		\ket{x_L,x_R+\sigma_R(x_L)}.
	\end{equation}
	The function $f:\bF_2^\ell\to\bF_2$ is chosen from a $2k$-wise independent binary family, and $\sigma_L,\sigma_R:\bF_2^h\to\bF_2^h$ are chosen from $2k$-wise independent vector-valued families.  These families are realized by polynomial evaluation over $\bF_{2^\ell}$ and $\bF_{2^h}$ as in Lemma~\ref{lem:twise}.
	
	For $p\ge1$, define the amplified local ensemble
	\begin{equation}
		\mathcal L_\Lambda^{(p)}:=(\mathrm{LRFC}_\Lambda)^p
	\end{equation}
	as the product of $p$ independently sampled LRFC circuits on $\Lambda$.
	
	\begin{fact}[Amplified LRFC designs, {\cite[Lemma~14]{Cui2025UnitaryDesignsOptimal}}]
		\label{fact:local_amplified_lrfc_design}
		Let $\Lambda$ be an even-size block of $\ell$ qubits, and let $p\ge8k+1$.  If $k\le2^{\ell/8}$, then $\mathcal L_\Lambda^{(p)}$ is a multiplicative-$\eta_\ell$ approximate unitary $k$-design on $\Lambda$, with $\eta_\ell\le2k^2 2^{-\ell/2}$.
	\end{fact}
	
	We now combine the exact design and finite-field multiplication ingredients to implement the amplified LRFC designs in $\delta$D architectures. 
	\begin{lemma}[Amplified LRFC blocks with $\delta$D implementations]
		\label{lem:finite-dimensional-lrfc-local}
		Fix $\delta\ge1$.  Let $\Lambda$ be an even-size block of $\ell$ qubits and let $p=8k+1$.  The ensemble $\mathcal L_\Lambda^{(p)}$ can be implemented in a $\delta$D grid using $\cO_\delta(k\ell)$ space and depth
		\begin{equation}
			\cO_\delta\left(k\log(k)(k\ell)^{1/\delta}\right),
			\label{eq:finite-dimensional-lrfc-local-depth}
		\end{equation}
		with clean ancillas.
	\end{lemma}
	
	\begin{proof}
		One LRFC layer has four components.  The exact $2$-design unitary $C$ is implemented by Theorem~\ref{thm:finite-dimensional-exact-2-design} in space $\cO_\delta(\ell)$ and depth $\cO_\delta(\ell^{1/\delta})$.  The binary phase $F$ is implemented by Theorem~\ref{thm:hash-eval} over $\bF_{2^\ell}$, with space $\cO_\delta(k\ell)$ and depth $\cO_\delta(\log(k)(k\ell)^{1/\delta})$.  The shuffles $S_L$ and $S_R$ are vector-valued polynomial evaluations over $\bF_{2^h}$, again by Theorem~\ref{thm:hash-eval}, followed by field addition into the target half.  Their space is $\cO_\delta(kh)$ and their depth is $\cO_\delta(\log(k)(kh)^{1/\delta})$.
		
		Thus one LRFC layer is implemented in space $\cO_\delta(k\ell)$ and depth $\cO_\delta(\log(k)(k\ell)^{1/\delta})$.  The amplified ensemble uses $p=8k+1$ independent layers sequentially and reuses the same workspace, giving Eq.~\eqref{eq:finite-dimensional-lrfc-local-depth}.
	\end{proof}
	
	The amplified LRFC construction above is stated for even-sized blocks. The next lemma turns this construction into a local design for every support size $\ell$.
	
	\begin{lemma}[Local block designs for arbitrary block sizes]
		\label{lem:arbitrary_size_local_design}
		Let $\Lambda$ be a block of $\ell$ qubits with $\ell\ge2\xi$, and suppose $k\le2^{\xi/4}$.  Then there exists an ensemble $\widehat{\mathcal L}_\Lambda$ on $\Lambda$ which is a multiplicative-$\eta$ approximate unitary $k$-design, with $\eta\le Ck^2 2^{-\xi}$, where $C>0$ is a universal constant.  It has a $\delta$D nearest-neighbor implementation using $\cO_\delta(k\ell)$ space and depth $\cO_\delta\left(k\log(k)(k\ell)^{1/\delta}\right)$, with clean ancillas.
	\end{lemma}
	
	\begin{proof}
		If $\ell$ is even, apply Fact~\ref{fact:local_amplified_lrfc_design} with $p=8k+1$.  Since $\ell\ge2\xi$, the assumption $k\le2^{\xi/4}$ implies $k\le2^{\ell/8}$, and hence $\eta_\ell\le2k^2 2^{-\ell/2}\le2k^2 2^{-\xi}$.
		The resource bounds follow from Lemma~\ref{lem:finite-dimensional-lrfc-local}.
		
		If $\ell$ is odd, write $\ell=2h+1$.  Then $h\ge\xi$.  Choose two even-size subblocks $\Lambda_-$ and $\Lambda_+$, each of size $2h$, obtained by deleting one endpoint qubit from $\Lambda$ in two different ways.  Their overlap has size $2h-1$.  Apply the even-size construction independently on $\Lambda_-$ and $\Lambda_+$ and compose the two sampled unitaries.  Each subblock ensemble has error at most $2k^2 2^{-h}\le2k^2 2^{-\xi}$.  The two-block gluing lemma underlying Fact~\ref{fact:relative_gluing} gives total multiplicative error
		\begin{equation}
			\cO\left(k^2 2^{-h}+k^2 2^{-(2h-1)}\right)
			\le
			Ck^2 2^{-\xi}.
		\end{equation}
		The two subblock implementations are sequential and use regions of size $\cO_\delta(k\ell)$, so the space and depth bounds increase only by a constant factor.
	\end{proof}
	
	\subsection{Global multiplicative-error designs with $\delta$D implementations}
	\label{subsec:finite_dimensional_design_theorem}
	
	We now combine the local $\delta$D LRFC construction with the double-layer gluing lemma.  
	
	\begin{theorem}[Unitary designs with $\delta$D low-depth implementations, formal version of Theorem~\ref{thm:finite-dimensional-design}]
		\label{thm:finite_dimensional_design_app}
		Fix $\delta\ge1$.  There are constants $C_\delta,C_*>0$, with $C_\delta$ depending only on $\delta$ and $C_*$ universal, such that the following holds.  Let $n$ be the number of data qubits, let $k\ge2$, and let $0<\epsilon\le1$.  Suppose there is an integer $\xi$ satisfying
		\begin{equation}
			2\xi\le n,
			\qquad
			k\le2^{\xi/4},
			\qquad
			C_*\frac n\xi k^2 2^{-\xi}\le\epsilon.
			\label{eq:finite_dim_parameter_condition}
		\end{equation}
		Then there exists an $n$-qubit random unitary ensemble $\cU$ which is a multiplicative-$\epsilon$ approximate unitary $k$-design and is implementable by a $\delta$D nearest-neighbor circuit with clean ancillas.  The required depth is $\cO_{\delta}(k\log(k)(k\xi)^{1/\delta})$ and the number of clean ancillas is $\cO_{\delta}(nk)$.
	\end{theorem}
	
	\begin{proof}
		Choose $m:=\lfloor n/\xi\rfloor$.  Since $2\xi\le n$, we have $m\ge2$.  Write $n=m\xi+r$ with $0\le r<\xi$, and define patch sizes
		\begin{equation}
			s_a:=\xi\quad(1\le a<m),
			\qquad
			s_m:=\xi+r.
		\end{equation}
		Then $\xi\le s_a\le2\xi$, $\sum_{a=1}^m s_a=n$, and $m\le n/\xi$.
		
		We place these patches in a $\delta$D grid with clean ancillas.  Let $Q_{\mathrm{cell}}:=C_0k\xi$ for a sufficiently large constant $C_0$, and set $B:=\lceil Q_{\mathrm{cell}}^{1/\delta}\rceil$.  For $a=1,\ldots,m$, define the cell
		\begin{equation}
			\Gamma_a
			:=
			\{(x_1,\ldots,x_\delta):(a-1)B<x_1\le aB,\ 1\le x_j\le B\text{ for }2\le j\le\delta\}.
		\end{equation}
		The cells are disjoint boxes arranged consecutively along the first coordinate, and consecutive cells share a $(\delta-1)$-dimensional face.
		
		Inside each cell $\Gamma_a$, choose a connected subset $P_a\subset\Gamma_a$ of size $s_a$ and place the data qubits of patch $P_a$ there.  All remaining sites in $\Gamma_a$ are clean ancillas.  For an edge $e=(a,a+1)$, let $Q_e:=P_a\sqcup P_{a+1}$.  Then
		\begin{equation}
			2\xi\le |Q_e|\le4\xi.
		\end{equation}
		
		Apply Lemma~\ref{lem:arbitrary_size_local_design} to each support $Q_e$.  We obtain a local ensemble $\mathcal L_e$ which is a multiplicative-$\eta$ approximate unitary $k$-design on $Q_e$, with $\eta\le Ck^2 2^{-\xi}$.  The same lemma gives a nearest-neighbor implementation in the physical region $\Gamma_a\cup\Gamma_{a+1}$ using space $\cO_\delta(k\xi)$ and depth $\cO_\delta(k\log(k)(k\xi)^{1/\delta})$.  
		
		Now form the double-layer blocked ensemble $\cU_{\mathrm{dbl}}$ from Definition~\ref{def:double_layer_blocked_circuit}.  Since every patch has size at least $\xi$, neighboring blocks in opposite layers overlap on at least $\xi$ qubits.  Fact~\ref{fact:relative_gluing} gives that $\cU_{\mathrm{dbl}}$ is a multiplicative-$\epsilon'$ approximate unitary $k$-design with
		\begin{equation}
			\epsilon'
			\le
			C_{\mathrm g}m\left(\eta+k^2 2^{-\xi}\right)
			\le
			C_*\frac n\xi k^2 2^{-\xi} \le \epsilon.
		\end{equation}
		
		The qubits used by distinct edges in $\mathcal M_{\mathrm{odd}}$ are disjoint, so all odd-layer local blocks are implemented in parallel.  The same is true for the even layer.  Thus, the two-layer circuit increases the local depth only by a constant factor.  Finally, the total number of physical sites is
		\begin{equation}
			mB^\delta = \cO_\delta\left(\frac n\xi\cdot k\xi\right) = \cO_\delta(nk),
		\end{equation}
		so the number of clean ancillas is at most $C_\delta nk$ after increasing $C_\delta$.
	\end{proof}

\end{document}